\documentclass[superscriptaddress,prx,nofootinbib,notitlepage,onecolumn,longbibliography,10pt]{revtex4-1}

\usepackage{amssymb,amsmath,amsthm,amsfonts,bbm}
\usepackage{url}
\usepackage{graphicx}
\usepackage{float}
\usepackage{textcomp}
\usepackage{gensymb}
\usepackage{array}
\usepackage{colortbl}
\usepackage{mathtools}
\usepackage{enumerate}

\usepackage{xcolor}

\makeatletter
\def\l@subsection#1#2{}
\def\l@subsubsection#1#2{}
\makeatother

\usepackage{setspace}

\usepackage[hyperindex,breaklinks,colorlinks]{hyperref}
\hypersetup{
    colorlinks=true,
    linkcolor=myred,
    filecolor=magenta,      
    urlcolor=black,
    citecolor=blue
}

\newtheorem{theorem}{Theorem}
\newtheorem{corollary}{Corollary}
\newtheorem{fact}{Fact}
\newtheorem{lemma}{Lemma}

\newcommand{\orth}{\mathrm{O}(2n)}
\newcommand{\Hn}{\mathcal{H}_n}
\newcommand{\LHn}{\mathcal{L}(\mathcal{H}_n)}
\newcommand{\LLHn}{\mathcal{L}(\mathcal{L}(\mathcal{H}_n))}
\newcommand{\mU}{\mathcal{U}}
\newcommand{\gr}[1]{\left\langle #1 \right\rangle}

\definecolor{mypurple}{RGB}{164,64,214}
\definecolor{mypurple2}{RGB}{170,0,255}
\definecolor{mycyan}{RGB}{0, 191, 255}
\definecolor{myred}{RGB}{255, 0, 85}
\definecolor{mypink}{RGB}{255, 0, 213}

\definecolor{mymaroon}{RGB}{128, 0, 0}

\usepackage[english]{babel}
\makeatletter
\let\ORIbbl@fixname\bbl@fixname
\def\bbl@fixname#1{%
  \@ifundefined{languagealias@\expandafter\string#1}
    {\ORIbbl@fixname#1}
    {\edef\languagename{\@nameuse{languagealias@#1}}}%
}
\newcommand{\definelanguagealias}[2]{%
  \@namedef{languagealias@#1}{#2}%
}
\makeatother
\definelanguagealias{en}{english}

\allowdisplaybreaks

\usepackage{soul}
\usepackage[all]{hypcap}

\usepackage{newfloat}
\DeclareFloatingEnvironment[fileext=frm,placement={!ht},name=Algorithm]{myalgorithmfloat}
\usepackage[framemethod=TikZ]{mdframed}

\makeatletter
\DeclareFontFamily{OMX}{MnSymbolE}{}
\DeclareSymbolFont{MnLargeSymbols}{OMX}{MnSymbolE}{m}{n}
\SetSymbolFont{MnLargeSymbols}{bold}{OMX}{MnSymbolE}{b}{n}
\DeclareFontShape{OMX}{MnSymbolE}{m}{n}{
    <-6>  MnSymbolE5
   <6-7>  MnSymbolE6
   <7-8>  MnSymbolE7
   <8-9>  MnSymbolE8
   <9-10> MnSymbolE9
  <10-12> MnSymbolE10
  <12->   MnSymbolE12
}{}
\DeclareFontShape{OMX}{MnSymbolE}{b}{n}{
    <-6>  MnSymbolE-Bold5
   <6-7>  MnSymbolE-Bold6
   <7-8>  MnSymbolE-Bold7
   <8-9>  MnSymbolE-Bold8
   <9-10> MnSymbolE-Bold9
  <10-12> MnSymbolE-Bold10
  <12->   MnSymbolE-Bold12
}{}

\let\llangle\@undefined
\let\rrangle\@undefined
\DeclareMathDelimiter{\llangle}{\mathopen}%
                     {MnLargeSymbols}{'164}{MnLargeSymbols}{'164}
\DeclareMathDelimiter{\rrangle}{\mathclose}%
                     {MnLargeSymbols}{'171}{MnLargeSymbols}{'171}
\makeatother

\newcommand{\ket}[1]{| #1 \rangle}
\newcommand{\bra}[1]{\langle #1 |}
\newcommand{\braket}[2]{\langle #1|#2 \rangle}

\newcommand{\tr}{\mathrm{tr}}
\newcommand{\eps}{\varepsilon}
\renewcommand{\Pr}{\mathbb{P}}
\newcommand{\E}{\mathop{{}\mathbb{E}}}
\newcommand{\Var}{\mathrm{Var}}
\newcommand{\Pf}{\mathrm{pf}}

\newcommand{\kett}[1]{| #1 \rrangle}
\newcommand{\braa}[1]{\llangle #1 |}
\newcommand{\brakett}[2]{\llangle #1|#2 \rrangle}

\newcommand{\Google}{\affiliation{%
Google Quantum AI, Venice, CA 90291, USA}}

\newcommand{\Stanford}{\affiliation{%
Stanford Institute for Theoretical Physics, Stanford University, Stanford, CA 94305, USA}}

\newcommand{\Columbia}{\affiliation{Department of Chemistry, Columbia University, New York, NY, United States}}

\newcommand{\Harvard}{\affiliation{Department of Chemistry and Chemical Biology, Harvard University, MA 02138}}

\usepackage{thmtools}
\usepackage{thm-restate}
\usepackage{cleveref}

\newcommand{\citen}[1]{Ref.~\citenum{#1}}
\newcommand{\etal}[0]{\textit{et al.}~}

\begin{document}

\title{Matchgate Shadows for Fermionic Quantum Simulation}

\date{\today}

\author{Kianna Wan}\email{corresponding author: kianna.wan@gmail.com}
\Google
\Stanford

\author{William J.~Huggins}
\Google

\author{Joonho Lee}
\Google
\Columbia
\Harvard

\author{Ryan Babbush}
\Google

\begin{abstract}
``Classical shadows'' are estimators of an unknown quantum state, constructed from suitably distributed random measurements on copies of that state \cite{Huang2020}. In this paper, we analyze classical shadows obtained using random matchgate circuits, which correspond to fermionic Gaussian unitaries.
We prove that the first three moments of the Haar distribution over the \textit{continuous} group of matchgate circuits are equal to those of the \textit{discrete} uniform distribution over only the matchgate circuits that are also Clifford unitaries; thus, the latter forms a ``matchgate 3-design.''
This implies that the classical shadows resulting from the two ensembles are functionally equivalent.
We show how one can use these matchgate shadows to efficiently estimate inner products between an arbitrary quantum state and fermionic Gaussian states, as well as the expectation values of local fermionic operators and various other quantities, thus surpassing the capabilities of prior work. 
As a concrete application, this enables us to apply wavefunction constraints that control the fermion sign problem in the quantum-classical auxiliary-field quantum Monte Carlo algorithm (QC-AFQMC)~\cite{Huggins2022}, without the exponential post-processing cost incurred by the original approach.
\end{abstract}

\maketitle

\tableofcontents

\section{Introduction}

Efficiently extracting information from a quantum mechanical system is a task of both theoretical and practical importance.
As quantum information technology develops, we may wish to use it to characterise unknown quantum states and processes. 
Even when we possess a complete description of a quantum system, such as instructions for preparing a quantum state via a quantum circuit, we may need to compute properties of the system that are not efficiently obtainable from this description using classical computation.
The formalism of ``classical shadows'' introduced by Huang \textit{et al.}~\cite{Huang2020} provides us with a rigorous approach to solving this problem in some cases. 

A classical shadow of a quantum state \(\rho\) is an unbiased estimator of \(\rho\), constructed from the outcomes of random measurements on copies of \(\rho\). As its name suggests, samples of this estimator are stored classically, and can be used to estimate properties of $\rho$, such as the expectation values of a set of observables.
A classical shadows protocol is specified by choosing a distribution over unitaries; measurements are performed by sampling a unitary from this distribution, applying it to \(\rho\), and measuring in the computational basis.
The efficiency of any classical shadows scheme depends on both the choice of unitary distribution and the particular observables of interest.
Ideally, a scheme is efficient both in terms of quantum resources (the circuit depth and the number of random measurements required), and classical resources (the complexity of the post-processing for obtaining estimates from the classical shadow samples).

Recently, Huggins \etal\cite{Huggins2022} used classical shadows to implement a new hybrid algorithm, called ``quantum-classical hybrid quantum Monte Carlo,'' on a near-term quantum processor. 
Their approach is based on classical computational techniques, known as projector quantum Monte Carlo (QMC) methods, that approximate the ground state of a quantum Hamiltonian by stochastically implementing the imaginary time-evolution operator. 
Generically, QMC algorithms for fermionic systems are made to scale polynomially by imposing constraints derived from a ``trial wavefunction'' \(\ket{\Psi_{\textrm{trial}}}\), an ansatz for the ground state that is provided as an input to the algorithm.
This polynomial scaling comes at the expense of introducing a bias that depends on the quality of the ansatz. The use of richer families of trial wavefunctions to increase the accuracy of these constrained QMC calculations is an active area of research. Quantum computing offers the ability to efficiently prepare a new class of trial wavefunctions that are difficult to access classically, which motivated the development of a hybrid quantum-classical algorithm for QMC.

The details vary between specific methods for constrained QMC, but generally the constraints can naturally be cast in terms of inner products with \(\ket{\Psi_{\textrm{trial}}}\).
In \citen{Huggins2022}, the authors implemented a quantum-classical hybrid algorithm for auxiliary-field quantum Monte Carlo (QC-AFQMC), by collecting classical shadow samples of $\ket{\Psi_{\text{trial}}}$ and using these to estimate the inner products $\braket{\Psi_{\text{trial}}}{\varphi_i}$ between $\ket{\Psi_{\text{trial}}}$ and Slater determinants $\ket{\varphi_i}$.
The particular classical shadows protocol implemented in \citen{Huggins2022} is the one based on random Clifford circuits, proposed and analysed in Ref.~\citenum{Huang2020}. The use of classical shadows substantially reduced the number of circuit repetitions required
compared to the alternative approach of using Hadamard tests. However, the proposed scheme for classically post-processing the Clifford-based classical shadows to obtain the requisite inner product estimates is inefficient, with a runtime that scales exponentially with the number of qubits.

To remedy this exponential bottleneck in the QC-AFQMC algorithm, and motivated by the importance of fermionic quantum simulation in general, we develop new tomographic protocols based on classical shadows from random matchgate circuits. We refer to these classical shadows as ``matchgate shadows'' for brevity. 
Matchgate circuits, which we formally define in subsection~\ref{sec:background_matchgates}, are generated by a certain set of two-qubit Pauli rotations, and are equivalent to fermionic Gaussian unitaries under the Jordan-Wigner transformation~\cite{JordanWigner}.
We consider two distributions over matchgate circuits: the Haar-uniform distribution over the continuous group of all matchgate circuits, and the uniform distribution over the discrete subset of matchgate circuits that are also members of the Clifford group. In Theorem~\ref{thm:moments}, we establish that the first three moments of two distributions are the same, by finding explicit expressions for the corresponding twirl channels and showing that they are equal.
Thus, in the same way the uniform distribution over the Clifford group is a unitary 3-design~\cite{Zhu2017-xf}, we can colloquially describe our discrete distribution over Clifford matchgate circuits as a ``matchgate 3-design.''
In the context of classical shadows, the form of the estimators as well as their variance depend only on the first three moments, so it follows that the two distributions lead to the same results. In addition to potential practical implications, the $3$-design property is useful theoretically, as it allows us to exploit the more explicit symmetry of the full matchgate group when analysing the discrete ensemble.

Crucially, we also show how to \emph{efficiently} post-process our matchgate shadows to estimate three kinds of quantities: 
(i) expectation values of local fermionic operators, (ii) fidelities $\tr(\varrho\rho)$ between unknown quantum states $\rho$ and fermionic Gaussian states $\varrho$, and (iii) inner products $\braket{\psi}{\varphi}$ between any pure state \(\ket{\psi}\) (accessed via a state preparation circuit) and arbitrary Slater determinants \(\ket{\varphi}\).
We also analyse the variances of the resulting estimates, proving explicit polynomial upper bounds in cases (i) and (ii). In case (iii), we derive an efficiently computable bound, which we evaluate for system sizes up to \(1000\) qubits, finding a modest growth rate consistent with a sublinear scaling in the system size.
Beyond these three classes of observables, we provide a general framework for efficiently estimating the expectation values of arbitrary products of local fermionic operators, fermionic Gaussian density operators, and fermionic Gaussian unitaries, though we do not address the task of bounding the variance in this more general situation. We apply this framework to obtain an efficient procedure for estimating inner products between any pure state and arbitrary pure fermionic Gaussian states (not necessarily Slater determinants) using our matchgate shadows. Our post-processing procedures are based on novel methods for classically evaluating free-fermion quantities by exploiting their underlying Clifford algebra structure, which may be of independent interest in both classical and quantum computation. 

Our work builds on the classical shadows formalism of Huang \etal~\cite{Huang2020}, but several differences arise when considering random matchgate circuits rather than random Clifford circuits. First, the fact that the \(n\)-qubit Clifford group has only one non-trivial irreducible representation leads to a particularly simple form for the corresponding measurement channel~\cite{Chen2021-cq}. In contrast, the group of $n$-qubit matchgate circuits has \(2n + 1\) inequivalent irreducible representations~\cite{Helsen2022}, complicating the analysis. Another interesting difference is that when using classical shadows based on the Clifford group~\cite{Huang2020}, the choice of ensemble, between single- and $n$-qubit random Clifford circuits, allows one to efficiently estimate \textit{either} local qubit observables, or low-rank observables such as fidelities. 
In contrast, our work shows that matchgate shadows are capable of simultaneously estimating both local fermionic observables as well as certain global properties (e.g., the fidelities with fermionic Gaussian states).

The idea of using classical shadows from random matchgate circuits was also explored by Zhao \textit{et~al.} in \citen{Zhao2021-fv}.
We make a brief comparison here and contrast our work with theirs more thoroughly in subsection~\ref{sec:summarylocal}. Zhao \etal analyse a discrete ensemble of matchgate circuits that is a subset of the discrete ensemble considered in this paper. They apply the resulting shadows to estimate the expectation values of local fermionic observables in a particular basis. We obtain the same scaling as their approach for these local observables. 
Our practical results go considerably beyond theirs, however, by developing efficient methods and bounding the variances for estimating the nonlocal properties described above (including those required for QC-AFQMC), in addition to the expectation values of local fermionic observables in arbitrary bases. We thus broaden the scope of applicability of the classical shadows formalism to the quantum simulation of fermionic systems. 

We refer the reader to Section~\ref{sec:summary} for a summary of our main results; the relevant background material is reviewed in Section~\ref{sec:background}.
In subsection~\ref{sec:summary_random_matchgates}, we characterise the moments of the two distributions over matchgate circuits we consider. Then, in subsection~\ref{sec:summary_shadows}, we provide an expression for the corresponding measurement channel (necessary for classically constructing the matchgate shadows), as well as a general formula for the variance of expectation value estimates. 
We consider several applications in subsection~\ref{sec:summary_applications}, giving an overview of our methods for efficiently extracting estimates of various quantities from matchgate shadows via classical post-processing, along with bounds on the variances of these estimates. As a specific example, we present a concise description of our protocol for efficiently estimating the inner products required for QC-AFQMC in Algorithm~\ref{alg:overlaps}.
In Section~\ref{sec:ensembles}, we supply the proofs of our results on the ensembles of random matchgate circuits and the classical shadows they generate.
Section~\ref{sec:compute} provides details and proofs of correctness for our efficient post-processing methods, while Section~\ref{sec:variance} gives the proofs of our variance bounds. 
In Section~\ref{sec:discussion}, we discuss the context of our work and some directions for future exploration. The appendices contain a mixture of generalisations and more technical details related to the results of the main text.

\section{Background} \label{sec:background}

In this section, we provide the background material required for developing our results and putting them into context. In subsection~\ref{sec:background_preliminaries}, we summarise basic concepts and introduce some notational conventions that will be used throughout the paper. Table~\ref{tab:notation} contains an abbreviated list of commonly used notation. We then review the classical shadows framework of~\citen{Huang2020} in subsection~\ref{sec:background_shadows}, and describe the application of classical shadows to the QC-AFQMC algorithm of Ref.~\citenum{Huggins2022} in subsection~\ref{sec:background_QCQMC}. Readers familiar with this background material can skip to the summary of our main results in Section~\ref{sec:summary}, after skimming subsection~\ref{sec:background_preliminaries} or Table~\ref{tab:notation} to gain familiarity with some of the notation.

\begin{center}
\begin{table}
\caption{Some notation used throughout the paper.}
\label{tab:notation}
\small
\begin{tabular}{l l l}
\hline \vspace{-1em} \\
$n$ = number of fermionic modes = number of qubits &\hspace{2em} &$\mathrm{O}(2n) = $ (real) orthogonal group in dimension $2n$ \\[3pt]
$\zeta$ = number of fermions & &$\mathrm{B}(2n)$ = group of $2n \times 2n$ signed permutation matrices \\[3pt]
$\Hn \cong \mathbb{C}^{2^n}$ = space of $n$-qubit states & & $U_Q$ = Gaussian unitary/matchgate circuit  \\[3pt]
$\mathcal{L}(\mathcal{V}) = $ space of linear operators on vector space $\mathcal{V}$ & & \qquad\quad $(U_Q^\dagger \gamma_\mu U_Q = \sum_{\nu=1}^{2n} Q_{\mu\nu}\gamma_\nu$ for $Q \in \orth)$ \\[3pt]  
$[k] = \{1,\dots, k\}$, for $k \in \mathbb{Z}_{>0}$ & &$\mathcal{U}_Q(\,\cdot\,) = U_Q^\dagger(\,\cdot\,)U_Q$ \\[3pt]
${S\choose j} = \{T \subseteq S: |T| = j\}$, for set $S$  & &$\mathrm{M}_n$ = group of $n$-mode Gaussian unitaries/ \\[3pt]
$\gamma_\mu =$ $\mu$th Majorana operator & & \qquad \quad {$n$-qubit matchgate circuits}  \\[3pt] 
$\gamma_S = \gamma_{\mu_1}\dots \gamma_{\mu_{|S|}}$, where $S = \{\mu_1,\dots, \mu_{|S|}\}$, $\mu_1< \dots < \mu_{|S|}$ & &$\mathrm{Cl}_n$ = $n$-qubit Clifford group \\[3pt]
$\Gamma_k = \mathrm{span}\{\gamma_S: S \in {[2n]\choose k}\}$; \enspace $\Gamma_{\text{even}} = \bigoplus_{k \text{ even}} \Gamma_k$  & &$\Pf(M)$ = Pfaffian of a matrix $M$\\[7pt]
$\mathcal{P}_k =$ projector onto $\Gamma_k$ & & $M\big|_S$ = matrix $M$ restricted to rows and columns in $S$ \\ [5pt]
\multicolumn{3}{l}{Liouville representation:  $\kett{A} \equiv \frac{1}{\sqrt{\tr(A^\dagger A)}}A$, \quad $\mathcal{E}\kett{A} \equiv \frac{1}{\sqrt{\tr(A^\dagger A)}} \mathcal{E}(A)$, \quad $\brakett{A}{B} \equiv \frac{\tr(A^\dagger B)}{\sqrt{\tr(A^\dagger A)\tr(B^\dagger B)}}$ \enspace for $A,B \in \LHn \setminus \{0\}$} \\ \hline
\end{tabular}
\end{table}
\end{center}

\subsection{Preliminaries and notation} \label{sec:background_preliminaries}

Throughout, we use $\mathcal{H}_n$ to denote the space of $n$-qubit states,  and $\mathcal{L}(\mathcal{V})$ the space of linear operators on a vector space $\mathcal{V}$. Thus, $\LHn$ is the space of $n$-qubit operators, and $\LLHn$ the space of superoperators. All vector spaces we consider are over the field $\mathbb{C}$.

\subsubsection{Majorana operators} \label{sec:background_Majoranas}

For a system of $n$ fermionic modes corresponding to creation operators $a_1^\dagger,\dots, a_n^\dagger$, we define $2n$ \textit{Majorana operators} $\gamma_1, \dots \gamma_{2n}$ by
\begin{equation} \label{Majoranadef} \gamma_{2j-1} \coloneqq a_j + a_j^\dagger, \qquad \gamma_{2j} \coloneqq -i(a_j - a_j^\dagger) \end{equation}
for $j \in [n]\coloneqq \{1,\dots, n\}$.
The Majorana operators are Hermitian and satisfy the anticommutation relations
\begin{equation} \label{MajoranaCCR} \{\gamma_\mu, \gamma_\nu\} = 2\delta_{\mu,\nu}I \end{equation}
for all $\mu,\nu \in [2n]$. 
For a subset of indices $S \subseteq [2n]$, we denote by $\gamma_S$ the product of the Majorana operators indexed by the elements in $S$ in increasing order. That is, 
\begin{align*} &\gamma_{S} \coloneqq \gamma_{\mu_1}\dots \gamma_{\mu_{|S|}} \quad \text{for $S = \{\mu_1,\dots \mu_{|S|}\} \subseteq [2n]$ with $\mu_1 < \dots < \mu_{|S|}$}, \\
&\gamma_{\varnothing} \equiv I,
\end{align*}
where $I$ denotes the identity operator in $\LHn$.
It will be clear from context whether the subscript of $\gamma$ is a single index (usually represented in terms of lowercase Greek or Roman letters) or a subset of indices (usually represented by an uppercase Roman letter). We further define the independent subspaces 
\begin{equation} \label{Gammak} \Gamma_k \coloneqq \mathrm{span}\left\{\gamma_S: S \in {[2n]\choose k}\right\} \end{equation}
for $k \in \{0,\dots, 2n\}$, where ${[2n]\choose k}$ denotes the set of subsets of $[2n]$ of cardinality $k$. Many of the operators we will be considering are \textit{even operators}, i.e., they are in the subspace $\Gamma_{\text{even}}$ spanned by products of an even number of Majorana operators:
\begin{equation}\label{Gamma_even} \Gamma_{\text{even}} \coloneqq \bigoplus_{\ell = 0}^{n} \Gamma_{2\ell}. \end{equation}

We represent $n$-mode fermionic operators by $n$-qubit operators via the Jordan-Wigner transformation~\cite{JordanWigner} (and generally, we will not distinguish between a fermionic operator and its qubit representation\footnote{The results in this paper are actually essentially independent of the particular fermion-to-qubit mapping. That is, all of the results can be interpreted purely in terms of fermionic states and operators, if one interprets ``computational basis states'' as occupation number states, ``matchgate circuts'' as fermionic Gaussian unitaries (see subsection~\ref{sec:background_matchgates}), and the set of ``Clifford'' matchgate circuits with the set of Gaussian unitaries $U_Q$ such that $Q$ is a signed permutation matrix (see Eq.~\eqref{MnCl}). We consider the Jordan-Wigner transformation to obtain a simple and explicit qubit implementation of our protocols.}): 
\begin{equation} \label{JW}
    \gamma_{2j-1} = \left(\prod_{i=1}^{j-1} Z_i\right) X_j, \qquad \gamma_{2j} = \left(\prod_{i=1}^{j-1} Z_i\right) Y_j,
\end{equation}
where $Z_i$ denotes the $n$-qubit operator that acts as Pauli $Z$ on the $i$th qubit and as the identity on the rest of the qubits, and similarly for $X_i$ and $Y_i$. Then, the space $\LHn$ of $n$-qubit operators is spanned by products of Majorana operators, i.e., $\LHn = \bigoplus_{k=0}^{2n} \Gamma_k$. Denoting the eigenstates of $Z$ as $\ket{0}$ and $\ket{1}$, the simultaneous eigenstates of $\{Z_j\}_{j\in [n]}$ are $\bigotimes_{j \in [n]}\ket{b_j} \eqqcolon \ket{b}$ for $b = (b_1,\dots, b_n) \in \{0,1\}^n$. We refer to $\{\ket{b}\}_{b\in \{0,1\}^n}$ as the \textit{computational basis}. This corresponds under the Jordan-Wigner transformation to the occupation-number basis with respect to our fermionic modes $\{a_j^\dagger\}_{j\in [n]}$: we have $\ket{b} = (a_1^\dagger)^{b_1} \dots (a_n^\dagger)^{b_n}\ket{\mathbf{0}}$, where $\ket{\mathbf{0}} \equiv \ket{0}^{\otimes n}$ is the vacuum state, and
\begin{equation} \label{bb} \ket{b}\bra{b} = \prod_{j=1}^n\frac{1}{2}\left(I-i(-1)^{b_j}\gamma_{2j-1}\gamma_{2j}\right).\end{equation}

More generally, we will refer to any set of Hermitian operators satisfying the anticommutation relations in Eq.~\eqref{MajoranaCCR} as a set of Majorana operators. We can form different sets of Majorana operators by taking appropriate linear combinations of $\gamma_1,\dots,\gamma_{2n}$. Specifically, if
\begin{equation}\label{Majoranatilde} \widetilde{\gamma}_\mu = \sum_{\nu =1}^{2n} Q_{\mu\nu} \gamma_\nu \end{equation}
for each $\mu \in [2n]$, then $\{\widetilde{\gamma}_\mu\}_{\mu \in [2n]}$ are self-adjoint and satisfy $\{\widetilde{\gamma}_\mu, \widetilde{\gamma}_\nu\} = 2\delta_{\mu,\nu}$, if and only if the $2n \times 2n$ matrix $Q$ is real and orthogonal. Having picked out a special basis---the computational basis---for the space of $n$-qubit states, we will typically use $\gamma_\mu$ (without tildes) to denote the particular set of Majorana operators satisfying Eq.~\eqref{bb}, and sometimes refer to this set as the ``canonical'' basis of Majorana operators.

\subsubsection{Matchgate circuits/fermionic Gaussian unitaries} \label{sec:background_matchgates}

Let $\orth$ be the group of real orthogonal $2n \times 2n$ matrices. For any $Q \in \orth$, we use $U_Q$ to denote any unitary (acting on $\mathcal{H}_n$) such that
\begin{equation} \label{UQdef} U_Q^\dagger \gamma_\mu U_Q = \sum_{\nu=1}^{2n}Q_{\mu\nu}\gamma_\nu \end{equation}
for all $\mu \in [2n]$. We call any such unitary a \textit{(fermionic) Gaussian unitary}.  Gaussian unitaries transform between valid sets of Majorana operators [cf.~Eq.~\eqref{Majoranatilde}]. It is easily verified using Eq.~\eqref{UQdef} that for any $S \subseteq [2n]$,
\begin{equation} \label{UQgammaS}
U_Q^\dagger \gamma_S U_Q = \sum_{S' \in {[2n] \choose |S|}} \det(Q\big|_{S,S'})\gamma_{S'},
\end{equation} 
where $M\big|_{S,S'}$ denotes the restriction of the matrix $M$ to rows indexed by $S$ and columns indexed by $S'$.
Since $\LHn$ is spanned by $\{\gamma_S: S \subseteq [2n]\}$, $U_Q$ is fully determined by $Q$ up to an irrelevant global phase. It follows from Eq.~\eqref{UQgammaS} that for each $k \in \{0,\dots, 2n\}$, the subspace $\Gamma_k$ spanned by products of $k$ Majorana operators [Eq.~\eqref{Gammak}] is invariant under conjugation by any $U_Q$. The set of all Gaussian unitaries on $n$ modes forms a group, which we denote by $\mathrm{M}_n$.

\textit{Matchgate circuits} are qubit representations of fermionic Gaussian unitaries under the Jordan-Wigner transformation. We will generally use the terms matchgate circuit and Gaussian unitary interchangeably, and call $\mathrm{M}_n$ the ``matchgate group.''\footnote{We note that in \citen{Helsen2022}, the matchgate group is used to refer to only those Gaussian unitaries corresponding to special orthogonal matrices $Q \in \mathrm{SO}(2n)$, while the group corresponding to all of $\orth$ is referred to as the ``generalised'' matchgate group.} \textit{Matchgates} are a particular class of two-qubit gates, generated by two-qubit $X$ rotations of the form $\exp(i\theta X\otimes X)$ and single qubit $Z$ rotations $\exp(i\theta Z\otimes I)$ and $\exp(i\theta I\otimes Z)$. 
Matchgate circuits can then be defined as the unitaries generated by nearest-neighbour matchgates (where the $n$ qubits are placed on a line) along with the Pauli $X$ operator $X_n$ on the last qubit. To see why this definition of matchgate circuits corresponds to that of Gaussian unitaries, note that 
\[ iX_j X_{j+1} = \gamma_{2j}\gamma_{2j+1},\qquad iZ_j = \gamma_{2j-1}\gamma_{2j} \]
from Eq.~\eqref{JW}, and that for $\mu \neq \nu$,
\[ \exp(\theta \gamma_\mu\gamma_\nu)^\dagger \gamma_\xi \exp(\theta \gamma_\mu\gamma_\nu) = \begin{dcases} \cos(2\theta) \gamma_\mu + \sin(2\theta) \gamma_\nu \qquad &\xi = \mu, \\
-\sin(2\theta) \gamma_\mu + \cos(2\theta)\gamma_\nu \qquad &\xi = \nu, \\
\gamma_\xi, \qquad &\xi \not\in \{\mu,\nu\}.\end{dcases} \]
Thus, the nearest neighbour $X \otimes X$ rotations and single-qubit $Z$ rotations implement Gaussian unitaries corresponding to Givens rotations in planes spanned by the $\gamma_\mu$, $\gamma_{\mu+1}$ axes for every $\mu \in [2n-1]$; these generate all rotations in $\mathrm{SO}(2n)$. Adding in the $X$ operator on the $n$th qubit then generates $\mathrm{O}(2n)$, since $X_n$ implements the reflection that takes $\gamma_{2n} \mapsto -\gamma_{2n}$ and leaves all the other $\gamma_\mu$ unchanged, as can be seen from Eq.~\eqref{JW}.

\subsubsection{Fermionic Gaussian states and Slater determinants} \label{sec:background_Gaussianstates}

There are several equivalent ways of defining \textit{(fermionic) Gaussian states}. Physically speaking, they are the ground states and thermal states of non-interacting fermionic Hamiltonians. For our purposes, an $n$-mode Gaussian state is any state whose density operator $\varrho$ can be written as
\begin{equation} \label{varrhodef} \varrho = \prod_{j=1}^n \frac{1}{2}(I - i\lambda_j \widetilde{\gamma}_{2j-1}\widetilde{\gamma}_{2j})\end{equation}
for some coefficients $\lambda_j \in [-1,1]$ and Majorana operators $\widetilde{\gamma}_\mu = U_Q^\dagger \gamma_\mu U_Q = \sum_{\mu=1}^{2n} Q_{\mu\nu}\gamma_\nu$ with $Q \in \orth$. If $\lambda_j \in \{-1,1\}$ for all $j \in [n]$, then $\varrho$ is a pure Gaussian state; otherwise, $\varrho$ is a mixed Gaussian state. Gaussian unitaries map Gaussian states to Gaussian states. In particular, from Eq.~\eqref{bb}, we see that computational basis states are all Gaussian states, and that any pure Gaussian state can be prepared from the vacuum state $\ket{\mathbf{0}}$ by a Gaussian unitary $U_Q \in \mathrm{M}_n$. 

The density operator of any Gaussian state is in $\Gamma_{\text{even}}$, and (as the name suggests) a Gaussian state is fully determined by its two-point correlations $\tr(\varrho \gamma_\mu \gamma_\nu)$, which form its \textit{covariance matrix}. Specifically, for any $n$-qubit state $\rho$, the associated covariance matrix $C_\rho$ is an antisymmetric $2n\times 2n$ matrix with entries
\begin{equation} \label{covariancematrix} (C_\rho)_{\mu\nu} \coloneqq -\frac{i}{2}\tr([\gamma_\mu,\gamma_\nu] \rho) \end{equation}
for $\mu,\nu \in [2n]$. For example, by Eq.~\eqref{bb}, the covariance matrix of a computational basis state $\ket{b}\bra{b}$ is 
\begin{equation} \label{covarianceb} C_{\ket{b}} \coloneqq \bigoplus_{j=1}^n \begin{pmatrix} 0 & (-1)^{b_j} \\ (-1)^{b_j + 1} &0 \end{pmatrix}. \end{equation} For a general Gaussian state $\varrho$ specified as in Eq.~\eqref{varrhodef}, the covariance matrix is 
\begin{equation} \label{covariancegeneral} C_\varrho = Q^{\mathrm{T}} \bigoplus_{j=1}^n \begin{pmatrix} 0 & \lambda_j \\ -\lambda_j &0 \end{pmatrix} Q,
\end{equation}
and a useful relation between the covariance matrices of $\varrho$ and $U_Q^\dagger \varrho U_Q$ for some $U_Q \in \mathrm{M}_n$ is $C_{U_Q^\dagger\varrho U_Q} = Q^{\mathrm{T}} C_\varrho Q$. 

For $\zeta \in \mathbb{Z}_{\geq 0}$, a $\zeta$-fermion \text{Slater determinant} is a Gaussian state that is also an eigenstate of the number operator $\sum_{j =1}^n a_j^\dagger a_j$. 
Not every Gaussian state is a Slater determinant. Indeed, the definition of a Slater determinant depends on the choice of fermionic modes $\{a_j^\dagger\}_{j \in [n]}$; as discussed in subsection~\ref{sec:background_Majoranas}, throughout this paper, $\{a_j^\dagger\}_{j \in [n]}$ are the ``canonical'' modes whose occupation-number states correspond to computational basis states under the Jordan-Wigner transformation. Hence, computational basis states are all Slater determinants, and any $\zeta$-fermion Slater determinant can be prepared from a computational basis state $\ket{x}$ of Hamming weight $|x| = \zeta$ by a Gaussian unitary that commutes with the number operator. 

Fermionic Gaussian states, including Slater determinants, can be efficiently described classically via their covariance matrices (Eq.~\eqref{covariancematrix}). Any $\zeta$-fermion Slater determinant $\ket{\varphi}$ can also be written as
\begin{equation} \label{ajtildeSlater} \ket{\varphi} = \widetilde{a}_1^\dagger\dots \widetilde{a}_{\zeta}^\dagger\ket{\mathbf{0}}, \qquad \text{where }\widetilde{a}_j = \sum_{k=1}^n V_{jk} a_k \text{ for each $j \in [n]$}, \end{equation}
for some $n \times n$ unitary matrix $V$. Hence, we can also specify a Slater determinant by specifying $V$ (or the first $\zeta$ rows of $V$). For reference, the Majorana $\{\widetilde\gamma_\mu\}_{\mu \in [2n]}$ operators corresponding to these $\{\widetilde{a}_j\}_{j \in [n]}$ $\{\gamma_\mu\}_{\mu \in [2n]}$ by $\widetilde{\gamma}_{2j-1} = \sum_k (\mathrm{Re}(V_{jk})\gamma_{2k-1} -\mathrm{Im}(V_{jk})\gamma_{2k})$ and $\widetilde{\gamma}_{2j} = \sum_{k}(\mathrm{Im}(V_{jk})\gamma_{2k-1} + \mathrm{Re}(V_{jk})\gamma_{2k})$, so the fermionic Gaussian unitary $U_{\widetilde{Q}}$ that implements the transformation in Eq.~\eqref{ajtildeSlater} is given by the (special) orthogonal matrix \begin{equation} \label{QSlater} \widetilde{Q} = \begin{pmatrix} R_{11} &\dots &R_{1n} \\ 
\vdots &\ddots &\vdots \\ 
R_{n1} &\dots &R_{nn} \end{pmatrix}, \qquad \text{with blocks } R_{jk} \coloneqq \begin{pmatrix} \mathrm{Re}(V_{jk}) &-\mathrm{Im}(V_{jk}) \\ \mathrm{Im}(V_{jk}) &\mathrm{Re}(V_{jk}) \end{pmatrix}. \end{equation}

\subsubsection{Liouville representation} \label{sec:Liouville}

In some parts of this paper, predominantly in Section~\ref{sec:ensembles}, we will use the Liouville representation (or Pauli-transfer matrix representation) for operators and superoperators, for the purpose of making certain expressions more clear. In this representation, operators are notated using ``double'' kets, and a ``double'' braket is used to represent the Hilbert-Schmidt inner product. By convention, we take all double kets to be normalised with respect to the Hilbert-Schmidt norm. Thus, for any nonzero operators $A,B \in \LHn$, 
\[ \kett{A} \equiv \frac{1}{\sqrt{\tr(A^\dagger A)}}A, \qquad \brakett{A}{B} \equiv \frac{\tr(A^\dagger B)}{\sqrt{\tr(A^\dagger A)\tr(B^\dagger B)}} \]
(and set $\kett{0} \equiv 0$ for the zero operator).
In particular, since Majorana operators square to the identity, and their products are Hilbert-Schmidt orthogonal, we have
\begin{equation} \label{kettgammaS} \kett{\gamma_S} \equiv \frac{1}{\sqrt{2^n}}\gamma_S, \qquad \brakett{\gamma_S}{\gamma_{S'}} = \delta_{S,S'}. \end{equation} 
As with usual (state) kets, we will freely write e.g., $\kett{A}\kett{B}$ as shorthand for the tensor product $\kett{A}\otimes \kett{B}$.
A superoperator acting on an operator is represented by placing the superoperator to the left of the operator's double ket, i.e., for $\mathcal{E} \in \LLHn$ and $A \in \LHn$, 
\[ \mathcal{E}\kett{A} \equiv \frac{1}{\sqrt{\tr(A^\dagger A)}} \mathcal{E}(A). \] We will also write $\mathcal{E}\mathcal{E}'$ in place of $\mathcal{E}\circ \mathcal{E}'$. The fact that the $2^n$ ordered products $\gamma_S$ of Majorana operators forms an orthogonal basis for $\LHn$ can be expressed as a resolution of the identity superoperator $\mathcal{I} \in \mathcal{L}(\mathcal{L}(\Hn))$:
\begin{equation} \label{resolutionofidentity} \mathcal{I} = \sum_{S \subseteq [2n]}\kett{\gamma_S}\braa{\gamma_S}. \end{equation}

\subsection{Review of the classical shadows framework} \label{sec:background_shadows}

In this subsection, we review the \textit{classical shadows} framework of~\citen{Huang2020},\footnote{Similar ideas were developed independently in Refs.~\cite{paini2019approximate,Paini_2021}.} introducing generalisations where necessary. The objective is to estimate the expectation values $\tr(O_1 \rho),\dots, \tr(O_M\rho)$ of $M$ ``observables''\footnote{Usually, the term observable is reserved for operators that are Hermitian. In this paper, we abuse terminology and allow our observables to be non-Hermitian in general (using the term to refer to any operator whose expectation value we are interested in).} $O_1,\dots, O_M \in \LHn$ with respect to an unknown $n$-qubit quantum state $\rho$, assuming that we are given copies of the state and some classical description of the observables. To apply the classical shadows protocol, we first choose a distribution $D$ over some set of unitaries. For each copy of $\rho$, we 1) randomly draw a unitary $U$ from this distribution, 2) apply $U$ to $\rho$, and 3) measure in the computational basis $\{\ket{b}\}_{b\in\{0,1\}^n}$. (Steps 2 and 3 are equivalent to measuring in the basis $\{U^\dagger\ket{b}\}_{b \in \{0,1\}^n}$.) Then, consider applying $U^\dagger$ to the post-measurement state. The quantum channel $\mathcal{M} \in \LLHn$ describing the overall process is given by 
\begin{align} 
    \mathcal{M}(\rho) &= \E_{U \sim D}\sum_{b \in \{0,1\}^n} \bra{b}U \rho U^\dagger\ket{b}U^\dagger\ket{b}\bra{b}U \nonumber\\
    &= \tr_1\left[\sum_{b \in \{0,1\}^n} \E_{U \sim D} \mathcal{U}^{\otimes 2} (\ket{b}\bra{b}^{\otimes 2}) (\rho\otimes I) \right], \label{Mgeneral2}
\end{align}
where $\tr_1$ denotes the partial trace over the first tensor component, and $\mathcal{U}$ denotes the unitary channel corresponding to $U^\dagger$, i.e., $\mathcal{U}(\,\cdot\,) = U^\dagger(\,\cdot\,)U$. Now, suppose that $\mathcal{M}$ is invertible on some subspace $\mathcal{X}$ of $\LHn$, and assume that $\mathcal{X}$ contains $\rho$ and $U^\dagger\ket{b}\bra{b}U$ for all $U \in D$ and $b \in \{0,1\}^n$. Let $\mathcal{M}^{-1}: \mathcal{X} \to \mathcal{X}$ denote the inverse of $\mathcal{M}$ restricted to this subspace. Then, define the random operator $\hat{\rho}$ by 
\begin{equation} \label{hatrho} \hat{\rho} \coloneqq \mathcal{M}^{-1}(\hat{U}^\dagger\ket{\hat{b}}\bra{\hat{b}}\hat{U}), \end{equation}
where $\hat{U}$ is distributed according to $D$ and $\Pr[\ket{\hat{b}} = \ket{b}\, |\, \hat{U} = U] = \bra{b}U\rho U^\dagger \ket{b}$. By construction, $\hat{\rho}$ is an unbiased estimator for $\rho$:
\[ \E[\hat\rho] = \rho, \]
and in the literature, the term ``classical shadow'' is often used to refer to both $\hat{\rho}$ (a random variable) or a sample of it obtained in one realisation of the above procedure (i.e., $\mathcal{M}^{-1} (U^\dagger\ket{b}\bra{b}U)$ for some outcomes $U$ and $\ket{b}$). These samples can be used to estimate the expectation values $\tr(O_i\rho)$, since
\[ \E[\hat{o}_i] = \tr(O_i \rho), \quad \text{where $\hat{o}_i \coloneqq \tr(O_i \hat{\rho}) = \tr\left(O_i \mathcal{M}^{-1}(\hat{U}^\dagger\ket{\hat{b}}\bra{\hat{b}}\hat{U})\right)$ for $i \in [M]$}. \]
Note that only steps 1)--3) are performed on the quantum computer; the remaining computations (constructing the classical shadow sample $\mathcal{M}^{-1}(U\ket{b}\bra{b}U)$ and calculating expectation values with respect to it) are classical. 

To bound the number of samples of the classical shadow estimator $\hat{\rho}$ (and hence the number of copies of $\rho$) required to estimate the expectation values to within some desired precision with high probability, we can consider the variances of the estimators $\hat{o}_i$:
\begin{align*}
    \Var[\hat{o}_i] &= \E[|\hat{o}_i|^2] - |\E[\hat{o}_i]|^2 \\
    &= \E_{U \sim D} \sum_{b \in \{0,1\}^n} \bra{b} U\rho U^\dagger\ket{b} \left|\tr\left(O_i \mathcal{M}^{-1}(U^\dagger \ket{b}\bra{b}U)\right)\right|^2 - |\tr(O_i\rho)|^2 \\
    &= \E_{U \sim D} \sum_{b \in \{0,1\}^n} \tr\left[U^\dagger \ket{b}\bra{b} U \rho \otimes \mathcal{M}^{-1}(U^\dagger \ket{b}\bra{b}U)O_i \otimes  \mathcal{M}^{-1}(U^\dagger\ket{b}\bra{b}U) O_i^\dagger\right] - |\tr(O_i\rho)|^2.
\end{align*}
If we further assume that $O_i,O_i^\dagger \in \mathcal{X}$, then we can use the fact that $\tr(\mathcal{M}^{-1}(A) B) = \tr(A \mathcal{M}^{-1}(B))$ for $A,B, \in \mathcal{X}$ to rewrite this as
\begin{equation} \label{Varoi}
    \Var[\hat{o}_i] = \tr\left[\sum_{b \in \{0,1\}^n} \E_{U \in D} \mathcal{U}^{\otimes 3}(\ket{b}\bra{b}^{\otimes 3}) \left(\rho \otimes \mathcal{M}^{-1}(O_i) \otimes \mathcal{M}^{-1}(O_i^\dagger)\right) \right] - |\tr(O_i\rho)|^2,
\end{equation}
which can be upper bounded by the first term.
Thus, we see from Eqs.~\eqref{Mgeneral2} and~\eqref{Varoi} that while the measurement channel $\mathcal{M}$ in the classical shadows protocol depends on the chosen unitary distribution $D$ through the 2-fold twirl $\E_{U \sim D}\mathcal{U}^{\otimes 2}$, the variances of the resulting estimates depend on $D$ through the 3-fold twirl $\E_{U \sim D}\mathcal{U}^{\otimes 3}$. If median-of-means estimators (see e.g., \cite{lerasle2019lecture}) are used, it follows straightforwardly from Chebyshev's and Hoeffding's inequalities that 
\begin{equation} \label{Nsample} N_{\mathrm{sample}} = \mathcal{O}\left(\frac{\log (M/\delta)}{\eps^2} \max\limits_{i \in [M]} \Var[\hat{o}_i] \right) \end{equation} 
classical shadow samples suffice to estimate every $\tr(O_i\rho)$ to within additive error $\eps$ with probability at least $1-\delta$.\footnote{Specifically, to form the median-of-means estimator for each observable, we take $N_{\mathrm{sample}} = KL$ samples in total, where $K = \mathcal{O}(\log(M/\delta))$ and $L = \mathcal{O}(\max_{i \in [M]}\Var[\hat{o}_i]/\eps^2$) (cf. Step~2 of Algorithm~\ref{alg:overlaps}). Dividing into $K$ groups of $L$ samples each, we compute the sample mean of each group, then take the median of the sample means. The median-of-means estimator is employed to achieve logarithmic scaling with the inverse failure probability. If, for instance, a sample mean estimator were used instead, we would not be able to apply tighter concentration bounds such as Hoeffding's inequality, because estimates derived from the classical shadow are not in general bounded: like density operators, the classical shadow estimator satisfies $\tr[\hat{\rho}] = 1$, but unlike density operators, $\hat{\rho}$ is not positive semi-definite.}

It is important to note that the classical shadows framework does not in general provide a protocol that is efficient in terms of quantum and classical resources. Even in cases where there is an efficient procedure to sample unitaries $U$ from $D$ and implement them on a quantum computer, the variance of the estimates $\hat{o}_i$ may be large, necessitating a large number of copies of $\rho$ (if the number of samples is chosen according to Eq.~\eqref{Nsample}). In addition, one needs to be able to somehow classically compute $\tr(O_i \mathcal{M}^{-1}(U^\dagger\ket{b}\bra{b}U))$ for all $i$.
The particular classical shadows approach taken in \citen{Huggins2022} to implement QC-AFQMC, described in the following subsection, provides an example of a situation where the variance is small (in fact, bounded by a constant in the system size $n$), but the classical post-processing is inefficient (having a runtime exponential in $n$). One of the motivations behind the present work is to provide a protocol that has $\mathrm{poly}(n)$ quantum \textit{and} classical complexity for estimating expectation values of a large class of fermionic observables, including (but not limited to) those required for QC-AFQMC.

\subsection{Classical shadows applied to quantum-classical auxiliary-field quantum Monte Carlo (QC-AFQMC)} \label{sec:background_QCQMC}

In this subsection, we review one of the motivating use cases for the protocols we develop in this paper.
We begin by briefly discussing projector QMC techniques.
Formally, we can find the ground state of a Hamiltonian \(H\) by applying the imaginary time-evolution operator $e^{-\tau H}$ to some initial state $\ket{\psi_\textrm{init}}$ that has non-vanishing overlap with the true ground state $\ket{\psi_\textrm{ground}}$:
\begin{equation} \label{QMCgs}
    \ket{\psi_\textrm{ground}} \propto \lim_{\tau \rightarrow \infty} e^{-\tau H} \ket{\psi_\textrm{init}}.
\end{equation}
In order to avoid explicitly storing and manipulating exponentially large objects, projector QMC methods approximate this projection onto the ground state by implementing it stochastically.
In some cases, such as for unfrustrated systems of bosons, this yields a polynomially scaling procedure for computing the ground state energy and other properties.

However, when we consider systems containing multiple identical fermions, projector QMC methods are typically faced with the \textit{fermion sign problem}. 
In projector QMC methods based on second quantisation, such as \textit{auxiliary-field qantum Monte Carlo} (AFQMC)~\cite{Zhang2003-ql}, the sign problem manifests as an exponentially large variance in the estimator of the energy~\cite{Zhang2003-ql}.
When a polynomially scaling approach is desired, the fermion sign problem is usually controlled by applying constraints to the statistical samples of the RHS of Eq.~\eqref{QMCgs} using an approximation to the ground state referred to as the ``trial wavefunction'' \(\ket{\Psi_\textrm{trial}}\). Details vary between methods, but the basic idea is to modify the statistical samples so as to constrain their overlaps with the trial wavefunction to be positive. Implementing this constraint correctly for a statistical sample \(\ket{\varphi}\) requires calculating \(\braket{\Psi_\textrm{trial}}{\varphi}\).

In \citen{Huggins2022}, Huggins \textit{et al.} proposed and implemented a quantum-classical hybrid quantum Monte Carlo algorithm, which involved preparing the trial wavefunction $\ket{\Psi_{\textrm{trial}}}$ on a quantum computer, and using a classical shadows protocol to estimate its overlaps with the statistical samples.
\citen{Huggins2022} focused on AFQMC, where the statistical sample are Slater determinants $\ket{\varphi_i}$ in general single-particle bases. Given a quantum circuit \(U_\Psi\) that prepares \(\ket{\Psi_\textrm{trial}}\) from the vacuum state $\ket{\mathbf{0}}$, it is straightforward to use Hadamard tests to estimate the overlap between   \(\ket{\Psi_\textrm{trial}}\) and any Slater determinant \(\ket{\varphi_i}\).
However, evaluating these overlaps one at a time (and preparing $\ket{\Psi_{\text{trial}}}$ for each of the Hadamard tests performed for each overlap) for the large number of Slater determinants that arise in a typical AFQMC calculation could be prohibitively expensive.\footnote{Another difficulty with implementing a hybrid quantum-classical AFQMC scheme using Hadamard tests is that the Slater determinants whose overlaps are needed in later stages of the Monte Carlo calculation depend on the results of earlier overlap calculations. The necessary communication between the quantum and classical computers could lead to further slowdowns.}

In order to make an experiment on a near-term quantum computer feasible, Huggins \etal designed a protocol that involves first collecting a pre-determined number of classical shadow samples using the quantum device. 
Then, as the QMC calculation is run on the classical computer, the necessary overlaps can be evaluated classically using the stored shadow samples. 
More specifically, the protocol entails preparing the state 
\begin{equation*}
    \rho = \frac{1}{2} \left(\ket{\mathbf{0}} + \ket{\Psi_\textrm{trial}}\right)\left(\bra{\mathbf{0}} + \bra{\Psi_\textrm{trial}}\right),
\end{equation*}
and collecting classical shadow samples of it. 
On the classical computer, one specifies the operator
\begin{equation*}
    \ket{\varphi_i}\bra{\mathbf{0}} = \widetilde{a}_1^\dagger\dots \widetilde{a}_{\zeta}^\dagger \ket{\mathbf{0}}\bra{\mathbf{0}},
\end{equation*}
where the free parameters of \(\ket{\varphi_i}\) are contained in the choice of basis for the \(\widetilde{a}_k^\dagger\) operators (see subsection~\ref{sec:background_Gaussianstates}). Ref.~\citenum{Huggins2022} considers trial wavefunctions $\ket{\Psi_{\text{trial}}}$ that are eigenstates of the number operator, with eigenvalue $\zeta > 0$, and Slater determinants $\ket{\varphi_i}$ with the same number of particles $\zeta$. 
(In Appendix~\ref{app:overlaps_Slater} we discuss how to relax this requirement and allow for \(\ket{\Psi_\textrm{trial}}\) that does not have a fixed particle number.) In this case, we have $\braket{\Psi_{\text{trial}}}{\mathbf{0}} = \braket{\varphi_i}{\mathbf{0}} = 0$, so 
the overlap of interest can be expressed as
\begin{equation*}
    \braket{\Psi_{\textrm{trial}}}{\varphi_i} = 2 \tr\left(\ket{\varphi_i}\bra{\mathbf{0}}\rho\right).
\end{equation*}

Hence, the overlap \(\ket{\Psi_\textrm{trial}}\) can be estimated by estimating the expectation value of $\ket{\varphi_i}\bra{\mathbf{0}}$ with respect to the state $\rho$. This can in principle be done using the classical shadow samples of \(\rho\). In the particular classical shadows protocol used by Ref.~\citenum{Huggins2022}, the distribution $D$ (see subsection~\ref{sec:background_shadows}) is the uniform distribution over the $n$-qubit Clifford group. The variances of the resulting estimators were analysed in~\cite{Huang2020}, and it is straightforward to show that the variance for estimating the expectation value of $\ket{\varphi_i}\bra{\mathbf{0}}$ is bounded above by a constant. However, the cost of classically post-processing the classical shadow samples to obtain these estimates appears to scale exponentially with the system size $n$. Explicitly, 
Clifford classical shadow estimates $o_i$ of the expectation value of $\ket{\varphi_i}\bra{\mathbf{0}}$ are of the form
\begin{equation*}
    o_i = \left(2^{n} + 1\right) \bra{b}U \ket{\varphi_i}\bra{\mathbf{0}}U^\dagger\ket{b},
\end{equation*}
where $U$ is a Clifford unitary and $\ket{b}$ is a computational basis state.
In order to evaluate $o_i$ to within a constant additive error, we would therefore need to calculate $\bra{b}U \ket{\varphi_i}\bra{\mathbf{0}}U^\dagger\ket{b}$ up to an additive error that is exponentially small in \(n\). It is not clear how to evaluate the $\bra{b}U \ket{\varphi_i}$ component of this expression with the necessary degree of precision without resorting to methods that scale exponentially with \(n\) in the general case.

In the experimental implementation of QC-AFQMC in \citen{Huggins2022}, Clifford classical shadows were used despite this exponential complexity of the classical post-processing, 
because the system sizes considered were sufficiently small. 
The techniques we present in this work, based on classical shadows from different ensembles of unitaries, will allow this exponentially costly step to be removed in future implementations of QC-AFQMC.

\section{Summary} \label{sec:summary}

In this section, we present an overview of the key results of this paper. We provide references to sections that contain the proofs and additional details. All of the notation and background concepts we use in this section are explained in Section~\ref{sec:background}.

\subsection{Random matchgate circuits} \label{sec:summary_random_matchgates}
 
In this work, we consider the classical shadows resulting from two different distributions over matchgate circuits. As discussed in subsection~\ref{sec:background_matchgates}, matchgate circuits correspond to fermionic Gaussian unitaries via the Jordan-Wigner transformation, and form a continuous group $\mathrm{M}_n$ which is in one-to-one correspondence with the orthogonal group $\orth$ (if we ignore global phases): 
\begin{equation} \label{Mn} \mathrm{M}_n = \{U_Q: Q \in \orth\}. \end{equation}
The first distribution we study is the ``uniform'' distribution over $\mathrm{M}_n$, where uniformity is more precisely given by the normalised Haar measure $\mu$ on $\orth$. The second distribution is the uniform distribution over the discrete subgroup of $\mathrm{M}_n$ consisting only of matchgate circuits that are also in the $n$-qubit Clifford group $\mathrm{Cl}_n$. From Eq.~\eqref{UQdef}, these coincide with the group $\mathrm{B}(2n)$ of $2n \times 2n$ signed permutation matrices:
\begin{equation} \label{MnCl} \mathrm{M}_n \cap \mathrm{Cl}_n = \{U_Q: Q \in \mathrm{B}(2n)\}. \end{equation}
We explain how to efficiently sample from these two distributions in Appendix~\ref{app:random_circuit_sampling}.

For $j \in \mathbb{Z}_{>0}$, we use $\mathcal{E}^{(j)}_{\mathrm{M}_n}$ and $\mathcal{E}^{(j)}_{\mathrm{M}_n \cap \mathrm{Cl}_n}$ to denote the $j$-fold twirl channels corresponding to the distributions over $\mathrm{M}_n$ and $\mathrm{M}_n \cap \mathrm{Cl}_n$, respectively:
\begin{equation} \label{EjMn} \mathcal{E}_{\mathrm{M}_n}^{(j)} \coloneqq \int_{\orth} d\mu(Q)\, \mathcal{U}_Q^{\otimes j} \end{equation}
\begin{equation} \label{EjMnCln} \mathcal{E}_{\mathrm{M}_n \cap \mathrm{Cl}_n}^{(j)} \coloneqq \frac{1}{|\mathrm{B}(2n)|} \sum_{Q \in \mathrm{B}(2n)} \mathcal{U}_Q^{\otimes j}, \end{equation}
where \begin{equation} \label{mathcalUQ} \mathcal{U}_Q(\, \cdot \,) \coloneqq U_Q^\dagger (\,\cdot \,)U_Q \end{equation} denotes the unitary channel for the Gaussian unitary $U_Q^\dagger$. 
Since the measurement channel $\mathcal{M}$ in the classical shadows procedure and the variance of the estimates obtained from the classical shadows are determined by the 2- and 3-fold twirls [see Eqs.~\eqref{Mgeneral2} and~\eqref{Varoi}], our first step is to evaluate $\mathcal{E}^{(j)}_{\mathrm{M}_n}$ and $\mathcal{E}^{(j)}_{\mathrm{M}_n \cap \mathrm{Cl}_n}$ for $j$ up to $3$:
\begin{theorem}[First three moments of uniform distributions over $\mathrm{M}_n$ and $\mathrm{M}_n \cap \mathrm{Cl}_n$] \label{thm:moments} Let $\mathcal{E}^{(j)}_{\mathrm{M}_n}, \mathcal{E}^{(j)}_{\mathrm{M}_n  \cap \mathrm{Cl}_n} \in \LLHn^{\otimes j}$ be defined as in Eqs.~\eqref{EjMn} and~\eqref{EjMnCln}. Then, we have
\begin{align} 
    {(i) \enspace} &\mathcal{E}^{(1)}_{\mathrm{M}_n} = \mathcal{E}^{(1)}_{\mathrm{M}_n \cap \mathrm{Cl}_n} = \kett{I}\braa{I}, \nonumber \\
    {(ii) \enspace} & \mathcal{E}^{(2)}_{\mathrm{M}_n} = \mathcal{E}^{(2)}_{\mathrm{M}_n \cap \mathrm{Cl}_n} = \sum_{k = 0}^{2n} \kett{\Upsilon^{(2)}_k}\braa{\Upsilon^{(2)}_k}, \nonumber \\
    {(iii) \enspace} & \mathcal{E}^{(3)}_{\mathrm{M}_n} = \mathcal{E}^{(3)}_{\mathrm{M}_n \cap \mathrm{Cl}_n} =\sum_{\substack{k_1,k_2,k_3 \geq 0 \\ k_1 + k_2 + k_3 \leq 2n}}\kett{\Upsilon^{(3)}_{k_1,k_2,k_3}}\braa{\Upsilon^{(3)}_{k_1,k_2,k_3}},\nonumber\\
    \noalign{\noindent where} 
    &\quad \kett{\Upsilon^{(2)}_k} \coloneqq {2n \choose k}^{-1/2} \sum_{\substack{S \subseteq [2n] \label{Psi2}\\ |S|= k}} \kett{\gamma_S}\kett{\gamma_S}, \\ 
    &\quad \kett{\Upsilon^{(3)}_{k_1,k_2,k_3}} \coloneqq {{2n \choose k_1,k_2, k_3,2n-k_1-k_2-k_3}}^{-1/2} \sum_{\substack{A_1,A_2, A_3 \subseteq [2n] \text{ } \mathrm{ disjoint} \\ |A_1| = k_1, |A_2| = k_2, |A_3| = k_3} }   \kett{\gamma_{A_1}\gamma_{A_2}}\kett{\gamma_{A_2}\gamma_{A_3}} \kett{\gamma_{A_3}\gamma_{A_1}}. \label{Psi3}
\end{align}
\end{theorem}
We prove Theorem~\ref{thm:moments} in subsection~\ref{sec:moments} (see also subsection~\ref{sec:Liouville} for an explanation of the Liouville representation notation used here). In addition to providing explicit expressions for the relevant twirl channels, Theorem~\ref{thm:moments} shows in particular that the third moments of the two distributions are equal. Thus, the \textit{discrete} ensemble of Clifford matchgate circuits is a 3-design for the \textit{continuous} Haar-uniform distribution over all matchgate circuits, in the same way that the Clifford group is a unitary 3-design~\cite{Zhu2017-xf}. We can state this result informally as:
\begin{corollary} \label{cor:3design} The group of matchgate circuits that are also Clifford unitaries forms a ``matchgate 3-design.''
\end{corollary}
This is a general result that can be applied in any context that involves up to the third moment of uniformly random matchgate circuits. In the specific context of classical shadows, it implies that the discrete and continuous ensembles we defined above lead to the same measurement channel and variances, so we can in principle use either ensemble and obtain the same results. In addition to possible practical implications (e.g., it may be easier to sample and implement unitaries from one distribution than the other, depending on hardware capabilities), from a mathematical perspective, the 3-design property is useful in that it allows us to use the additional symmetry of the full matchgate group to more easily analyse the discrete ensemble. We will see explicit examples of this in the following subsections. 

\subsection{Matchgate shadows} \label{sec:summary_shadows}
Theorem~\ref{thm:moments} allows us to characterise the classical shadows associated with the two distributions over matchgate circuits described above. We will refer to these colloquially as ``matchgate classical shadows'' or more simply ``matchgate shadows.''  Substituting the expression for $\mathcal{E}_{\mathrm{M}_n}^{(2)} = \mathcal{E}_{\mathrm{M}_n \cap \mathrm{Cl}_n}^{(2)}$ from Theorem~\ref{thm:moments}(ii) for $\E_{U \sim D} \mathcal{U}^{\otimes 2}$ in Eq.~\eqref{Mgeneral2}, we show in subsection~\ref{sec:shadows_channel} that the classical shadows measurement channel $\mathcal{M}$ (for both distributions) is given by\footnote{The reader may notice that this is the same channel as that in \citen{Zhao2021-fv}, which uses a different distribution over matchgate circuits; see subsection~\ref{sec: comparison prior} for a detailed discussion and comparison.}
\begin{equation} \label{matchgateshadowschannel} \mathcal{M} = \sum_{\ell =0}^{n}{n\choose \ell} {2n \choose 2\ell}^{-1} \mathcal{P}_{2\ell},\end{equation}
where for $k \in \{0,\dots, n\}$, $\mathcal{P}_k \in \LLHn$ denotes the projector onto the subspace $\Gamma_k$ of $\LHn$ spanned by all products of $k$ Majorana operators, i.e., in Liouville representation,
\begin{equation} \label{Pk} \mathcal{P}_k \coloneqq \sum_{S \in {[2n]\choose k}} \kett{\gamma_S}\braa{\gamma_S}. \end{equation}
Consequently, the image of $\mathcal{M}$ is the subspace $\Gamma_{\text{even}} \coloneqq \bigoplus_{\ell=0}^{n} \Gamma_{2\ell}$
of $\LHn$ consisting of even operators, and the (pseudo)inverse $\mathcal{M}^{-1}: \Gamma_{\text{even}} \to \Gamma_{\text{even}}$ on this subspace is 
\begin{equation} \label{Minverse}
    \mathcal{M}^{-1} = \sum_{\ell = 0}^{n} {2n\choose 2\ell} {n\choose \ell}^{-1} \mathcal{P}_{2\ell}.
\end{equation}
Since $U_Q^\dagger \ket{b}\bra{b}U_Q \in \Gamma_{\text{even}}$ for any $Q \in \orth$ and computational basis state $\ket{b}$ (see Eqs.~\eqref{bb} and~\eqref{UQgammaS}), our matchgate shadow samples $\mathcal{M}^{-1}(U_Q^\dagger\ket{b}\bra{b}U_Q)$ are well-defined. Furthermore, it follows from the fact that $\Gamma_{\text{even}}$ is Hilbert-Schmidt orthogonal to its complement  $\Gamma_{\text{odd}} \coloneqq \bigoplus_{\ell = 1}^n \Gamma_{2n-1}$ that these classical shadows produce unbiased estimates of the expectation values $\tr( O_1\rho),\dots,\tr( O_M\rho)$ provided that \textit{either} the state $\rho$ is in $\Gamma_{\text{even}}$, or the observables $O_1, \dots, O_M$ are in $\Gamma_{\text{even}}$. Indeed, for many physical problems, the observables of interest are even operators, due to fermionic parity conservation. In our particular application to QC-AFQMC (see subsection~\ref{sec:background_QCQMC}, and~\ref{sec:summary_Slater} below), the starting state $\rho$ and relevant observables are all even operators or can all be made even, as shown in Appendix~\ref{app:overlaps}.

Supposing that $O \in \Gamma_{\text{even}}$ (and hence $O^\dagger \in \Gamma_{\text{even}}$), we can substitute the expression for $\mathcal{E}_{\mathrm{M}_n}^{(3)} = \mathcal{E}_{\mathrm{M}_n \cap \mathrm{Cl}_n}^{(3)}$ from Theorem~\ref{thm:moments}(iii) for $\E_{U \sim D} \mathcal{U}^{\otimes 3}$ in Eq.~\eqref{Varoi}, leading to
\begin{align} \label{variancebound} \mathrm{Var}[\hat{o}] \leq \E[|\hat{o}|^2] = \frac{1}{2^{2n}}\sum_{\substack{\ell_1,\ell_2,\ell_3 \geq 0 \\ \ell_1+ \ell_2 + \ell_3 \leq n}} \alpha_{\ell_1,\ell_2,\ell_3}\sum_{\substack{A_1, A_2, A_3 \subseteq [2n] \text{ disjoint} \\ |A_1| = 2\ell_1, |A_2| = 2\ell_2, |A_3| = 2\ell_3}} \tr\left(\rho \gamma_{A_1}\gamma_{A_2}\right) \tr\left(O\gamma_{A_2}\gamma_{A_3}\right)\tr\left(O^\dagger \gamma_{A_3}\gamma_{A_1}\right),
\end{align}
with 
\begin{equation} \label{alphadef} \alpha_{\ell_1,\ell_2,\ell_3} \coloneqq \frac{{n\choose \ell_1,\ell_2,\ell_3, n-\ell_1-\ell_2-\ell_3}}{{2n\choose 2\ell_1,2\ell_2,2\ell_3, 2(n-\ell_1-\ell_2-\ell_3)}}\frac{{2n\choose 2(\ell_1 + \ell_3)}}{{n\choose \ell_1 + \ell_3}} \frac{{2n\choose 2(\ell_2 + \ell_3)}}{{n\choose \ell_2 + \ell_3}}, \end{equation}
for the variance of the estimator $\hat{o}$ for $\tr(O\rho)$ that we obtain from (a single sample of) our matchgate shadows. Eq.~\eqref{variancebound}, proven in subsection~\ref{sec:shadows_variance}, is expressed in terms of a particular set of Majorana operators $\gamma_{\mu}$. However, note from Eq.~\eqref{EjMn} that $\mathcal{E}^{(3)}_{\mathrm{M}_n}$ is invariant under composition with any Gaussian unitary channel $\mathcal{U}_Q^{\otimes 3}$, and by Corollary~\ref{cor:3design}, so is $\mathcal{E}^{(3)}_{\mathrm{M}_n \cap\mathrm{Cl}_n}$. This symmetry can be used to show that, for both the continuous ensemble $\mathrm{M}_n$ and the discrete ensemble $\mathrm{M}_n \cap \mathrm{Cl}_n$, we can in fact replace the $\gamma_\mu$'s in Eq.~\eqref{variancebound} with \textit{any} other basis of Majorana operators, i.e., for any $Q \in \orth$, we can write
\begin{equation} \label{varianceboundtilde} \E[|\hat{o}|^2] = \frac{1}{2^{2n}}\sum_{\substack{\ell_1,\ell_2,\ell_3 \geq 0 \\ \ell_1+ \ell_2 + \ell_3 \leq n}} \alpha_{\ell_1,\ell_2,\ell_3}\sum_{\substack{A_1, A_2, A_3 \subseteq [2n] \text{ disjoint} \\ |A_1| = 2\ell_1, |A_2| = 2\ell_2, |A_3| = 2\ell_3}} \tr\left(\rho \widetilde\gamma_{A_1}\widetilde\gamma_{A_2}\right) \tr\left(O\widetilde\gamma_{A_2}\widetilde\gamma_{A_3}\right)\tr\left(O^\dagger \widetilde\gamma_{A_3}\widetilde\gamma_{A_1}\right) \end{equation}
where $\widetilde{\gamma}_\mu = \sum_{\nu = 1}^{2n} Q_{\mu\nu}\gamma_\mu$. The freedom to choose the Majorana basis in Eq.~\eqref{varianceboundtilde} (which is not \textit{a priori} obvious for $\mathrm{M}_n \cap \mathrm{Cl}_n$, without Corollary~\ref{cor:3design}) will allow us to more easily bound the variance for large classes of fermionic observables. We can also obtain a bound that does not depend on our unknown state $\rho$ from Eq.~\eqref{varianceboundtilde} by applying a triangle inequality and noting that $|\tr(\rho \widetilde\gamma_{A_1}\widetilde\gamma_{A_2})| \leq \|\widetilde\gamma_{A_1}\widetilde\gamma_{A_2}\| \leq 1$ for any $A_1, A_2 \subseteq [n]$:
\begin{align} \label{variancebound2}
\mathrm{Var}[\hat{o}] &\leq  \frac{1}{2^{2n}} \sum_{\substack{\ell_1,\ell_2,\ell_3 \geq 0 \\ \ell_1 + \ell_2 + \ell_3 \leq n}}\alpha_{\ell_1,\ell_2,\ell_3} \sum_{\substack{A_1, A_2, A_3 \subseteq [2n] \text{ disjoint} \\ |A_1| = 2\ell_1, |A_2|=2\ell_2,|A_3|=2\ell_3}} \big|\tr(O\widetilde\gamma_{A_2\cup A_3})\tr(O \widetilde\gamma_{A_3 \cup A_1})\big| \\
&= \frac{1}{2^{2n}} \sum_{\substack{S_1, S_2 \subseteq [2n] \\ |S_1|, |S_2|, |S_1 \cap S_2| \text{ even}}} \alpha_{\frac{1}{2}|S_2\setminus S_1|, \frac{1}{2}|S_1\setminus S_2|,\frac{1}{2}|S_1 \cap S_2|} \big|\tr(O \widetilde{\gamma}_{S_1}) \tr(O\widetilde{\gamma}_{S_2})\big|. \nonumber
\end{align} 

Eqs.~\eqref{Minverse} and~\eqref{variancebound2} characterise the classical shadows obtained from performing random measurements corresponding to either of our two distributions over matchgate circuits ($\mathrm{M}_n$ or $\mathrm{M}_n \cap\mathrm{Cl}_n$), but they should be viewed only as a starting point. To provide a viable protocol for estimating the expectation values of some observables $O_1,\dots, O_M$, we must also be able to 1) efficiently compute (on a classical computer) the expectation values $\tr(O_i \mathcal{M}^{-1}(U_Q^\dagger\ket{b}\bra{b}U_Q))$ of each $O_i$ with respect to any classical shadow sample $\mathcal{M}^{-1}(U_Q^\dagger\ket{b}\bra{b}U_Q)$, and 2) efficiently compute a bound on the variance for each $O_i$, so as to determine the number of samples needed to achieve a given precision. In the following subsection, we describe efficient computation schemes and analyse the variance for three general classes of observables.

\subsection{Applications} 
\label{sec:summary_applications}

\subsubsection{Local fermionic observables} \label{sec:summarylocal}

First, as a simple example, we consider observables that are products of an even number of Majorana operators, i.e., $O = \widetilde{\gamma}_S$ for $S \subseteq [2n]$ with $|S|$ even, where $\widetilde{\gamma}_\mu = \sum_{\mu = 1}^{2n} Q'_{\mu\nu}\gamma_\nu$ for some $Q'\in \orth$. The expectation value of $\widetilde{\gamma}_S$ with respect to a classical shadow sample $\mathcal{M}^{-1}(U_Q^\dagger \ket{b}\bra{b}U_Q)$ can be computed as
\begin{equation} \label{computegammaS} \tr\left(\widetilde{\gamma}_S \mathcal{M}^{-1}(U_Q^\dagger\ket{b}\bra{b}U_Q)\right) = {2n\choose |S|}{n\choose |S|/2}^{-1} \Pf\left(i (Q' Q^{\mathrm{T}} C_{\ket{b}} Q {Q'}^{\mathrm{T}})\big|_{S}\right), \end{equation}
where $\Pf$ denotes the Pfaffian, $C_{\ket{b}}$ is the covariance matrix of the computational basis state $\ket{b}$, given in Eq.~\eqref{covarianceb}, and $M\big|_S$ denotes the restriction of a matrix $M$ to rows and columns indexed by $S$. Eq.~\eqref{computegammaS} follows directly from the form of our inverse channel $\mathcal{M}^{-1}$  [Eq.~\eqref{Minverse}] together with Wick's theorem, and is efficiently computable since the Pfaffian of a $2n\times 2n$ matrix can be computed in $\mathcal{O}(n^3)$ time (see e.g., \citen{Wimmer2012-jo}). From Eq.~\eqref{varianceboundtilde} or~\eqref{variancebound2}, the variance of the estimates for $\tr(\widetilde{\gamma}_S\rho)$ is bounded by ${2n\choose |S|}{n\choose |S|/2}^{-1}$ (to see this, observe that the only non-vanishing term in the sum corresponds to $A_1 = A_2 = \varnothing$ and $A_3 = S$), which scales as $n^{|S|/2}$ for constant $|S|$. Thus, the expectation values of local fermionic observables can be efficiently estimated using the classical shadows we obtain from either $\mathrm{M}_n$ or $\mathrm{M}_n \cap \mathrm{Cl}_n$.

\subsubsection{Gaussian density matrices} \label{sec:summary_densityoperators}

Next, we consider observables that are density operators of fermionic Gaussian states, i.e., $O = \varrho$, where $\varrho$ has the form in Eq.~\eqref{varrhodef}. In the case where $\varrho$ or the unknown state $\rho$ is a pure state, the expectation value $\tr(\varrho\rho)$ of $\varrho$ with respect to $\rho$ gives the fidelity between $\rho$ and $\varrho$. From Eq.~\eqref{Minverse}, the expectation value of $\varrho$ with respect to a classical shadow sample $\mathcal{M}^{-1}(U_Q^\dagger\ket{b}\bra{b}U_Q)$, which gives an unbiased estimate for $\tr(\varrho\rho)$, is 
\begin{equation} \label{estimatedensityoperators} \tr\left(\varrho \mathcal{M}^{-1}(U_Q^\dagger\ket{b}\bra{b}U_Q) \right) = \sum_{\ell = 0}^{n} {2n\choose 2\ell}{n\choose \ell}^{-1} \tr\left( \varrho \mathcal{P}_{2\ell}(U_Q^\dagger\ket{b}\bra{b}U_Q)\right). \end{equation}
Theorem~\ref{prop:densityoperators}, stated for a special case below and in full generality in subsection~\ref{sec:compute_densityoperators}, allows us to efficiently compute $\tr(\varrho \mathcal{P}_{2\ell}(U_Q^\dagger\ket{b}\bra{b}U_Q))$ for every $\ell$, and hence $\tr(\varrho \mathcal{M}^{-1}(U_Q^\dagger\ket{b}\bra{b}U_Q) )$, for any Gaussian state $\varrho$. Note that $U_Q^\dagger\ket{b}\bra{b}U_Q$ is also a Gaussian state, for any $U_Q \in \mathrm{M}_n$ and computational basis state $\ket{b}$.
\begingroup
\def\thetheorem{2*}
\begin{theorem}[specialised to invertible $C_{\varrho_1}$] \makeatletter\def\@currentlabel{1*}\makeatother\label{prop:densityoperators_short}
For any $n \in \mathbb{Z}_{>0}$, let $\varrho_1$ and $\varrho_2$ be density operators of $n$-mode fermionic Gaussian states $\mathrm{(\text{Eq.}~\eqref{varrhodef})}$, with covariance matrices $C_{\varrho_1}$ and $C_{\varrho_2}$ $\mathrm{(\text{Eq.}~\eqref{covariancematrix})}$. Then, for each $\ell \in \{0,\dots, n\}$, $\tr(\varrho_1\mathcal{P}_{2\ell}(\varrho_2))$ is the coefficient of $z^\ell$ in the polynomial $p_{\varrho_1,\varrho_2}(z)$, where
\begin{equation} \label{pinvertible} p_{\varrho_1,\varrho_2}(z) = \frac{1}{2^n} \Pf\big(C_{\varrho_1}\big) \Pf\left( -C_{\varrho_1}^{-1} + zC_{\varrho_2}\right)  \end{equation}
if $C_{\varrho_1}$ is invertible.
\end{theorem}
\endgroup

We give the form of $p_{\rho_1,\rho_2}(z)$ for the general case, where $C_{\varrho_1}$ is not necessarily invertible, in subsection~\ref{sec:compute_densityoperators}, where we also provide the proof. 
This polynomial has degree at most $r$ in general, where $2r \leq 2n$ is the rank of $C_{\varrho_1}$, so its coefficients could be computed using polynomial interpolation\footnote{As an alternative to polynomial interpolation, we note that there are division-free algorithms for efficiently computing the Pfaffian over any commutative ring (see e.g., Refs.~\citenum{mahajan,Rote2001}).} in $\mathcal{O}(r^4)$ time. 
In Appendix~\ref{app:linearPfaffian}, we describe a different strategy that computes all of the coefficients in $\mathcal{O}(r^3)$ time.
Thus, by taking $\varrho_1 = \varrho$ and $\varrho_2 = U_Q^\dagger\ket{b}\bra{b}U_Q$ in Theorem~\ref{prop:densityoperators}, we can compute our classical shadows estimate, Eq.~\eqref{estimatedensityoperators}, in at most $\mathcal{O}(n^3)$ time.  For instance, in the case where $C_\varrho$ is invertible, the terms $\tr(\varrho \mathcal{P}_{2\ell}(U_Q^\dagger\ket{b}\bra{b}U_Q))$ in Eq.~\eqref{estimatedensityoperators} are the coefficients of the polynomial $2^{-n}\Pf(C_{\varrho})\Pf(-C_\varrho^{-1} + z Q^{\mathrm{T}} C_{\ket{b}}Q)$, with $C_{\varrho}$ given as in Eq.~\eqref{covariancegeneral} and $C_{\ket{b}}$  in Eq.~\eqref{covarianceb}.

A result that is essentially a special case of Theorem~\ref{prop:densityoperators} (where $\varrho_1$ is a computational basis state) was derived in \citen{Helsen2022}, by combining Wick's theorem with a minor summation formula for Pfaffians. We use a different, more elementary approach (directly exploiting the structure of the Clifford algebra generated by the Majorana operators), to prove Theorem~\ref{prop:densityoperators}, because it can be extended to the more complicated case of estimating overlaps, discussed in the following subsection (whereas the proof strategy of \citen{Helsen2022} would require some other summation formula, tailored to the overlap case; as far as we are aware, such a formula is not available in the literature.\footnote{As a side note, it may be of interest to observe that going the other way, Theorem~\ref{prop:overlaps} could be used to derive a new minor summation formula, by expanding the quantity on the left-hand side na\"ively using Wick's theorem, and equating the result with the right-hand side. Likewise, the more general method for computing free-fermion quantities presented in subsection~\ref{sec:compute_general} could be used to extract an infinity of minor summation formulae, by expanding the quantity na\"ively using Wick's theorem and/or Eq.~\eqref{UQgammaS}, and equating the result with the efficiently computable expression given by the method.}

As for the variance, we show in subsection~\ref{sec:variance_densityoperators} that for any Gaussian state $\varrho$, Eq.~\eqref{variancebound2} becomes
\begin{equation} \label{variancedensityoperators} \mathrm{Var}[\hat{o}]\Big|_{O = \varrho} \leq \frac{1}{2^{2n}} \sum_{\substack{\ell_1,\ell_2,\ell_3 \geq 0\\\ell_1 + \ell_2 + \ell_3 \leq n}} \frac{{n\choose \ell_1,\ell_2,\ell_3,n-\ell_1-\ell_2-\ell_3}^2}{{2n\choose 2\ell_1,2\ell_2,2\ell_3, 2(n-\ell_1-\ell_2-\ell_3)}}\frac{{2n\choose 2(\ell_1+\ell_3)}}{{n\choose \ell_1 + \ell_3}}\frac{{2n\choose 2(\ell_2 + \ell_3)}}{{n\choose \ell_2 + \ell_3}}  \end{equation}
for the variance of the classical shadows estimator for the expectation value of $\varrho$. This can be straightforwardly upper bounded by $\mathcal{O}(n^3)$; we provide a more refined argument in Appendix~\ref{app:variancebound} that $\Var[\hat{o}]\big|_{O = \hat{\rho}} = \mathcal{O}(\sqrt{n} \log n)$. The RHS of Eq.~\eqref{sec:variance_densityoperators} is also plotted as the red line in Figure~\ref{fig:bnzetaplot}.

\subsubsection{Overlaps with Slater determinants} \label{sec:summary_Slater}

We can also efficiently estimate the overlap between a pure state $\ket{\psi}$ and an arbitrary Slater determinant $\ket{\varphi}$ using our matchgate shadows, provided that we can prepare $\ket{\psi}$ using a quantum circuit.\footnote{The procedure described here and in subsection~\ref{sec:background_QCQMC} involves preparing the superposition state $\frac{1}{\sqrt{2}}(\ket{\mathbf{0}} + \ket{\psi})$. In ~\citen{Huggins2022}, where $\ket{\psi}$ is the trial state $\ket{\Psi_{\text{trial}}}$ of the quantum Monte Carlo algorithm, $\frac{1}{\sqrt{2}}(\ket{\mathbf{0}} + \ket{\Psi_{\text{trial}}})$ is prepared by first preparing $\frac{1}{\sqrt{2}}(\ket{\mathbf{0}} + \ket{x})$ where $\ket{x} \neq \ket{\mathbf{0}}$ is a computational basis state, then applying a fermion number-conserving unitary which maps $\ket{x}$ to $\ket{\Psi_{\text{trial}}}$ and leaves $\ket{\mathbf{0}}$ invariant. However, more generally, we show in Appendix~\ref{app:overlaps} how overlap estimation can be performed given only a (controlled) quantum circuit for preparing $\ket{\psi}$ (rather than requiring the ability to prepare $\frac{1}{\sqrt{2}}(\ket{\mathbf{0}} + \ket{\psi})$).} As explained in subsection~\ref{sec:background_QCQMC}, assuming $\ket{\psi}$ has no support on the vacuum state $\ket{\mathbf{0}}$ (see Appendix~\ref{app:overlaps} for modified protocols that remove this assumption), the overlap $\braket{\psi}{\varphi}$ can be obtained by evaluating the expectation value of $\ket{\varphi}\bra{\mathbf{0}}$ with respect to the initial state $\rho = \frac{1}{2}(\ket{\mathbf{0}} + \ket{\psi})(\bra{\mathbf{0}} + \bra{\psi})$. Note that $\ket{\varphi}\bra{\mathbf{0}}$ is an even operator if and only if the number of electrons $\zeta$ in $\ket{\varphi}$ is even. (One way to see this is by using the fact that an operator is in $\Gamma_{\text{even}}$ if and only if it commutes with the parity operator $P = \prod_{j = 1}^n (I - 2a_j^\dagger a_j) = (-i)^n \gamma_1 \dots \gamma_{2n}$.) Since our analysis in subsection~\ref{sec:summary_shadows} applies directly only to even operators, we show in Appendix~\ref{app:overlaps_Slater} that in the case where $\ket{\varphi}$ has an odd number of fermions, we can reduce the problem of evaluating $\braket{\varphi}{\psi}$ to evaluating the expectation value of $\ket{\varphi'}\bra{\mathbf{0}}$ for a Slater determinant $\ket{\varphi'}$ with an even number of electrons, by introducing an extra qubit and making a simple modification to the initial state $\rho$. 

Hence, it suffices to show how to estimate the expectation value of $\ket{\varphi}\bra{\mathbf{0}} \in \Gamma_{\text{even}}$ for an arbitrary Slater determinant $\ket{\varphi}$ with an even number of fermions $\zeta$. 
Taking $\ket{\varphi}\bra{\mathbf{0}}$ to be an observable in the context of the classical shadows protocol,
its expectation values
with respect to classical shadow samples of any initial state $\rho$ are unbiased estimates of $\tr(\ket{\varphi}\bra{\mathbf{0}}\rho)$. By Eq.~\eqref{Minverse}, for our matchgate shadows, these estimates have the form
\begin{equation} \label{estimateketphibra0}
\tr\left(\ket{\varphi}\bra{\mathbf{0}} \mathcal{M}^{-1}(U_Q^\dagger\ket{b}\bra{b}U_Q)\right) = \sum_{\ell=0}^n {2n\choose 2\ell}{n\choose \ell}^{-1} \tr\left(\ket{\varphi}\bra{\mathbf{0}}\mathcal{P}_{2\ell}(U_Q^\dagger\ket{b}\bra{b}U_Q) \right)
\end{equation}
where $U_Q \in \mathrm{M}_n$ and $\ket{b}$ is a computational basis state. Eq.~\eqref{estimateketphibra0} can be efficiently computed using the following result, which we prove in subsection~\ref{sec:compute_overlaps}.
\begin{theorem} \label{prop:overlaps}
For any $n \in \mathbb{Z}_{>0}$ and even integer $0 \leq \zeta \leq n$, let $\ket{\varphi}$ be an $n$-mode, $\zeta$-fermion Slater determinant specified as in Eq.~\eqref{ajtildeSlater}. Let $\varrho$ be the density operator of any $n$-mode  fermionic Gaussian state $\mathrm{(}\text{Eq.~\eqref{varrhodef}}\mathrm{)}$, with covariance matrix $C_{\varrho}$ $\mathrm{(}\text{Eq.~\eqref{covariancegeneral}}\mathrm{)}$. Then, for each $\ell \in \{0,\dots, n\}$, $\tr(\ket{\varphi}\bra{\mathbf{0}} \mathcal{P}_{2\ell}(\varrho))$ is the coefficient of $z^\ell$ in the polynomial
\[ q_{\ket{\varphi},\varrho}(z) = \frac{1}{2^{n - \zeta/2}} i^{\zeta/2} \Pf\left(\left(C_{\ket{\mathbf{0}}} + zW^* \widetilde{Q} C_{\varrho} \widetilde{Q}^{\mathrm{T}} W^\dagger\right) \Big|_{\overline{S}_\zeta}  \right),  \]
where $C_{\ket{\mathbf{0}}}$ is the covariance matrix of the vacuum state $\ket{\mathbf{0}}$ $\mathrm{(}Eq.~\eqref{covarianceb}\mathrm{)}$,  $\widetilde{Q}$ is the orthogonal matrix defined in Eq.~\eqref{QSlater},\footnote{If only the first $\zeta$ rows of the matrix $V$ in Eq.~\eqref{ajtildeSlater} are specified, one can choose the remaining entries such that the resulting matrix is unitary, to find $\widetilde{Q}$.} 
\begin{equation} \label{Wdef} W \coloneqq \bigoplus_{j = 1}^\zeta \frac{1}{\sqrt{2}}\begin{pmatrix} 1 &-i \\ 1 &i \end{pmatrix} \oplus \bigoplus_{j = \zeta + 1}^n \begin{pmatrix} 1 & 0 \\ 0 &1 \end{pmatrix}, \end{equation}
and $\overline{S}_\zeta \coloneqq [2n] \setminus \{1,3,\dots,2\zeta - 1\}$. 
\end{theorem} 

The matrix $(C_{\ket{\mathbf{0}}} + zW^* \widetilde{Q} C_{\varrho} \widetilde{Q}^{\mathrm{T}} W^\dagger) \big|_{\overline{S}_\zeta}$ has size $2n - \zeta$, so the
polynomial $q_{\ket{\varphi}, \varrho}(z)$ has degree at most $n - \zeta/2$ and all of its coefficients can be found using polynomial interpolation in $\mathcal{O}((n - \zeta/2)^4)$ time. Thus, taking $\varrho = U_Q^\dagger \ket{b}\bra{b}U_Q$ in Theorem~\ref{prop:overlaps} gives an efficient way of computing our classical shadows estimate, Eq.~\eqref{estimateketphibra0}. 

As we prove in subsection~\ref{sec:variance_overlaps}, for any initial state $\rho$ and Slater determinant $\ket{\varphi}$ with an even number of $\zeta$ fermions, the variance of our classical shadows estimator for $\tr(\ket{\varphi}\bra{\mathbf{0}} \rho)$ is bounded as 
\begin{equation} \label{varianceoverlaps} \Var[\hat{o}]\Big|_{O = \ket{\varphi}\bra{\mathbf{0}}} \leq b(n, \zeta) \coloneqq \frac{1}{2^{2n}} \sum_{\substack{\ell_1,\ell_2,\ell_3 \geq 0\\ \ell_1 + \ell_2 + \ell_3 \leq n}}\alpha_{\ell_1,\ell_2,\ell_3}\, \kappa(n, \zeta, \ell_1,\ell_2,\ell_3),  \end{equation}
where $\alpha_{\ell_1,\ell_2,\ell_3}$ is given by Eq.~\eqref{alphadef} and\footnote{We adopt the convention that ${n\choose k_1,k_2,\dots} = 0$ whenever any $k_i < 0$. Note that the upper limit of the sum in Eq.~\eqref{kappa} can be replaced by $\min\{\zeta/2,\ell_3, n-\ell_1-\ell_2-\ell_3\}$.}
\begin{equation} \label{kappa} \kappa(n,\zeta, \ell_1,\ell_2,\ell_3) \coloneqq 2^{\zeta}\sum_{j=0}^{\zeta/2} {\zeta\choose 2j}{n-\zeta \choose \ell_1 - \zeta/2 +j,\, \ell_2 - \zeta/2 + j, \, \ell_3 - j, \,n - \ell_1 - \ell_2 - \ell_3 -j}. \end{equation}
Note that for $\zeta = 0$, the RHS of Eq.~\eqref{varianceoverlaps} reduces to the RHS of Eq.~\eqref{variancedensityoperators}, which is our variance bound for estimating expectation values of Gaussian density operators. This is consistent, because $\ket{\varphi} = \ket{\mathbf{0}}$ if $\zeta = 0$, so $\ket{\varphi}\bra{\mathbf{0}} = \ket{\mathbf{0}}\bra{\mathbf{0}}$ is a Gaussian density operator.
While we do not provide an asymptotic bound on the variance for $\zeta >0$, Eq.~\eqref{varianceoverlaps} is an explicit upper bound that can be computed in $\mathrm{poly}(n)$ time (and in order to use the classical shadows procedure, it suffices to be able to efficiently compute an upper bound on the variance, in order to choose the number of samples). We plot the bound $b(n,\zeta)$ in Eq.~\eqref{varianceoverlaps} in Fig.~\ref{fig:bnzetaplot} for $n$ up to $1000$ and various values of $\zeta$. The plot strongly suggests that the variance bound for $\zeta > 0$ is always smaller than the bound for $\zeta = 0$, which scales as $\mathcal{O}(\sqrt{n} \log n)$ (see~Appendix~\ref{app:variancebound}). This would imply that the number of matchgate shadow samples required to estimate overlaps with arbitrary Slater determinants is sublinear in the number of fermionic modes $n$. 

\begin{figure*}
    {\centering
    \includegraphics[width=0.97\textwidth]{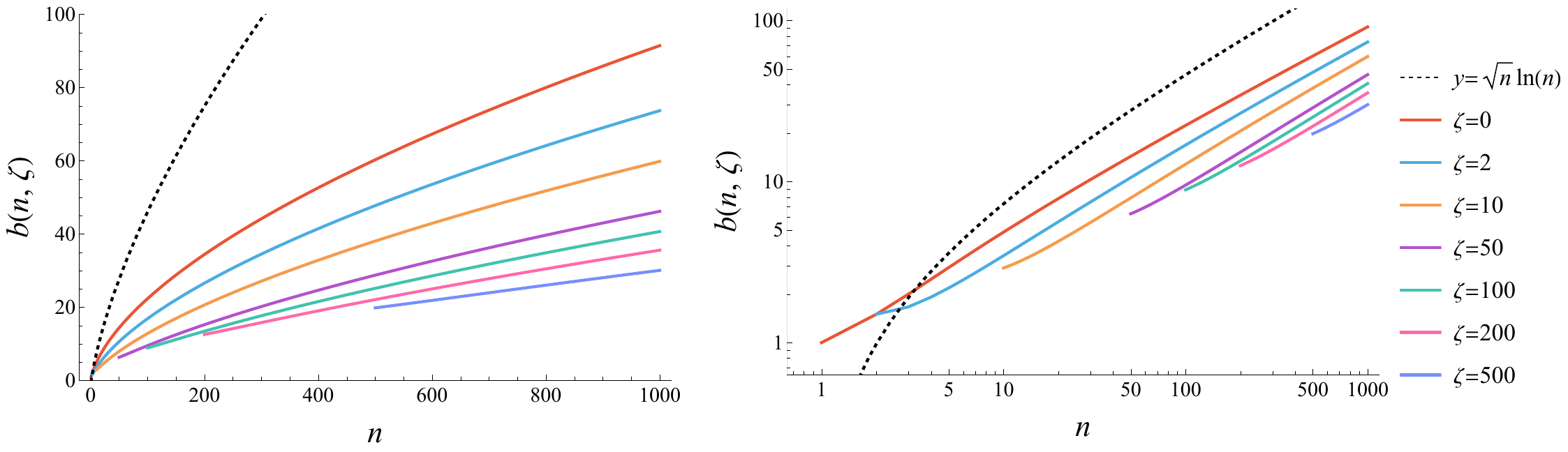}}
    \caption{Linear and log-log plots of $b(n,\zeta)$ vs.\ $n$, for $\zeta \in \{0,2,10,50,100,200,500\}$ (evaluated at integer values of $n$ and joined). We also plot $y = \sqrt{n} \ln(n)$ for comparison (in Appendix~\ref{app:variancebound}, we place an $\mathcal{O}(\sqrt{n}\ln(n))$ bound on $b(n,0)$, plotted in red). $b(n,\zeta)$, defined in Eq.~\eqref{sec:variance_overlaps}, is our bound on the variance for estimating the expectation value of $\ket{\varphi}\bra{\mathbf{0}}$ using our matchgate classical shadows, where $\ket{\varphi}$ is any $n$-mode, $\zeta$-fermion Slater determinant. As shown in Appendix~\ref{app:overlaps_Slater}, estimating the expectation value of $\ket{\varphi}\bra{\mathbf{0}}$ allows us to estimate the overlap between any pure state and $\ket{\varphi}$. Note that $b(n,0)$ is also equal to the RHS of Eq.~\eqref{variancedensityoperators}, which is our variance bound for estimating the expectation values of arbitrary Gaussian density operators.}
    \label{fig:bnzetaplot}
\end{figure*}

\begin{myalgorithmfloat}[tb]
\begin{mdframed}[roundcorner=3pt]
{\small
\begin{center}
\vspace{1em}
\textsc{Algorithm 1: Estimating overlaps with Slater determinants via matchgate shadows}
\end{center}
\begin{tabbing}
\textbf{Inputs:} \=$n$-qubit quantum state $\ket{\psi}$ (accessed via a state preparation circuit) \\
\> $n$-mode Slater determinants $\ket{\varphi_1}, \dots, \ket{\varphi_M}$ (specified as in Eq.~\eqref{ajtildeSlater}) with $\zeta_1,\dots, \zeta_M$ electrons, respectively \\
\> additive error $\varepsilon > 0$; failure probability $\delta > 0$  
\end{tabbing}
\noindent \textbf{Output:} estimates $\widetilde{o}_1,\dots,\widetilde{o}_M$ such that with probability at least $1-\delta$, $|\widetilde{o}_i - \braket{\psi}{\varphi_i}| \leq \varepsilon$ for all $i \in [M]$

\medskip
\noindent\textbf{Procedure:}

\smallskip

\noindent \underline{Collect matchgate shadow samples}

\smallskip

0. Choose $D$ = uniform distribution over $\mathrm{M}_n$ or $D =$ uniform distribution over $\mathrm{M}_n \cap \mathrm{Cl}_n$

1. $b_{\max} \leftarrow \max\limits_{i \in [M]} b(n,\zeta_i)$, where $b(n,\zeta)$ is the variance bound in Eq.~\eqref{varianceoverlaps} 

2. $K \leftarrow \left\lceil \frac{9}{2} \ln\left(\frac{M}{\delta}\right)\right\rceil$, $L \leftarrow \left\lceil 24\frac{b_{\max}}{\varepsilon^2}\right\rceil$, $N_{\mathrm{sample}} \leftarrow KL$ \hfill \textit{// $K$ subsamples of size $L$ for median-of-means estimator}

\smallskip

3. For $j = 1, \dots, N_{\mathrm{sample}}$, \hfill \textit{// obtain $N_{\mathrm{sample}}$ matchgate shadow samples}

\qquad {\color{mypurple2}\boxed{{\color{black}\text{a.}}}} prepare $\ket{\rho} = \frac{1}{\sqrt{2}}(\ket{\mathbf{0}} + \ket{\psi})$ (see Appendix~\ref{app:overlaps} for alternatives) 

\qquad {\color{white}\boxed{{\color{black}\text{b.}}}} sample a unitary $U_{Q^{(j)}}$ from $D$, by uniformly sampling $Q^{(j)}$ from $\orth$ if $D = \mathrm{M}_n$, or from $\mathrm{B}(2n)$ if $D = \mathrm{M}_n \cap \mathrm{Cl}_n$

\qquad \quad\,\,\, (see Appendix~\ref{app:random_circuit_sampling}) 

\qquad {\color{mypurple2}\boxed{{\color{black}\text{c.}}}} apply $U_{Q^{(j)}}$ to $\ket{\rho}$

\qquad {\color{mypurple2}\boxed{{\color{black}\text{d.}}}} measure in the computational basis, obtaining outcome $b^{(j)} \in \{0,1\}^n$

\qquad {\color{white}\boxed{{\color{black}\text{e.}}}} store $(Q^{(j)}, b^{(j)})$

\medskip
\underline{Compute overlap estimates}

\smallskip

4. For $i = 1, \dots, M$,

\qquad a. For $j = 1,\dots, N_{\mathrm{sample}}$, 

\qquad \qquad i. compute the coefficients of the polynomial
\[ q_{\ket{\varphi_i}, U_Q^{(j)\dagger}\ket{b^{(j)}}}(z) \coloneqq  \frac{1}{2^{n - \zeta_i/2}} i^{\zeta_i/2} \Pf\left(\left(C_{\ket{\mathbf{0}}} + zW^* \widetilde{Q} Q^{(j)T} C_{\ket{b^{(j)}}} Q^{(j)} \widetilde{Q}^{\mathrm{T}} W^\dagger\right) \Big|_{\overline{S}_\zeta}  \right) \]

\qquad \qquad \,\,\,\,\, from Theorem~\ref{prop:overlaps}; \enspace set \enspace $c_\ell \leftarrow $ coefficient of $z^\ell$ in $q_{\ket{\varphi_i}, U_Q^{(j)\dagger}\ket{b^{(j)}}}(z)$

\qquad \qquad ii. $o_i^{(j)} \leftarrow 2\sum\limits_{\ell=0}^n {2n\choose 2\ell} {n\choose \ell}^{-1} c_\ell$ \hfill \textit{// a single unbiased estimate of $\braket{\psi}{\varphi_i}$}

\qquad b. $\widetilde{o}_i = \text{median}(\overline{o}_{i,1}, \dots \overline{o}_{i,K})$, where $\overline{o}_{i,k} = \sum\limits_{j = (k-1)L + 1}^{kL} o_{i}^{(j)}/L$ \hfill \textit{// median-of-means estimate of $\braket{\psi}{\varphi_i}$}

\noindent 5. Return $\widetilde{o}_1, \dots \widetilde{o}_M$
}
\vspace{1em}
\end{mdframed}
\parbox{0.98\linewidth}{
\caption{Our matchgate shadows protocol applied to estimating the overlaps between a pure state $\ket{\psi}$ and $M$ Slater determinants $\ket{\varphi_1}, \dots, \ket{\varphi_M}$. (See subsection~\ref{sec:summary_applications} for examples of other observables that can be efficiently estimated using our shadows.) The procedure is stated for the case where the Slater determinants have even, nonzero numbers of electrons for simplicity, but can be easily adapted to arbitrary Slater determinants using the modification discussed in Appendix~\ref{app:overlaps}. Note that the collection of shadow samples (Steps 0-3) can be performed without knowing the Slater determinants (the number of samples needed depends only on $n$ and the $\zeta_i$), and that only Steps 3a, 3c, and 3d (boxed) are performed on a quantum computer, while the remaining steps are classical. This is a useful feature for application to, for instance, the QC-AFQMC technique of \citen{Huggins2022} (see subsection~\ref{sec:background_QCQMC} for a review), where the number of electrons is fixed but the relevant Slater determinants are not known \textit{a priori} and are instead determined over the course of the Monte Carlo procedure.}\label{alg:overlaps}}
\end{myalgorithmfloat}

For reference, we summarise our matchgate shadows protocol applied to estimating overlaps with Slater determinants in Algorithm~\ref{alg:overlaps}. Using this protocol (with $\ket{\psi} = \ket{\Psi_{\text{trial}}}$) in place of the Clifford-based shadows protocol implemented in \citen{Huggins2022} removes the exponential classical post-processing cost incurred in QC-AFQMC. 

\subsubsection{More general fermionic observables} \label{sec:summary_computegeneral}

While the first three types of applications discussed above likely cover many cases of interest, we also develop an explicit framework for efficiently evaluating expectation values of a much broader class of observables using our matchgate classical shadows.
This broader class includes products of operators of the form $A^{(1)}\dots A^{(m)}$, where each $A^{(i)}$ is an arbitrary linear combination of Majorana operators $\{\gamma_\mu\}_{\mu\in [2n]}$, a fermionic Gaussian unitary, or the density operator of a fermionic Gaussian state. 
As a specific application of this general framework, we show how to use it to estimate the inner product between an arbitrary pure state \(\ket{\psi}\) and an arbitrary pure fermionic Gaussian state (not restricted to be a Slater determinant), thereby extending the results in subsection~\ref{sec:summary_Slater}.
We describe the framework in detail in subsection~\ref{sec:compute_general}, which concludes with our procedure for inner product estimation in subsection~\ref{sec:workedexample}. 

Our post-processing procedures are built upon new \emph{classical} simulation results and proof techniques that may find application in other contexts, beyond their use in the specific classical shadows protocols we consider here. In particular, we present a method for efficiently evaluating any expression that can be written in the form $\tr(A^{(1)}\dots A^{(m)})$, which encompasses a wide range of free-fermion quantities of interest. At a high level, the general method consists of three main steps. First, we give a general recipe for recasting the trace of a product of arbitrary operators as a Grassmann integral (see Theorem~\ref{prop:tr(ABCD)}), in a Grassmann algebra that is related to the Clifford algebra generated by the Majorana operators. Then, we show that if each operator in the product falls into one of the three categories described above, the Grassmann integral can be massaged into a particular form. Finally, we develop an algorithm (Algorithm~\ref{alg:gMB}) for efficiently evaluating any integral of this form. 

\subsection{Comparison to related work}

\subsubsection{Prior work} \label{sec: comparison prior}

We now compare our results to those of Zhao \textit{et al.}~\cite{Zhao2021-fv}, which considers classical shadows resulting from the discrete uniform distribution over matchgate circuits $U_Q$ such that $Q$ is in the alternating group $\mathrm{A}(2n)$; these constitute a proper subset of $\mathrm{M}_n \cap \mathrm{Cl}_n$. We note that the measurement channel $\mathcal{M}$ for this distribution is the same as ours in Eq.~\eqref{matchgateshadowschannel}, even though its corresponding $2$-fold twirl $\E_{Q \in \mathrm{A}(2n)} \mU_Q^{\otimes 2}$ is different from the 2-fold twirl $\mathcal{E}_{\mathrm{M}_n}^{(2)} = \mathcal{E}_{\mathrm{M}_n \cap \mathrm{Cl}_n}^{(2)}$ for our distributions. 
(In fact, the $j$-fold twirl channels differ, i.e., $\E_{Q \in \mathrm{A}(2n)} \mU_Q^{\otimes j} \neq \mathcal{E}_{\mathrm{M}_n}^{(j)} = \mathcal{E}_{\mathrm{M}_n \cap \mathrm{Cl}_n}^{(j)}$, for all $j \in \mathbb{Z}_{>0}$.\footnote{To see this, note that the $1$-fold twirls differ: e.g., $\E_{Q\in\mathrm{A}(2n)}\mU_Q\kett{\gamma_{[2n]}} = \kett{\gamma_{[2n]}}$ by Eq.~\eqref{UQgammaS}, whereas $\mathcal{E}^{(1)}_{\mathrm{M}_n}\kett{\gamma_{[2n]}} = 0$ from Theorem~\ref{thm:moments}(i). This implies that $j$-fold twirls differ for all $j$, since the $j$-fold twirl is fully determined by the $(j+1)$-fold twirl, as shown in Eq.~\eqref{EjtoEj-1}.}) Since the measurement channels are the same, the expectation values of $\widetilde{\gamma}_S$ can likewise be computed using Eq.~\eqref{computegammaS} for the distribution in ~\citen{Zhao2021-fv}. However, the authors only consider and bound the variance for products $\gamma_S$ of the canonical Majorana operators $\gamma_\mu$ (also finding a variance bound of ${2n\choose |S|}{n\choose |S|/2}^{-1}$). In the absence of an analogue of Corollary~\ref{cor:3design} for their discrete distribution (which naturally picks out $\{\gamma_\mu\}_{\mu \in [2n]}$ as a preferred basis), the variance for arbitrary Majorana products $\widetilde{\gamma}_S$ is more difficult to (tightly) bound, using their basis-dependent expression for the variance.%\footnote{\citen{Zhao2021-fv} does give a bound for the variance of arbitrary operators, in Equation~65 in the Supplemental Material of Ref.~\citenum{Zhao2021-fv}. However, this equation is incorrect (and underestimates the variance in general), due to a mistake in their Equation~63.}

More importantly, compared to \citen{Zhao2021-fv}, we provide methods for efficiently computing estimates of the expectation values of more families of observables, beyond single products of Majorana operators, as discussed in subsection~\ref{sec:summary_applications}. In particular, one of these is the set of $\ket{\varphi}\bra{\mathbf{0}}$ operators which allow us to obtain the overlap estimates in the QC-AFQMC algorithm (see subsection~\ref{sec:background_QCQMC}). Thus, we obtain a more generally applicable shadows protocol for estimating fermionic observables. 

\subsubsection{Subsequent work} 

Shortly after the preprint of this manuscript was published, two related papers, Refs.~\cite{ogorman2022fermionic} and~\cite{low2022classical}, were posted.

O'Gorman~\cite{ogorman2022fermionic} considers the same problem as we do in subsection~\ref{sec:summary_Slater}, that is, of estimating the overlaps between an unknown pure state and arbitrary Slater determinants  (likewise motivated by the application to QC-AFQMC), using classical shadows associated with the same discrete distribution analysed by Zhao \text{et al.} in Ref.~\citenum{Zhao2021-fv}. Hence, as discussed above, the measurement channel $\mathcal{M}$ for this distribution is the same as ours (Eq.~\eqref{matchgateshadowschannel}). However, without having proven an analogue of our ``matchgate 3-design'' result (Corollary~\ref{cor:3design}), the proof of their variance bounds is incomplete. Indeed, there is a gap between Lemma~2 and Theorem~4 of Ref.~\citenum{ogorman2022fermionic}, as it is not proved that the variance bound for $\ket{\mathbf{0}}\bra{\mathbf{0}}$ also applies to any arbitrary fermionic Gaussian state, nor that the variance bound for $\ket{1}^{\otimes \zeta} \ket{0}^{\otimes n- \zeta}\bra{\mathbf{0}}$ also applies to $\ket{\varphi}\bra{\mathbf{0}}$ for any $\zeta$-fermion Slater determinant. In our paper, the variance analysis (in the case of the discrete distribution over $\mathrm{M}_n \cap \mathrm{Cl}_n$) for arbitrary Gaussian states and Slater determinants relies on the matchgate 3-design result. This is similar to how the fact that the Clifford group forms a unitary 3-design is a key ingredient in the analysis of the Clifford-based classical shadows of Ref.~\citenum{Huang2020}. Ref.~\citenum{ogorman2022fermionic} places a bound of $\mathrm{Var}[\hat{o}]\big|_{O = \ket{\mathbf{0}}\bra{\mathbf{0}}} = \mathcal{O}(n)$ on the variance for $\ket{\mathbf{0}}\bra{\mathbf{0}}$, which is consistent with (though looser than) our bound of $\mathrm{Var}[\hat{o}]\big|_{O = \varrho} = \mathcal{O}(\sqrt{n}\log n)$ for an arbitrary Gaussian density matrix $\varrho$, and also presents numerical evidence that $\mathrm{Var}[\hat{o}]\big|_{O = \ket{\varphi}\bra{\mathbf{0}}}$ scales sublinearly with $n$ for Slater determinants $\ket{\varphi}$. In addition, Ref.~\citenum{ogorman2022fermionic} applies the variance bounds from Ref.~\citenum{Zhao2021-fv} to show how to learn a Slater determinant from copies thereof. 

For the classical post-processing required to extract estimates of $\tr(\ket{\varphi}\bra{\mathbf{0}}\rho)$ from the classical shadow samples, Ref.~\citenum{ogorman2022fermionic} shows that $\tr(\ket{\varphi}\bra{\mathbf{0}}\mathcal{M}^{-1}(U_Q^\dagger\ket{b}\bra{b}U_Q))$ can be decomposed into $n + 1$ matchgate tensor networks, then appeals to the fact that certain classes of such tensor networks can be contracted efficiently (see references therein). In contrast, we give an explicit expression for this quantity (Eq.~\eqref{estimateketphibra0} and Theorem~\ref{prop:overlaps}), and provide a self-contained algorithm for efficiently computing more general fermionic observables as well (see subsections~\ref{sec:summary_computegeneral} and~\ref{sec:compute_general}).

Low~\cite{low2022classical} considers the uniform distribution over number-conserving fermionic Gaussian unitaries, i.e., $\{U \in \mathrm{M}_n: [U, \sum_j a_j^\dagger a_j] = 0\}$, and shows how the classical shadows corresponding to this distribution can be used to estimate $k$-fermion reduced density matrices ($k$-RDMs) $\tr(a_{p_1}^\dagger \dots a_{p_k}^\dagger a_{q_1}\dots a_{q_k} \rho_\zeta)$ of a state $\rho_\zeta$ with fixed particle number $\zeta$, with an \textit{average-case} variance that is asymptotically better than the \textit{worst-case} variance resulting from the shadows considered in the present paper and Ref.~\citenum{Zhao2021-fv}. In particular, for $k = \mathcal{O}(1)$, the variance averaged over all $k$-RDMs is $\mathcal{O}(\zeta^k)$ for the shadows of Ref.~\citenum{low2022classical} (which can be much smaller than the variance bound of $\mathcal{O}(n^k)$, derived in Ref.~\citenum{Zhao2021-fv} and subsection~\ref{sec:summarylocal}), while for $k = \zeta$, the average variance is $\mathcal{O}(1)$. 

Ref.~\citenum{low2022classical} then reduces the estimation of the overlap with an arbitrary $n$-mode, $\zeta$-fermion Slater determinant to the estimation of a $\zeta$-RDM of a $(n + \zeta)$-mode state with $\zeta$ fermions, for which the ``average'' variance is $\mathcal{O}(1)$. However, there is an important subtlety: this average variance is taken over \textit{all} of the $\zeta$-RDMs, whereas only a subset of the $\zeta$-RDMs correspond to the estimation of a Slater determinant overlap. Therefore, the fact that this average variance is $\mathcal{O}(1)$ does \emph{not} imply that the variance for estimating Slater determinant overlaps is $\mathcal{O}(1)$ when averaged over $\zeta$-fermion Slater determinants. Without further analysis, it remains unclear what the worst-case or average-case variance would be for overlap estimation using these classical shadows. On the other hand, Eq.~\eqref{varianceoverlaps} provides a guarantee on the worst-case variance for overlap estimation using our matchgate shadows, though this bound is likely not independent of $n$. (Moreover, the protocol of Ref.~\citenum{low2022classical} involves adding $\zeta$ ancillary fermionic modes, which may be prohibitive for near-term quantum computers, especially when $\zeta$ is comparable to $n$---e.g., for systems at half-filling.)

For the classical post-processing, Ref.~\citenum{low2022classical} employs our proof techniques in Sections~\ref{sec:compute_densityoperators} and~\ref{sec:compute_overlaps}, extending and adapting them to obtain efficiently computable expressions for the $k$-RDM estimators obtained from their classical shadows. In the same vein as Theorems~\ref{prop:densityoperators} and~\ref{prop:overlaps}, the classical post-processing procedure of Ref.~\citenum{low2022classical} involves finding the coefficients of a certain polynomial that can be evaluated in terms of Pfaffians.

\section{Ensembles of matchgate circuits} \label{sec:ensembles}

In this section, we analyse the two distributions over matchgate circuits defined in Section~\ref{sec:summary} (see Eqs.~\eqref{Mn}--\eqref{EjMnCln}). We begin by proving Theorem~\ref{thm:moments} in subsection~\ref{sec:moments}, then use it in 
subsection~\ref{sec:shadows} to characterise the classical shadows resulting from the distributions. 

\subsection{Moments of the distributions} \label{sec:moments}

We prove Theorem~\ref{thm:moments} by explicitly evaluating the twirl channels $\mathcal{E}^{(j)}_{\mathrm{M}_n}$ and $\mathcal{E}^{(j)}_{\mathrm{M}_n \cap \mathrm{Cl}_n}$ for $j \in \{1,2,3\}$. For convenience, we first collect some basic facts about the Gaussian unitary channels $\mathcal{U}_Q$.

\begin{fact} \label{fact:mUQ} Let $\mathcal{U}_Q \in \LLHn$ be defined by Eqs.~\eqref{mathcalUQ} and~\eqref{UQdef}. For any $Q, Q' \in \orth$,
\begin{enumerate}[(a)]
    \item $\mathcal{U}_Q(ABC\dots) = \mathcal{U}_Q(A) \mathcal{U}_Q(B)\mathcal{U}_Q(C)\dots$ for any operators $A,B,C, \dots \in \LHn$ 
    \item $\mU_{Q Q'} = \mU_{Q'} \circ \mU_{Q}$
    \item $\mU_Q^\dagger = \mU_Q^{-1} = \mU_{Q^{\mathrm{T}}}$, where the adjoint is with respect to the Hilbert-Schmidt inner product,
    \item $\mathcal{U}_Q (\Gamma_k) = \Gamma_k$ for all $k \in \{0,\dots, 2n\}$, where $\Gamma_k$ is defined in Eq.~\eqref{Gammak}.
\end{enumerate}
\end{fact}
\begin{proof}
    (a) is a simple consequence of the unitarity of $U_Q$. Using (a) in conjunction with Eq.~\eqref{UQdef} gives $(\mU_{Q'} \circ \mU_Q) (\gamma_S) = \mU_{QQ'}(\gamma_S)$ for all $S \subseteq [2n]$, which implies (b) since $\{\gamma_S : S \subseteq [2n]\}$ spans $\LHn$. For (c), it follows from (b) that $\mU_Q \circ \mU_{Q^{\mathrm{T}}} = \mU_{Q^{\mathrm{T}} Q} = \mU_I = \mathcal{I}$, so $\mU_{Q}^{-1} = \mU_{Q^{\mathrm{T}}}$. Also, clearly $\mU^{-1}(\,\cdot\,) = U_Q(\,\cdot\,)U_Q^\dagger$. Hence, for all $A,B \in \LHn$,
    \begin{align*}
        \tr(A^\dagger \mU_Q(B)) = \tr(U_Q A^\dagger U_Q^\dagger B) = \tr( \mU_Q^{-1}(A)^\dagger B),
    \end{align*}
    so $\mU_Q^{-1} = \mU_Q^\dagger$. (d) follows directly from Eq.~\eqref{UQgammaS} and the fact that $\mathcal{U}_Q$ is invertible.
\end{proof}
It follows from Fact~\ref{fact:mUQ}(b) that the map $\mathcal{U} : \orth \to \LLHn$ with $\mU(Q) = \mU_Q$ is a faithful representation of the orthogonal group $\orth$. The following fact follows straightforwardly from the group properties of $\orth$ and $\mathrm{B}(2n)$. 
\begin{fact} \label{fact:Eprojector} For any $j \in \mathbb{Z}_{>0}$, $\mathcal{E}_{\mathrm{M}_n}^{(j)}$ is the orthogonal projector onto the subspace $\mathcal{X}^{(j)}_{\mathrm{M}_n} \coloneqq \{A \in \LHn^{\otimes j}: \mU_Q^{\otimes j}(A) = A\enspace  \forall \, Q \in \orth\}$ of $\LHn^{\otimes j}$, and $\mathcal{E}_{\mathrm{M}_n \cap \mathrm{Cl}_n}^{(j)}$ is the orthogonal projector onto the subspace $\mathcal{X}^{(j)}_{\mathrm{M}_n \cap \mathrm{Cl}_n} \coloneqq \{A \in \LHn^{\otimes j}: \mU_Q^{\otimes j}(A) = A\enspace  \forall \, Q \in \mathrm{B}(2n)\}$.
\end{fact}
\begin{proof} For any $Q \in \orth$,
\begin{equation} \label{Mninvariance} \mU_Q^{\otimes j} \circ \mathcal{E}_{\mathrm{M}_n}^{(j)} = \mathcal{E}_{\mathrm{M}_n}^{(j)} = \mathcal{E}_{\mathrm{M}_n}^{(j)} \circ \mU_Q^{\otimes j} \end{equation} 
using Fact~\ref{fact:mUQ}(b) in conjunction with the left- and right-invariance of the Haar measure on $\orth$. It follows that 
\[ (\mathcal{E}_{\mathrm{M}_n}^{(j)})^2 = \mathcal{E}_{\mathrm{M}_n}^{(j)}. \] 
Also, 
\[ (\mathcal{E}_{\mathrm{M}_n}^{(j)})^\dagger = \int_{\orth} d\mu(Q)\, (\mathcal{U}_Q^\dagger)^{\otimes j} = \int_{\orth} d\mu(Q)\, \mU_{Q^{-1}}^{\otimes j} = \int_{\orth} d\mu(Q)\, \mU_{Q}^{\otimes j} = \mathcal{E}_{\mathrm{M}_n}^{(j)}, \]
where the second equality is Fact~\ref{fact:mUQ}(c), while the third follows from the fact that $\orth$ is unimodular. Thus, $\mathcal{E}_{\mathrm{M}_n}^{(j)}$ is an orthogonal projector. For $A \in \mathcal{X}_{\mathrm{M}_n}$, clearly $\mathcal{E}^{(j)}_{\mathrm{M}_n}(A) = A$, so $\mathcal{X}_{\mathrm{M}_n}^{(j)} \subseteq \mathrm{im}(\mathcal{E}^{(j)}_{\mathrm{M}_n})$. Conversely, if $A \in \mathrm{im}(\mathcal{E}^{(j)}_{\mathrm{M}_n})$, then $A = \mathcal{E}^{(j)}_{\mathrm{M}_n}(A)$, so $\mU_Q(A) = (\mU_Q \circ \mathcal{E}^{(j)}_{\mathrm{M}_n})(A) = \mathcal{E}^{(j)}_{\mathrm{M}_n}(A) = A$ by Eq.~\eqref{Mninvariance}.

The proof for $\mathcal{E}^{(j)}_{\mathrm{M}_n \cap \mathrm{Cl}_n}$ is analogous, with 
\begin{equation} \label{MnClninvariance} \mU_Q^{\otimes j} \circ \mathcal{E}_{\mathrm{M}_n\cap \mathrm{Cl}_n}^{(j)} = \mathcal{E}_{\mathrm{M}_n\cap \mathrm{Cl}_n}^{(j)} = \mathcal{E}_{\mathrm{M}_n\cap \mathrm{Cl}_n}^{(j)} \circ \mU_Q^{\otimes j} \end{equation} for any $Q \in \mathrm{B}(2n)$ following from the fact that $\mathrm{B}(2n)$ is closed under left- and right-multiplication by any group element, and $(\mathcal{E}_{\mathrm{M}_n\cap \mathrm{Cl}_n}^{(j)})^\dagger = \mathcal{E}_{\mathrm{M}_n\cap \mathrm{Cl}_n}^{(j)}$ from the fact that $\mathrm{B}(2n)$ is closed under inverse. 
\end{proof}

Throughout the remainder of this section, we use $\mathcal{E}^{(j)}$ to denote \emph{either} $\mathcal{E}_{\mathrm{M}_n}^{(j)}$ or $\mathcal{E}_{\mathrm{M}_n\cap \mathrm{Cl}_n}^{(j)}$, for $j \in \{1,2,3\}$. We have not yet proven that $\mathcal{E}_{\mathrm{M}_n}^{(j)} = \mathcal{E}_{\mathrm{M}_n\cap \mathrm{Cl}_n}^{(j)}$ for $j \in \{1,2,3\}$, but this notational simplification will be valid since any statement we make while evaluating $\mathcal{E}^{(j)}$ will be patently true for both $\mathcal{E}_{\mathrm{M}_n}^{(j)}$ and $\mathcal{E}_{\mathrm{M}_n\cap \mathrm{Cl}_n}^{(j)}$. The expression we arrive at for $\mathcal{E}^{(j)}$ will therefore be equal to both $\mathcal{E}_{\mathrm{M}_n}^{(j)}$ and $\mathcal{E}_{\mathrm{M}_n\cap \mathrm{Cl}_n}^{(j)}$. Eqs.~\eqref{Mninvariance} and~\eqref{MnClninvariance} will be particularly useful; we subsume these as
\begin{equation} \label{Einvariance} \mU_Q^{\otimes j} \circ \mathcal{E}^{(j)} = \mathcal{E}^{(j)} = \mathcal{E}^{(j)} \circ \mathcal{U}_Q^{\otimes j} \end{equation}
for all $Q \in \mathrm{B}(2n)$.

We start by calculating the 2-fold twirl $\mathcal{E}^{(2)}$, which can then be used to calculate $\mathcal{E}^{(1)}$, since for any $j > 1$,
\begin{equation} \label{EjtoEj-1} \mathcal{E}^{(j-1)}(A) = \tr_1\left[\mathcal{E}^{(j)}\left(\frac{I}{2^n} \otimes A\right) \right] \end{equation}
for all $\mathcal{A} \in \LHn$.
We then sketch the proof for $\mathcal{E}^{(3)}$, deferring the more technical parts to the appendix. Note that $\mathcal{E}^{(2)}$ could be derived from $\mathcal{E}^{(3)}$ using Eq.~\eqref{EjtoEj-1}. However, we present a direct, self-contained proof for $\mathcal{E}^{(2)}$, because the proof for $\mathcal{E}^{(3)}$ uses similar ideas, but is a bit more technically involved. 

\subsubsection{The 2-fold twirl $\mathcal{E}^{(2)}$}

As discussed above, we let $\mathcal{E}^{(2)}$ denote $\mathcal{E}_{\mathrm{M}_n}^{(2)}$ or $\mathcal{E}_{\mathrm{M}_n\cap \mathrm{Cl}_n}^{(2)}$. By Fact~\ref{fact:Eprojector}, $\mathcal{E}^{(2)}$ is an orthogonal projector, so we can determine it by finding its image. We consider the basis $\{\kett{\gamma_{S_1}}\kett{\gamma_{S_2}}:S_1,S_2 \subseteq [2n]\}$ for $\LHn^{\otimes 2}$. We start with a simple lemma that uses symmetry to preclude certain basis states from being in the image of $\mathcal{E}^{(2)}$, and partly characterise the action of $\mathcal{E}^{(2)}$ on the remaining basis states.

\begin{lemma} \label{lem:E2} Let $\mathcal{E}^{(2)}$ be $\mathcal{E}_{\mathrm{M}_n}^{(2)}$ or $\mathcal{E}_{\mathrm{M}_n\cap \mathrm{Cl}_n}^{(2)}$, defined as in Eqs.~\eqref{EjMn} and~\eqref{EjMnCln}. 

\noindent $\mathrm{(a)}$ For $S_1, S_2 \subseteq [2n]$, $\mathcal{E}^{(2)}\kett{\gamma_{S_1}}\kett{\gamma_{S_2}} \neq 0$ only if $S_1 = S_2$.

\smallskip

\noindent $\mathrm{(b)}$ $\mathcal{E}^{(2)}\kett{\gamma_S}\kett{\gamma_S} = \mathcal{E}^{(2)}\kett{\gamma_{S'}}\kett{\gamma_{S'}}$ for any $S, S' \subseteq [2n]$ such that $|S| = |S'|$.
\end{lemma}
\begin{proof} Let $S_1,S_2,S,S' \subseteq [2n]$.
\begin{enumerate}[(a)]
\item If $S_1 \neq S_2$, there must exist some index $\mu \in [2n]$ such that $\mu \in S_1$ and $\mu \not\in S_2$, or some index $\mu \in [2n]$ such that $\mu \in S_2$ and $\mu \not\in S_1$. In either case, let $Q \in \mathrm{B}(2n) \subset \orth$ be the reflection matrix such that $\mU_{Q}\kett{\gamma_\mu} = -\kett{\gamma_\mu}$ and $\mU_Q\kett{\gamma_\nu} = \kett{\gamma_\nu}$ for all $\nu \neq \mu$. Then, $\mU_Q^{\otimes 2}\kett{\gamma_{S_1}}\kett{\gamma_{S_2}} = -\kett{\gamma_{S_1}}\kett{\gamma_{S_2}}$ (using Fact~\ref{fact:mUQ}(a)).  Hence, by Eq.~\eqref{Einvariance},
\[ \mathcal{E}^{(2)}\kett{\gamma_{S_1}}\kett{\gamma_{S_2}} = \mathcal{E}^{(2)} \mU_Q^{\otimes 2}\kett{\gamma_{S_1}}\kett{\gamma_{S_2}} = -\mathcal{E}^{(2)}\kett{\gamma_{S_1}}\kett{\gamma_{S_2}}, \]
so $\mathcal{E}^{(2)}\kett{\gamma_{S_1}}\kett{\gamma_{S_2}} = 0$.
\item Suppose $|S| = |S'|$. Let $Q' \in \mathrm{B}(2n) \subset \orth$ be any permutation matrix such that $\mathcal{U}_{Q'}\kett{\gamma_S} = \kett{\gamma_{S'}}$. (Specifically, if $S = \{\mu_1,\dots, \mu_{|S|}\}$ with $\mu_1 < \dots <\mu_{|S|}$ and $S' = \{\mu_1',\dots, \mu_{|S|}'\}$ with $\mu_1' < \dots <\mu_{|S|}'$, take any permutation $Q'$ that maps $\mu_i \mapsto \mu_i'$ for each $i \in [|S|]$; it is clear that such a permutation always exists). Then, using Eq.~\eqref{Einvariance},
\[ \mathcal{E}^{(2)}\kett{\gamma_S}\kett{\gamma_S} = \mathcal{E}^{(2)} \mathcal{U}_Q^{\otimes 2}\kett{\gamma_S}\kett{\gamma_S} = \mathcal{E}^{(2)}\kett{\gamma_{S'}}\kett{\gamma_{S'}}. \]
\end{enumerate}
\end{proof}

We now prove Theorem~\ref{thm:moments}(ii).

\begin{proof}[Proof of Theorem~\ref{thm:moments}(ii)] 
Inserting resolutions of the identity [Eq.~\eqref{resolutionofidentity}] and using Lemma~\ref{lem:E2}(a), we have
\[ \mathcal{E}^{(2)} = \sum_{S, S'\subseteq [2n]} \kett{\gamma_{S}}\kett{\gamma_{S}}\braa{\gamma_{S}}\braa{\gamma_{S}}\mathcal{E}^{(2)}\kett{\gamma_{S'}}\kett{\gamma_{S'}}\braa{\gamma_{S'}}\braa{\gamma_{S'}} \]
(note that since $(\mathcal{E}^{(2)})^\dagger = \mathcal{E}^{(2)}$ by Fact~\ref{fact:Eprojector},  Lemma~\ref{lem:E2}(a) also implies that $\braa{\gamma_{S_1}}\braa{\gamma_{S_2}}\mathcal{E}^{(2)} = 0$ if $S_1\neq S_2$). It follows from Fact~\ref{fact:mUQ}(d) that $\mathcal{E}^{(2)}(\Gamma_k \otimes \Gamma_k) = \Gamma_k \otimes \Gamma_k$ for all $k \in \{0,\dots, 2n\}$, so $\braa{\gamma_{S}}\braa{\gamma_{S}}\mathcal{E}^{(2)}\kett{\gamma_{S'}}\kett{\gamma_{S'}} \neq 0$ only if $|S| = |S'|$. Hence, we can write
\[ \mathcal{E}^{(2)} = \sum_{k = 0}^{2n}\sum_{S, S' \in {[2n]\choose k}} \kett{\gamma_{S}}\kett{\gamma_{S}}\braa{\gamma_{S}}\braa{\gamma_{S}}\mathcal{E}^{(2)}\kett{\gamma_{S'}}\kett{\gamma_{S'}}\braa{\gamma_{S'}}\braa{\gamma_{S'}}. \]
Now, by Lemma~\ref{lem:E2}(b), the coefficient $\braa{\gamma_{S}}\braa{\gamma_{S}}\mathcal{E}^{(2)}\kett{\gamma_{S'}}\kett{\gamma_{S'}}$ is the same for \textit{all} pairs of subsets $S,S'$ of the same cardinality $k$, i.e., for some number $b'_k \in \mathbb{C}$, we have
\[ \braa{\gamma_{S}}\braa{\gamma_{S}}\mathcal{E}^{(2)}\kett{\gamma_{S'}}\kett{\gamma_{S'}} = b'_k \quad \forall \, S, S' \in {[2n]\choose k}. \]
Thus, 
\begin{align*}
    \mathcal{E}^{(2)} &= \sum_{k=0}^{2n} b'_k \sum_{S,S' \in {[2n]\choose k}} \kett{\gamma_{S}}\kett{\gamma_{S}}\braa{\gamma_{S'}}\braa{\gamma_{S'}} \\
    &= \sum_{k=0}^{2n} b_k \kett{\Upsilon^{(2)}_{k}}\braa{\Upsilon^{(2)}_k},
\end{align*}
where $\kett{\Upsilon_k^{(2)}}$ is defined as in Eq.~\eqref{Psi2} and we rescale $b_k'$ to $b_k \coloneqq b_k'{2n\choose k}$ to account for the normalisation of $\kett{\Upsilon_k^{(2)}}$.

Since $\mathcal{E}^{(2)}$ is a projector (Fact~\ref{fact:Eprojector}), each $b_k$ must equal $0$ or $1$. To complete the proof, we show that $b_k = 1$ for all $k \in \{0,\dots, 2n\}$ by showing that $\mU_Q^{\otimes 2}\kett{\Upsilon_k^{(2)}} = \kett{\Upsilon_k^{(2)}}$ for all $Q \in \orth$, so $\mathcal{E}^{(2)}\kett{\Upsilon_k^{(2)}} = \kett{\Upsilon_k^{(2)}}$: 
\begin{align*}
    \mathcal{U}_Q^{\otimes 2}\kett{\Upsilon_k^{(2)}} &= {2n\choose k}^{-1/2}\sum_{S \in {[2n]\choose k}}\mU_Q\kett{\gamma_S}\otimes \mU_Q\kett{\gamma_S} \\
    &= {2n \choose k}^{-1/2}\sum_{S \in {[2n]\choose k}}\sum_{S',S'' \in {[2n]\choose k}}\det(Q\big|_{S,S'})\det(Q\big|_{S,S''}) \kett{\gamma_{S'}}\kett{\gamma_{S''}} \\
    &= {2n\choose k}^{-1/2} \sum_{S',S'' \in {[2n]\choose k}} \det\left(Q^{\mathrm{T}}\big|_{S',[2n]}Q\big|_{[2n],S''}\right)\kett{\gamma_{S'}}\kett{\gamma_{S''}} \\
    &= {2n\choose k}^{-1/2}\sum_{S',S''\in{[2n]\choose k}} \det(I\big|_{S',S''}) \kett{\gamma_{S'}}\kett{\gamma_{S''}} \\
    &= {2n\choose k}^{-1/2}\sum_{S'\in{[2n]\choose k}}  \kett{\gamma_{S'}}\kett{\gamma_{S'}} \\
    &= \kett{\Upsilon_k^{(2)}},
\end{align*}
where we use Eq.~\eqref{UQdef} in the second line and the Cauchy-Binet formula in the fourth, noting that $\det(Q_{S,S'}) = \det((Q_{S,S'})^{\mathrm{T}}) = \det((Q^{\mathrm{T}})_{S',S})$.  
\end{proof}

\subsubsection{The 1-fold twirl $\mathcal{E}^{(1)}$}

\begin{proof}[Proof of Theorem~\ref{thm:moments}(i)] From Eq.~\eqref{EjtoEj-1}, 
\begin{align*}
    \mathcal{E}^{(1)}\kett{A} &= \tr_1\left[\mathcal{E}^{(2)} \frac{1}{\sqrt{2^n}}\kett{\gamma_{\varnothing}}\kett{A}\right]
\end{align*}
for any $A \in \LHn$. Substituting in the expression for $\mathcal{E}^{(2)}$ from Theorem~\ref{thm:moments}(ii) and using Hilbert-Schmidt orthogonality of the $\gamma_S$ (Eq.~\eqref{kettgammaS}), this evaluates to $\mathcal{E}^{(1)}\kett{A} = \kett{\gamma_\varnothing}\brakett{\gamma_{\varnothing}}{A}$, so $\mathcal{E}^{(1)} = \kett{\gamma_\varnothing}\braa{\gamma_{\varnothing}}$. 

Alternatively, it is easily seen that $\mathcal{E}^{(1)}\kett{\gamma_S} = 0$ for any $S \neq \varnothing$ (take any $\mu \in S$, and use Eq.~\eqref{Einvariance} with $Q$ the reflection that maps $\mu \mapsto -\mu$ and $\nu \mapsto \nu$ for all $\nu \neq \mu$ to obtain $\mathcal{E}^{(1)}\kett{\gamma_S} = -\mathcal{E}^{(1)}\kett{\gamma_S}$), while $\mathcal{E}^{(1)}\kett{\gamma_\varnothing} = \kett{\gamma_\varnothing}$.
\end{proof}

\subsubsection{The 3-fold twirl $\mathcal{E}^{(3)}$}

The 3-fold twirl channel $\mathcal{E}^{(3)}$ (which, as discussed above, represents $\mathcal{E}^{(3)}_{\mathrm{M}_n}$ or $\mathcal{E}^{(3)}_{\mathrm{M}_n \cap \mathrm{Cl}_n}$) can be calculated along the same lines as $\mathcal{E}^{(2)}$. The following lemma is the analogue of Lemma~\ref{lem:E2}, for $\mathcal{E}^{(3)}$. 
\begin{lemma} \label{lem:E3} Let $\mathcal{E}^{(3)}$ be $\mathcal{E}^{(3)}_{\mathrm{M}_n}$ or $\mathcal{E}^{(3)}_{\mathrm{M}_n \cap \mathrm{Cl}_n}$, defined as in Eqs.~\eqref{EjMn} and~\eqref{EjMnCln}.

\begin{enumerate}
    \item For $S_1,S_2,S_3 \subseteq [2n]$, $\mathcal{E}^{(3)}\kett{\gamma_{S_1}}\kett{\gamma_{S_2}}\kett{\gamma_{S_3}} \neq 0$ only if $S_1, S_2, S_3$ are of the form
\[ S_1 = A_1 \cup A_2, \quad S_2 = A_2 \cup A_3, \quad S_3 = A_3 \cup A_1 \]
for some mutually disjoint subsets $A_1,A_2,A_3 \subseteq [2n]$.
\item $\mathcal{E}^{(3)}\kett{\gamma_{A_1}\gamma_{A_2}} \kett{\gamma_{A_2}\gamma_{A_3}}\kett{\gamma_{A_3}\gamma_{A_1}} = \mathcal{E}^{(3)}\kett{\gamma_{A_1'}\gamma_{A_2'}} \kett{\gamma_{A_2'}\gamma_{A_3'}}\kett{\gamma_{A_3'}\gamma_{A_1'}}$ for any subsets $A_1,A_2,A_3,A_1',A_2',A_3' \subseteq [2n]$ such that $A_1,A_2,A_3$ are mutually disjoint, $A_1',A_2',A_3'$ are mutually disjoint, and $|A_i| = |A_i'|$ for all $i \in \{1,2,3\}$.
\end{enumerate} 
\end{lemma}
Note that since Majorana operators anticommute, $\kett{\gamma_{A_i}\gamma_{A_j}}$ may differ from $\kett{\gamma_{A_i \cup A_j}}$ by a minus sign. We provide the proof of Lemma~\ref{lem:E3}, which uses symmetry arguments similar to those in Lemma~\ref{lem:E2}, in Appendix~\ref{app:E3}. This lemma should perhaps make the form of $\kett{\Upsilon^{(3)}_{k_1,k_2,k_3}}$ (defined in Eq.~\eqref{Psi3}) in the expression for $\mathcal{E}^{(3)}$ somewhat more intuitive; it allows us to prove Theorem~\ref{thm:moments}(iii) as follows.

\begin{proof}[Proof sketch for Theorem~\ref{thm:moments}(iii)] 
Lemma~\ref{lem:E3}(b) (together with $(\mathcal{E}^{(3)})^\dagger = \mathcal{E}^{(3)}$, from Fact~\ref{fact:Eprojector}) implies that for each triplet of integers $k_1,k_2,k_3 \in \{0,\dots, 2n\}$ such that $k_1 + k_2 + k_3 \leq 2n$, there exists some number $c_{k_1,k_2,k_3}' \in \mathbb{C}$ such that
\begin{equation} \label{E3elements}
    \braa{\gamma_{A_1}\gamma_{A_2}}\braa{\gamma_{A_2}\gamma_{A_3}}\braa{\gamma_{A_3}\gamma_{A_1}}\mathcal{E}^{(3)}\kett{\gamma_{A_1'}\gamma_{A_2'}}\kett{\gamma_{A_2'}\gamma_{A_3'}}\kett{\gamma_{A_3'}\gamma_{A_1'}} = c'_{k_1,k_2,k_3}
\end{equation}
for any subsets $A_1,A_2, A_3, A_1', A_2',A_3' \subseteq [2n]$ such that $A_1, A_2, A_3$ are mutually disjoint, $A_1', A_2', A_3'$ are mutually disjoint, and $|A_i| = |A_i'| = k_i'$ for all $i \in \{1,2,3\}$. Inserting resolutions of identities to the left and right of $\mathcal{E}^{(3)}$, then using Lemma~\ref{lem:E3}(a) and Eq.~\eqref{E3elements} gives
\begin{align*}
    \mathcal{E}^{(3)} &= \sum_{\substack{k_1,k_2,k_3 \in \{0,\dots, 2n\}\\ k_1 + k_2 + k_3 \leq 2n}} c'_{k_1,k_2,k_3} \\
    &\qquad \times \sum_{\substack{A_1,A_2, A_3 \subseteq [2n] \text{ } \mathrm{ disjoint} \\ |A_1| = k_1, |A_2| = k_2, |A_3| = k_3} }\sum_{\substack{A_1',A_2', A_3' \subseteq [2n] \text{ } \mathrm{ disjoint} \\ |A_1'| = k_1, |A_2'| = k_2, |A_3'| = k_3}}\kett{\gamma_{A_1}\gamma_{A_2}} \kett{\gamma_{A_2}\gamma_{A_3}}\kett{\gamma_{A_3}\gamma_{A_1}}\braa{\gamma_{A_1'}\gamma_{A_2'}}\braa{\gamma_{A_2'}\gamma_{A_3'}}\braa{\gamma_{A_3'}\gamma_{A_1'}}
\end{align*}
which we can rewrite as 
\[ \mathcal{E}^{(3)} = \sum_{\substack{k_1,k_2,k_3 \in \{0,\dots, 2n\}\\ k_1 + k_2 + k_3 \leq 2n}} c_{k_1,k_2,k_3}\kett{\Upsilon_{k_1,k_2,k_3}^{(3)}}\braa{\Upsilon_{k_1,k_2,k_3}^{(3)}} \]
by letting $c_{k_1,k_2,k_3} \coloneqq c'_{k_1,k_2,k_3}{2n \choose k_1,k_2,k_3, 2n-k_1-k_2-k_3}$. The fact that $\mathcal{E}^{(3)}$ is a projector (Fact~\ref{fact:Eprojector}) implies that each $c_{k_1,k_2,k_3}$ is either $0$ or $1$. In Appendix~\ref{app:E3}, we show that for each $k_1,k_2, k_3$ appearing in the sum, $\mathcal{E}^{(3)}\kett{\Upsilon_{k_1,k_2,k_3}^{(3)}} = \kett{\Upsilon_{k_1,k_2,k_3}^{(3)}}$ by proving that $\mU_Q^{\otimes 3}\kett{\Upsilon_{k_1,k_2,k_3}^{(3)}} = \kett{\Upsilon_{k_1,k_2,k_3}^{(3)}}$ for all $Q \in \orth$; it then follows that $c_{k_1,k_2,k_3} = 1$.
\end{proof}

\subsection{Classical shadows via matchgate circuits} \label{sec:shadows}

In this subsection, we characterise the classical shadows arising from the uniform distribution over the continuous group $\mathrm{M}_n$ [Eq.~\eqref{Mn}] of all matchgate circuits, and from the uniform distribution over the discrete group $\mathrm{M}_n \cap \mathrm{Cl}_n$  [Eq.~\eqref{MnCl}] of matchgate circuits that are also in the Clifford group $\mathrm{Cl}_n$. To implement the classical shadows protocol, we must be able sample unitaries from these distributions and implement them on a quantum computer; we discuss efficient methods for doing so in Appendix~\ref{app:random_circuit_sampling}.

With Theorem~\ref{thm:moments} in hand, we can straightforwardly find explicit expressions for the measurement channel $\mathcal{M}$ [Eq.~\eqref{Mgeneral2}] in the protocol as well as the variance $\mathrm{Var}[\hat{o}_i]$ [Eq.~\eqref{Varoi}] of the expectation value estimators $\hat{o}_i$ obtained from the classical shadows, when either of the two distributions is used. Since Eqs.~\eqref{Mgeneral2} and~\eqref{Varoi} depend on the distribution only through the 2- and 3-fold twirl channels, and Theorem~\ref{thm:moments} shows that the $j$-fold twirl channels for $\mathrm{M}_n$ and $\mathrm{M}_n \cap \mathrm{Cl}_n$ coincide for $j \in \{1,2,3\}$, it follows that the classical shadows measurement channel and the variances are exactly the same for the two distributions. Hence, we will use the same notation ($\mathcal{M}$ and $\mathrm{Var}[\hat{o}_i]$) for both distributions. 

\subsubsection{Measurement channel} \label{sec:shadows_channel}

First, we calculate the classical shadows measurement channel $\mathcal{M}$ by substituting the expression for the 2-fold twirl from Theorem~\ref{thm:moments}(ii) into Eq.~\eqref{Mgeneral2}.

\begin{proof}[Proof of Eq.~\eqref{matchgateshadowschannel}]
By Eq.~\eqref{Mgeneral2}, the measurement channel $\mathcal{M}$ associated with the uniform distribution over $\mathrm{M}_n$ or over $\mathrm{M}_n\cap \mathrm{Cl}_n$ is given by 
\[ \mathcal{M}(A) = \tr_1\left[\sum_{b \in \{0,1\}^n} \mathcal{E}^{(2)}(\ket{b}\bra{b}^{\otimes 2})(A \otimes I) \right] \]
for all $\mathcal{A} \in \LHn$,
where $\mathcal{E}^{(2)} \coloneqq \mathcal{E}^{(2)}_{\mathrm{M}_n} = \mathcal{E}^{(2)}_{\mathrm{M}_n \cap \mathrm{Cl}_n} = \sum_{k=0}^{2n}\kett{\Upsilon_{k}^{(2)}}\braa{\Upsilon_k^{(2)}}$ with $\kett{\Upsilon_k}^{(2)} \coloneqq {2n\choose k}^{-1/2} \sum_{S \in {[2n]\choose k}} \kett{\gamma_S}\kett{\gamma_S}$, by Theorem~\ref{thm:moments}(ii). We can simplify this by noting that $\mathcal{E}^{(2)}(\ket{b}\bra{b}^{\otimes 2}) = \mathcal{E}^{(2)}(\ket{\mathbf{0}}\bra{\mathbf{0}}^{\otimes 2})$ for all $b\in \{0,1\}^n$, which follows from the fact that computational basis states $\ket{b}\bra{b}$ are all Gaussian states (see subsection~\ref{sec:background_Gaussianstates}), and $\mathcal{E}^{(2)} = \mathcal{E}^{(2)}_{\mathrm{M}_n}$ is invariant under composition with any Gaussian unitary channel $\mathcal{U}_Q$. More explicitly, for each $b \in \{0,1\}^n$, let $Q_b \in \mathrm{B}(2n)$ be the matrix such that $\mathcal{U}_{Q_b}$ maps $\gamma_{2j-1} \mapsto -\gamma_{2j-1}$ for every $j \in [n]$ such that $b_j = 1$, and leaves all the other $\gamma_\mu$ unchanged. Then, $\mU_{Q_b}(\ket{b}\bra{b}) = \ket{\mathbf{0}}\bra{\mathbf{0}}$ from Eq.~\eqref{bb}, so $\mathcal{E}^{(2)}(\ket{b}\bra{b}^{\otimes 2}) = (\mathcal{E}^{(2)}\circ \mU_{Q_b})(\ket{b}\bra{b}^{\otimes 2}) = \mathcal{E}^{(2)}(\ket{\mathbf{0}}\bra{\mathbf{0}})$ using Eq.~\eqref{Einvariance}. Hence,
\begin{equation} \label{ME0} \mathcal{M}(A) = 2^n \tr_1\left[\mathcal{E}^{(2)}(\ket{\mathbf{0}}\bra{\mathbf{0}}^{\otimes 2})(A\otimes I) \right],\end{equation}
so it remains to calculate $\mathcal{E}^{(2)}(\ket{\mathbf{0}}\bra{\mathbf{0}}^{\otimes 2})$.

For convenience, let $\kett{\Pi_0} \equiv \ket{\mathbf{0}}\bra{\mathbf{0}}$ in Liouville representation, so $\mathcal{E}^{( 2)}(\ket{\mathbf{0}}\bra{\mathbf{0}}^{\otimes 2}) = \sum_{k=0}^{2n}\kett{\Upsilon_k^{(2)}}\braa{\Upsilon_k^{(2)}}(\kett{\Pi_0}^{\otimes 2})$. From Eq.~\eqref{bb}, we have  $\ket{\mathbf{0}}\bra{\mathbf{0}} = \frac{1}{2^n}\sum_{T \subseteq [n]}\prod_{j \in T} (-i\gamma_{2j-1}\gamma_{2j})$,
so we can expand $\kett{\Pi_0}$ in the $\gamma_S$ basis as
\begin{equation} \label{Pi0pairs} \kett{\Pi_0} = \frac{1}{\sqrt{2^n}} \sum_{T \subseteq [n]} (-i)^{|T|} \kett{\gamma_{\text{pairs}(T)}},\end{equation}
where $\text{pairs}(T) \coloneqq \bigcup_{j \in T}\{2j-1,2j\}$. From this, it is clear that $\brakett{\gamma_S}{\Pi_0}$ is only nonzero for even-cardinality subsets $S$ of the form $\text{pairs}(T)$ for some $T \subseteq [n]$; in such cases, $\brakett{\gamma_S}{\Pi_0} = (-i)^{|T|}/\sqrt{2^n}$. Thus,
\begin{align*}
    \braa{\Upsilon_k^{(2)}}(\kett{\Pi_0}^{\otimes 2}) &= {2n\choose k}^{-1/2} \sum_{S \in {[2n]\choose k}} \brakett{\gamma_S}{\Pi_0}^2 \\
    &= {2n\choose k}^{-1/2} \mathbf{1}_{\text{$k$ even}} \sum_{T \in {[n]\choose k/2}} \left( \frac{1}{\sqrt{2^n}}(-i)^{|T|}\right)^2 \\
    &= {2n\choose k}^{-1/2} \mathbf{1}_{\text{$k$ even}}{n\choose k/2} \frac{1}{2^n} (-1)^{k/2}  ,
\end{align*}
so we have
\begin{align*}
    \mathcal{E}^{(2)}(\ket{\mathbf{0}}\bra{\mathbf{0}}^{\otimes 2}) &= \sum_{k = 0}^{2n} {2n\choose k}^{-1/2} \mathbf{1}_{\text{$k$ even}}{n\choose k/2} \frac{1}{2^n} (-1)^{k/2} \kett{\Upsilon_k^{(2)}} \\
    &= \frac{1}{2^n}\sum_{\ell =0}^n (-1)^\ell {2n\choose 2\ell}^{-1/2}{n\choose \ell} {2n\choose 2\ell}^{-1/2} \sum_{S \in {[2n]\choose 2\ell}}\kett{\gamma_S}\kett{\gamma_S} \\
    &= \frac{1}{2^{2n}} \sum_{\ell =0}^n {2n\choose 2\ell}^{-1} {n\choose \ell} \sum_{S \in {[2n]\choose 2\ell}}\gamma_S^\dagger \otimes \gamma_S,
\end{align*}
using $\gamma_{S}^\dagger = (-1)^{|S|(|S|-1)/2}\gamma_S$ in the last line, which is a simple consequence of the fact that Majorana operators anticommute (and recalling that $\kett{\gamma_S} \equiv \gamma_S/\sqrt{2^n}$ due to normalisation). Inserting this into Eq.~\eqref{ME0} gives
\begin{align*}
    \mathcal{M}(A) &= \frac{1}{2^n} \sum_{\ell = 0}^n {2n\choose 2\ell}^{-1} {n\choose \ell}\sum_{S \in {[2n]\choose 2\ell}}\tr(\gamma_S^\dagger A)\gamma_S \\
    &= \sum_{\ell = 0}^n {2n\choose 2\ell}^{-1} {n\choose \ell}\mathcal{P}_{2\ell}(A),
\end{align*}
by definition of the projectors $\mathcal{P}_k$ [Eq.~\eqref{Pk}].
\end{proof}

From Eq.~\eqref{matchgateshadowschannel}, we see that $\mathcal{M}$ maps $\LHn$ onto the subspace $\Gamma_{\text{even}} = \oplus_{\ell =0}^{n}\Gamma_{2\ell}$ of even operators. We denote the (pseudo)inverse of $\mathcal{M}$ on this subspace by $\mathcal{M}^{-1}: \Gamma_{\text{even}} \to \Gamma_{\text{even}}$, which clearly has the form given in Eq.~\eqref{Minverse}.

We now consider the consequences of the fact that the image of $\mathcal{M}$ is $\Gamma_{\text{even}}$ for the classical shadows protocol. Observe from Eq.~\eqref{bb} that for any computational basis state $\ket{b}$, the corresponding density operator $\ket{b}\bra{b}$ is in $\Gamma_{\text{even}}$. Then, since conjugation by any matchgate circuit leaves $\Gamma_{\text{even}}$ invariant, the post-measurement state $U_Q^\dagger \ket{b}\bra{b}U_Q$ is an even operator for any $b \in \{0,1\}^n$ and $Q \in \orth$. Thus, for both distributions ($\mathrm{M}_n$ and $\mathrm{M}_n \cap \mathrm{Cl}_n$), the classical shadows $\mathcal{M}^{-1}(U_Q^\dagger \ket{b}\bra{b}U_Q)$ are well-defined. As noted in section~\eqref{sec:background_shadows}, these classical shadows constitute unbiased estimates for the unknown state $\rho$ if $\rho$ is in the image of $\mathcal{M}$---in our case, if $\rho \in \Gamma_{\text{even}}$. More explicitly, defining the random variable $\hat{\rho} = \mathcal{M}^{-1}(\hat{U}_Q^\dagger\ket{\hat{b}}\bra{\hat{b}}\hat{U}_Q)$ as in Eq.~\eqref{hatrho} (with $D$ taken to be either of our distributions over matchgate circuits), we have 
\begin{align*} \E[\hat{\rho}] = \mathcal{M}^{-1}(\mathcal{M}(\rho)) = \mathcal{P}_{\text{even}}(\rho),
\end{align*}
where $\mathcal{P}_{\text{even}} \coloneqq \sum_{\ell = 0}^{n} \mathcal{P}_{2\ell}$ is the projector onto $\Gamma_{\text{even}}$. Thus, to obtain unbiased estimators $\hat{o}_i \coloneqq \tr(O_i\hat{\rho})$ for the expectation values of arbitrary observables $O_i$ with respect to $\rho$ using this classical shadows protocol, it suffices for $\rho$ to be an even operator. However, it would also suffice for all of the observables $O_i$ to be even (and $\rho$ to be arbitrary), due to the Hilbert-Schmidt orthogonality of Majorana operators. In particular, $\tr(\mathcal{P}_{\mathrm{even}}(A)B) = \tr(A\mathcal{P}_{\mathrm{even}}(B))$ for any $A,B \in \LHn$, so if $O_i \in \Gamma_{\text{even}}$, we have
\[ \E[\hat{o}_i] = \tr(O_i\mathcal{P}_{\text{even}}(\rho)) = \tr( \mathcal{P}_{\text{even}}(O_i)\rho)= \tr(O_i\rho). \]
Therefore, we require that \textit{either} $\rho \in \Gamma_{\text{even}}$, or $O_i \in \Gamma_{\text{even}}$ for all $i$. 
\subsubsection{Variance} \label{sec:shadows_variance}

Having calculated $\mathcal{M}^{-1}$, we now substitute it along with the expression for the 3-fold twirl from Theorem~\ref{thm:moments}(iii) into Eq.~\eqref{Varoi} to obtain the variance bound Eq.~\eqref{varianceboundtilde}, which holds for any $O \in \Gamma_{\text{even}}$. Note that it suffices to consider the variance for even observables---even if we are in the case where $\rho \in \Gamma_{\text{even}}$ while the observables $O_i$ can be arbitrary, we have $\hat{o}_i = \tr(O_i \hat{\rho}) = \tr( \mathcal{P}_{\text{even}}(O_i)\hat{\rho})$ since our classical shadow samples $\mathcal{M}^{-1}(U_Q^\dagger \ket{b}\bra{b} U_Q)$ are all even operators, so the variance for $O_i$ is equal to the variance for its projection $\mathcal{P}_{\text{even}}(O_i) \in \Gamma_{\text{even}}$ onto the even subspace.

\begin{proof}[Proof of Eq.~\eqref{varianceboundtilde}] 
The variance of the unbiased estimator $\hat{o}$ for $\tr(O\rho)$ can be upper bounded as $\Var[\hat{o}] \leq \E[|\hat{o}|^2]$, and from Eq.~\eqref{Varoi}, we have
\begin{align*}
  \E[|\hat{o}|^2] = \tr\left[\sum_{b \in \{0,1\}^n} \mathcal{E}^{(3)}(\ket{b}\bra{b}^{\otimes 3}) \left(\rho \otimes \mathcal{M}^{-1}(O) \otimes \mathcal{M}^{-1}(O^\dagger\right) \right]  
\end{align*}
for $O \in \Gamma_{\text{even}}$, where $\mathcal{E}^{(3)} \coloneqq \mathcal{E}^{(3)}_{\mathrm{M}_n} = \mathcal{E}^{(3)}_{\mathrm{M}_n \cap \mathrm{Cl}_n}$. By Theorem~\ref{thm:moments}(ii), we have $\mathcal{E}^{(3)} = \sum_{\substack{k_1,k_2,k_3 \geq 0\\k_1 + k_2 + k_3 \leq 2n}} \kett{\Upsilon_{k_1,k_2,k_3}^{(3)}}\braa{\Upsilon_{k_1,k_2,k_3}^{(3)}}$, with $\kett{\Upsilon^{(3)}_{k_1,k_2,k_3}}$ given by Eq.~\eqref{Psi3}. Just as in the proof of Eq.~\eqref{matchgateshadowschannel}, we first use Eq.~\eqref{Einvariance} to infer that $\mathcal{E}^{(3)}(\ket{b}\bra{b}) = \mathcal{E}^{(3)}(\ket{\mathbf{0}}\bra{\mathbf{0}}^{\otimes 3})$ for all $b \in \{0,1\}$, leading to
\begin{equation} \label{EoE0} \E[|\hat{o}|^2] = 2^{n} \tr \left[\mathcal{E}^{(3)}(\ket{\mathbf{0}}\bra{\mathbf{0}}^{\otimes 3})\left(\rho \otimes \mathcal{M}^{-1}(O) \otimes \mathcal{M}^{-1}(O^\dagger)\right) \right].
\end{equation}

Next, we calculate $\mathcal{E}^{(3)}(\ket{\mathbf{0}}\bra{\mathbf{0}}^{\otimes 3}) = \sum_{\substack{k_1,k_2,k_3 \geq 0\\k_1 + k_2 +k_3 \leq 2n}}\kett{\Upsilon^{(3)}_{k_1,k_2,k_3}}\braa{\Upsilon_{k_1,k_2,k_3}^{(3)}}(\kett{\Pi_0}^{\otimes 3})$, where $\kett{\Pi_0} \equiv \ket{\mathbf{0}}\bra{\mathbf{0}}$. By Eq.~\eqref{Psi3}, 
\[ \braa{\Upsilon_{k_1,k_2,k_3}^{(3)}}(\kett{\Pi_0}^{\otimes 3}) = {{2n\choose k_1,k_2,k_3}'}^{-1/2} \sum_{\substack{A_1,A_2, A_3 \subseteq [2n] \text{ disjoint} \\|A_1| = k_1, |A_2| = k_2, |A_3| = k_3}} \brakett{\gamma_{A_1}\gamma_{A_2}}{\Pi_0}\brakett{\gamma_{A_2}\gamma_{A_3}}{\Pi_0} \brakett{\gamma_{A_3}\gamma_{A_1}}{\Pi_0},\]
where for brevity we write \begin{equation} \label{multinomial_shorthand} {n\choose k_1,k_2,k_3}' \equiv {n\choose k_1,k_2,k_3, 2n-k_1-k_2-k_3} \end{equation}
for the multinomial coefficient.
By Eq.~\eqref{Pi0pairs}, $\kett{\Pi_0}$ is a linear combination of Majorana products $\gamma_{\text{pairs}(T)}$ corresponding to subsets $\text{pairs}(T) \coloneqq \bigcup_{j \in T}\{2j-1,2j\}$ consisting only of pairs of indices $2j-1$ and $2j$, so we see that $\brakett{\gamma_{A_1}\gamma_{A_2}}{\Pi_0}\brakett{\gamma_{A_2}\gamma_{A_3}}{\Pi_0} \brakett{\gamma_{A_3}\gamma_{A_1}}{\Pi_0}$ can only be nonzero if $A_1 \cup A_2 = \text{pairs}(T_1')$, $A_2 \cup A_3 = \text{pairs}(T_2')$, and $A_3 \cup A_1 = \text{pairs}(T_3')$ for some $T_1',T_2', T_3' \subseteq [n]$. For mutually disjoint $A_1,A_2,A_3\subseteq[2n]$, this condition is equivalent to $A_1 = \text{pairs}(T_1)$, $A_2 = \text{pairs}(T_2)$, and $A_3 = \text{pairs}(T_3)$ for some mutually disjoint $T_1,T_2, T_3 \subseteq [n]$, in which case $\brakett{\gamma_{A_i}\gamma_{A_j}}{\Pi_0} = (-i)^{|T_i| + |T_j|}/\sqrt{2^n}$ (for $i,j \in \{1,2,3\}$, $i \neq j$) by Eq.~\eqref{Pi0pairs}. Hence,
\begin{align*}
\braa{\Upsilon_{k_1,k_2,k_3}^{(3)}}(\kett{\Pi_0}^{\otimes 3}) &=   {{2n\choose k_1,k_2,k_3}'}^{-1/2}\mathbf{1}_{k_1,k_2,k_3 \text{ even}} \\
&\quad \times \sum_{\substack{T_1,T_2, T_3 \subseteq [n] \text{ disjoint} \\ |T_1| = k_1/2, |T_2| = k_2/2, |T_3| = k_3/2}}\left(\frac{1}{\sqrt{2^n}} (-i)^{|T_1| + |T_2|}\right)\left(\frac{1}{\sqrt{2^n}} (-i)^{|T_2| + |T_3|}\right)\left(\frac{1}{\sqrt{2^n}} (-i)^{|T_3| + |T_1|}\right) \\
&= {{2n\choose k_1,k_2,k_3}'}^{-1/2}\mathbf{1}_{k_1,k_2,k_3 \text{ even}} {n\choose k_1/2, k_2/2, k_3/2}' \frac{1}{2^{3n/2}}(-1)^{(k_1 + k_2 + k_3)/2},
\end{align*}
from which we obtain
\begin{align*}
    \mathcal{E}^{(3)}(\ket{\mathbf{0}}\bra{\mathbf{0}}^{\otimes 3}) &= \sum_{\substack{k_1,k_2,k_3 \geq 0\\ k_1 + k_2 + k_3 \leq 2n}}{{2n\choose k_1,k_2,k_3}'}^{-1/2}\mathbf{1}_{k_1,k_2,k_3 \text{ even}} {n\choose k_1/2, k_2/2, k_3/2}' \frac{1}{2^{3n/2}}(-1)^{(k_1 + k_2 + k_3)/2} \kett{\Upsilon^{(3)}_{k_1,k_2,k_3}} \\
    &= \frac{1}{2^{3n/2}} \sum_{\substack{\ell_1,\ell_2,\ell_3 \geq 0\\ \ell_1 + \ell_2 + \ell_3 \leq n}} (-1)^{\ell_1 + \ell_2 + \ell_3} {{2n\choose 2\ell_1,2\ell_2,2\ell_3}'}^{-1/2} {n\choose \ell_1, \ell_2, \ell_3}'  \\
    &\quad \times {{2n\choose 2\ell_1,2\ell_2,2\ell_3}'}^{-1/2} \sum_{\substack{A_1, A_2, A_3 \subseteq [2n] \text{ disjoint}\\|A_1| =2\ell_1, A_2 = 2\ell_2, |A_3|= 2\ell_3}} \kett{\gamma_{A_1}\gamma_{A_2}}\kett{\gamma_{A_2}\gamma_{A_3}}\kett{\gamma_{A_3}\gamma_{A_1}} \\
    &= \frac{1}{2^{3n}} \sum_{\substack{\ell_1,\ell_2,\ell_3 \geq 0\\ \ell_1 + \ell_2 + \ell_3 \leq n}} (-1)^{\ell_1 + \ell_2 + \ell_3}  \frac{{n\choose \ell_1, \ell_2, \ell_3}'}{{2n\choose 2\ell_1,2\ell_2,2\ell_3}'} \sum_{\substack{A_1, A_2, A_3 \subseteq [2n] \text{ disjoint}\\|A_1| =2\ell_1, A_2 = 2\ell_2, |A_3|= 2\ell_3}} \gamma_{A_1}\gamma_{A_2}\otimes\gamma_{A_2}\gamma_{A_3}\otimes\gamma_{A_3}\gamma_{A_1}.
\end{align*}
Now, note that we can change to any other Majorana basis $\widetilde{\gamma}_\mu = \sum_{\nu=1}^{2n} Q_{\mu\nu}\gamma_\mu = \mU_Q(\gamma_\mu)$ for $Q \in \orth$, by applying $\mU_Q^{\otimes 3}$ to both sides and using $\mU_Q^{\otimes 3} \circ \mathcal{E}^{(3)} = \mathcal{E}^{(3)}$ from Eq.~\eqref{Einvariance}, yielding
\[ \mathcal{E}^{(3)}(\ket{\mathbf{0}}\bra{\mathbf{0}}^{\otimes 3}) = \frac{1}{2^{3n}} \sum_{\substack{\ell_1,\ell_2,\ell_3 \geq 0\\ \ell_1 + \ell_2 + \ell_3 \leq n}} (-1)^{\ell_1 + \ell_2 + \ell_3}  \frac{{n\choose \ell_1, \ell_2, \ell_3}'}{{2n\choose 2\ell_1,2\ell_2,2\ell_3}'} \sum_{\substack{A_1, A_2, A_3 \subseteq [2n] \text{ disjoint}\\|A_1| =2\ell_1, A_2 = 2\ell_2, |A_3|= 2\ell_3}} \widetilde\gamma_{A_1}\widetilde\gamma_{A_2}\otimes\widetilde\gamma_{A_2}\widetilde\gamma_{A_3}\otimes\widetilde\gamma_{A_3}\widetilde\gamma_{A_1}.\]

Inserting this into Eq.~\eqref{EoE0} gives
\begin{align*}
    \E[|\hat{o}|^2] &= \frac{1}{2^{2n}} \sum_{\substack{\ell_1,\ell_2,\ell_3 \geq 0 \\ \ell_1 + \ell_2 + \ell_3 \leq n}} (-1)^{\ell_1 + \ell_2 + \ell_3}\frac{{n\choose \ell_1, \ell_2, \ell_3}'}{{2n\choose 2\ell_1,2\ell_2,2\ell_3}'} \\
    &\qquad \times \sum_{\substack{A_1, A_2, A_3 \subseteq [2n] \text{ disjoint}\\|A_1| =2\ell_1, A_2 = 2\ell_2, |A_3|= 2\ell_3}}\tr\left(\widetilde\gamma_{A_1}\widetilde\gamma_{A_2}\rho\right)\tr\left(\widetilde\gamma_{A_2}\widetilde\gamma_{A_3}\mathcal{M}^{-1}(O)\right)\tr\left(\widetilde\gamma_{A_3}\widetilde\gamma_{A_1}\mathcal{M}^{-1}(O^\dagger)\right).
\end{align*}
Finally, we use Eq.~\eqref{Minverse} to write, for $|A_2| = 2\ell_2$ and $|A_3| = 2\ell_3$,
\begin{align*}
    \tr\left(\widetilde\gamma_{A_2}\widetilde\gamma_{A_3}\mathcal{M}^{-1}(O)\right) &= \tr\left(\mathcal{M}^{-1}(\widetilde\gamma_{A_2}\widetilde\gamma_{A_3}) O\right) = {2n\choose 2\ell_2 + 2\ell_3}{n\choose \ell_2 + \ell_3}^{-1}\tr(\widetilde\gamma_{A_2}\widetilde\gamma_{A_3}O),
\end{align*}
and similarly for
$\tr(\gamma_{A_3}\gamma_{A_1}\mathcal{M}^{-1}(O^\dagger))$. We thus arrive at Eq.~\eqref{varianceboundtilde}:
\begin{align*}
    \Var[\hat{o}] \leq \E[|\hat{o}|^2] &= \frac{1}{2^{2n}} \sum_{\substack{\ell_1,\ell_2,\ell_3 \geq 0 \\ \ell_1 + \ell_2 + \ell_3 \leq n}} (-1)^{\ell_1 + \ell_2 + \ell_3}\frac{{n\choose \ell_1, \ell_2, \ell_3}'}{{2n\choose 2\ell_1,2\ell_2,2\ell_3}'}\frac{{2n\choose 2(\ell_2 + \ell_3)}}{{n\choose \ell_2 + \ell_3}}\frac{{2n\choose 2(\ell_3 + \ell_1)}}{{n\choose \ell_3 + \ell_1}} \\
    &\qquad \times \sum_{\substack{A_1, A_2, A_3 \subseteq [2n] \text{ disjoint}\\|A_1| =2\ell_1, A_2 = 2\ell_2, |A_3|= 2\ell_3}}\tr\left(\widetilde\gamma_{A_1}\widetilde\gamma_{A_2}\rho\right)\tr\left(\widetilde\gamma_{A_2}\widetilde\gamma_{A_3}O\right)\tr\left(\widetilde\gamma_{A_3}\widetilde\gamma_{A_1}O^\dagger\right).
\end{align*}
\end{proof}
We further analyse this variance bound for particular observables of interest in Section~\ref{sec:variance}.

\section{Efficient post-processing of matchgate shadows} \label{sec:compute}

In this section, we start by proving Theorems~\ref{prop:densityoperators} and~\ref{prop:overlaps} in subsection~\ref{sec:compute_densityoperators} and~\ref{sec:compute_overlaps}, which give explicit expressions that can be used to efficiently evaluate fidelities with Gaussian states as well as overlaps with Slater determinants via matchgate shadows. In subsection~\ref{sec:compute_general}, we then present a general procedure for efficiently estimating the expectation values of a broader family of observables, including those that yield overlaps with arbitrary pure Gaussian states. 

\subsection{Estimating fidelities with fermionic Gaussian states} \label{sec:compute_densityoperators}

In this subsection, we state and prove the general form of Theorem~\ref{prop:densityoperators} (stated for the special case of invertible $C_{\varrho_1}$ in \ref{sec:summary_densityoperators}) , which allows us to efficiently compute the expectation value $\tr(\varrho \mathcal{M}^{-1}(U_Q^\dagger\ket{b}\bra{b}U_Q))$ of any fermionic Gaussian density operator $\varrho$ with respect to any matchgate shadow sample $\mathcal{M}^{-1}(U_Q^\dagger\ket{b}\bra{b}U_Q)$ (see Eq.~\eqref{estimatedensityoperators}). Recall that $\tr(\varrho \mathcal{M}^{-1}(U_Q^\dagger\ket{b}\bra{b}U_Q))$ gives an unbiased estimate of $\tr(\varrho \rho)$, allowing us to use our classical shadows to efficiently estimate the fidelities between the unknown state $\rho$ and any Gaussian states. We will later bound the variance of these estimates in subsection~\ref{sec:variance_densityoperators}.
\begingroup
\def\thetheorem{2}
\begin{theorem}[For arbitrary $C_{\varrho_1}$] \label{prop:densityoperators}
Let $\varrho_1$ and $\varrho_2$ be density operators of $n$-mode fermionic Gaussian states $\mathrm{(\text{Eq.}~\eqref{varrhodef})}$, with covariance matrices $C_{\varrho_1}$ and $C_{\varrho_2}$ $\mathrm{(\text{Eq.}~\eqref{covariancematrix})}$. Let $2r$ be the rank of $C_{\varrho_1}$, and let $Q_1\in \orth$ be any orthogonal matrix and $C_{\varrho_1'}$ any invertible $2r\times 2r$ matrix such that 
\begin{equation} \label{Cvarrho1'} C_{\varrho_1} = Q_1^{\mathrm{T}} \begin{pmatrix} C_{\varrho_1}' &{0} \\ {0} &0\end{pmatrix} Q_1. \end{equation} Then, for any $\ell \in \{0,\dots, n\}$, $\tr(\varrho_1\mathcal{P}_{2\ell}(\varrho_2))$ is the coefficient of $z^\ell$ in the polynomial
\begin{equation} \label{pgeneral} p_{\varrho_1,\varrho_2}(z) = \frac{1}{2^{n}}\Pf\left(C_{\varrho_1}'\right)\Pf\left(-C'^{-1}_{\varrho_1} + z(Q_1 C_{\varrho_2} Q_1^{\mathrm{T}})\Big|_{[2r]} \right) \end{equation}
(where $M\big|_{[2r]}$ denotes the matrix $M$ restricted to the first $2r$ rows and columns).\end{theorem}
\addtocounter{theorem}{-1}
\endgroup
Note from Eq.~\eqref{covariancegeneral} that $Q_1$ and $C_{\varrho_1}'$ always exist. In the case where $C_{\varrho_1}$ is invertible, we can take $Q_1 = I$ and $C_{\varrho_1}' = C_{\varrho_1}$, and Eq.~\eqref{pgeneral} reduces to Eq.~\eqref{pinvertible}.

$p_{\varrho_1,\varrho_2}(z)$ is a polynomial of degree at most $r$. Its coefficients can be calculated via polynomial interpolation, which entails evaluating $p_{\varrho_1,\varrho_2}(z)$ at $r+1$ values of $z$. From Eq.~\eqref{pgeneral}, each evaluation involves computing the Pfaffian of two $2r\times 2r$ matrices, which takes $\mathcal{O}(r^3)$ time. Thus, the total runtime of this approach is $\mathcal{O}(r^4)$. 
In Appendix~\ref{app:linearPfaffian}, we give an alternative procedure for computing all of the coefficients of $p_{\varrho_1,\varrho_2}(z)$ in $\mathcal{O}(r^3)$ time. 
Therefore, for any Gaussian state $\varrho$ with covariance matrix $C_{\varrho}$ of rank $r$, we can compute $\tr(\varrho\mathcal{M}^{-1}(U_Q^\dagger \ket{b}\bra{b} U_Q))$ (for any matchgate circuit $U_Q$ and computational basis state $\ket{b}$) in $\mathcal{O}(r^3)$ time, by finding the coefficients in the polynomial
\[ p_{\rho, U_Q^\dagger \ket{b}\bra{b}U_Q}(z) = \frac{1}{2^n}\Pf\left(-C_{\rho}'\right)\Pf\left(-C_{\varrho_1}'^{-1} + z (Q_1 Q^{\mathrm{T}} C_{\ket{b}}Q Q_1^{\mathrm{T}}) \Big|_{[2r]} \right), \]
then substituting the coefficient of $z^\ell$ for $\tr(\varrho \mathcal{P}_{2\ell}(U_Q^\dagger\ket{b}\bra{b}U_Q))$, for each $\ell \in \{0,\dots, r\}$, into Eq.~\eqref{estimatedensityoperators}.

We prove Theorem~\ref{prop:densityoperators} (as well as Theorem~\ref{prop:overlaps}, in the next subsection) by interpreting Majorana operators as generators of a Clifford algebra and using basic Clifford algebra properties, which we briefly summarise here. We refer the reader to Ref.~\citenum{Hestenes_1984} for the definition of a Clifford algebra as well as the proofs of the identities we use. An important property of a Clifford algebra is that it is \textit{graded}: any element $A$ of the algebra can be written as
\[ A = \gr{A}_0 + \gr{A}_1 + \gr{A}_2 + \dots = \sum_{k \geq 0}\gr{A}_k, \]
where $\gr{A}_k$ is the \textit{grade-$k$} or \textit{$k$-vector} part of $A$. If $A = \gr{A}_k$ for some nonnegative integer $k$, then $A$ is said to be \textit{homogeneous} of grade $k$ and is called a \textit{$k$-vector}. $0$-vectors, $1$-vectors, and $2$-vectors are referred to as \textit{scalars}, \textit{vectors}, and \textit{bivectors}, respectively. Scalars are identified with elements of the underlying field (in our case, $\mathbb{C}$), while for $k \geq 1$, a $k$-vector is the sum of \textit{$k$-blades} (or \textit{simple $k$-vectors}), where $A$ is a $k$-blade if and only if $A = a_1a_2\dots a_k$
for some vectors $a_1,a_2,\dots, a_k$ such that $a_i a_j = -a_j a_i$ for all $i,j \in [k]$ with $i \neq j$. Every vector squares to a scalar, i.e., $a^2 = \gr{a^2}_0$ for any $1$-vector $a$. The grade operation $\gr{\, \cdot\, }_k$ has the properties $\gr{A + B}_k = \gr{A}_k + \gr{B}_k$, $\gr{c A}_k = c\gr{A}_k = \gr{A}_k c$ if $c = \gr{c}_0$, and $\gr{\gr{A}_k}_l = \delta_{kl} \gr{A}_k$.  We will often make use of the following identity: for arbitrary elements $A$ and $B$,
\begin{equation} \label{AB0}
    \gr{AB}_0 = \sum_{k}\gr{\gr{A}_k\gr{B}_k}_0.
\end{equation} 

Thus, our space $\LHn$ of $n$-qubit operators can be viewed as a representation of a complex $2^{2n}$-dimensional Clifford algebra $\mathcal{C}_{2n}$ (with the addition and multiplication operations of the algebra corresponding to operator addition and multiplication). For any set of Majorana operators $\{\widetilde\gamma_\mu\}_{\mu \in [2n]}$, each $\widetilde{\gamma}_\mu$ represents a $1$-vector; then, since Majorana operators mutually anticommute, products of $k$ Majorana operators, i.e., $\widetilde{\gamma}_S$ for some $S \in {[2n]\choose k}$, represent $k$-blades. Thus, the grade-$k$ part of any operator is its projection onto the subspace $\Gamma_k$, and $k$-vectors are the elements of $\Gamma_k$. In particular, any scalar $c$ in the Clifford algebra is represented by a multiple of the identity operator: $c \equiv cI$. Hence, using the same notation for Clifford algebra elements and their representations, we have
\begin{equation} \label{trtogr0} \tr(A) = 2^{n} \gr{A}_0 \end{equation}
for $A \in \LHn \cong \mathcal{C}_{2n}$. Another useful observation is that for any $k$-vector $A_k$, $\gamma_{[2n]}A_k$ is of grade $2n-k$, so for any element $A$, we have $\gr{\gamma_{[2n]}A}_{2n} = \gamma_{[2n]}\gr{A}_0$. Since $\gamma_{[2n]}^\dagger \gamma_{[2n]} = I \equiv 1$, we can multiply both sides by $\gamma_{[2n]}^\dagger$ to obtain
\begin{equation} \label{pseudoscalarduality} \gr{A}_0 = \gamma_{[2n]}^\dagger \gr{\gamma_{[2n]}A}_{2n}.\end{equation}

It will also be convenient to introduce notation for the \textit{wedge product} on the Clifford algebra, which is defined in terms of the multiplication operation of the algebra and the grade operation as follows. For homogeneous elements $A_k = \gr{A_k}_k$ and $B_l = \gr{B_l}_l$, the wedge product is defined by 
\begin{equation} \label{wedgedef1} A_k \wedge B_l \coloneqq \gr{A_k B_l}_{k+l}\end{equation} 
Then, for arbitrary elements $A$ and $B$,
\begin{equation} \label{wedgedef2} A \wedge B \coloneqq \sum_k \sum_l \gr{A}_k \wedge \gr{B}_l. \end{equation}
The wedge is associative, distributive with respect to addition, and antisymmetric under exchange of vectors. From its definition, we see for instance that $\widetilde{\gamma}_{\mu_1} \dots \widetilde{\gamma}_{\mu_k} = \widetilde{\gamma}_{\mu_1} \wedge \dots \wedge \widetilde{\gamma}_{\mu_k}$, for any set of Majorana operators $\{\widetilde{\gamma}_\mu\}_{\mu \in [2n]}$. We adopt the convention that in the absence of brackets, the wedge product is performed first (i.e., $AB \wedge C = A (B \wedge C)$).

Finally, we state the following fact, which expresses the coefficient of $\gamma_{[2n]} = \gamma_1 \wedge \dots \wedge \gamma_{2n}$ in the $n$th wedge power of an arbitrary bivector $B \in \Gamma_2$---or more generally, the coefficient of $e_1 \wedge \dots \wedge e_{2k} $ in $B^{\wedge k}$, where $B$ is a bivector in the subalgebra generated by arbitrary 1-vectors $e_1,\dots, e_k$---in terms of a Pfaffian. This is easily proven using the definition of the Pfaffian and the antisymmetry of the wedge product.
\begin{fact} \label{fact:bivectorPfaffian}
    Let $e_1,\dots, e_{2k}$ be arbitrary 1-vectors. For any bivector $B = \frac{1}{2} \sum_{\mu,\nu =1}^{2k} M_{\mu\nu}e_\mu \wedge e_\nu$, where $M$ is an antisymmetric $2k \times 2k$ matrix,  we have\[ B^{\wedge k} = k!\, \Pf(M) e_1 \wedge \dots \wedge e_{2k}. \]
\end{fact}

We are now ready to prove Theorem~\ref{prop:densityoperators}.
\begin{proof}[Proof of Theorem~\ref{prop:densityoperators}]
Let $\varrho_1 = \prod_{j=1}^n \frac{1}{2}\left(I - i\lambda_j\gamma'_{2j-1}\gamma'_{2j}\right)$ and $\varrho_2 = \prod_{j=1}^n \frac{1}{2}\left(I - i\lambda_j''\gamma''_{2j-1}\gamma''_{2j}\right)$,
where $\lambda_j,\lambda_j'' \in [-1,1]$, and $\{\gamma'_\mu\}_{\mu \in [2n]}$ and $\{\gamma''_{\mu}\}_{\mu\in[2n]}$ are sets of Majorana operators: $\gamma_{\mu}' = U_{Q'}^\dagger\gamma_{\mu}U_{Q'}$ and $\gamma_{\mu}'' = U_{Q''}^\dagger \gamma_{\mu} U_{Q''}$ for some $Q',Q'' \in \orth$. To find $\tr(\varrho_1 \mathcal{P}_{2\ell}(\varrho_2))$, observe that if we expand the expression for $\varrho_2$, the projector $\mathcal{P}_{2\ell}$ picks out products of precisely $\ell$ of the bivectors $\lambda_j''\gamma_{2j-1}''\gamma_{2j}''$. Therefore,  $\tr(\varrho_1\mathcal{P}_{2\ell}(\varrho_2))$ is the coefficient of $z^\ell$ in
\begin{equation} \label{prop1proof1} \tr\left(\left[\prod_{j=1}^n \frac{1}{2}\left(I - i\lambda_j\gamma'_{2j-1}\gamma'_{2j}\right)\right] \left[\prod_{k=1}^n \frac{1}{2}\left(I - iz\lambda_k''\gamma''_{2k-1}\gamma''_{2k}\right) \right]\right) \eqqcolon p_{\varrho_1,\varrho_2}(z). \end{equation}

For convenience, we define $\widetilde{\gamma}_\mu \coloneqq U_{Q'}\gamma_\mu'' U_{Q'}^\dagger$ and let $\widetilde{\lambda}_j \coloneqq z \lambda_j''$, so that we can use the cyclic property of the trace to write $p_{\varrho_1,\varrho_2}(z) = \tr([\prod_j \frac{1}{2}(I - i\lambda_j\gamma_{2j-1}\gamma_{2j})][\prod_k \frac{1}{2} (I - i \widetilde{\lambda}_k \widetilde{\gamma}_{2k-1}\widetilde{\gamma}_{2k})])$. We will also use the following notation for the simple bivectors $\gamma_{2j-1}\gamma_{2j}$ and $\widetilde{\gamma}_{2j-1}\widetilde{\gamma}_{2j}$, and products thereof: for any $j \in [n]$ and $T \subseteq [n]$,
\begin{equation} \label{betadef}
\begin{aligned} &\beta_j \coloneqq \gamma_{2j-1}\gamma_{2j}, \qquad &&\widetilde{\beta}_j \coloneqq \widetilde{\gamma}_{2j-1}\widetilde{\gamma}_{2j}, \\
&\beta_T \coloneqq \prod_{j\in T} \beta_j, \qquad &&\widetilde{\beta}_T \coloneqq \prod_{j \in T} \widetilde{\beta}_j
\end{aligned}
\end{equation}
for any $T \subseteq [n]$. Note that $\beta_T$ and $\widetilde{\beta}_T$ are blades of grade $2|T|$. Analogously, denote products of $\lambda_j$ and $\widetilde{\lambda}_j$ by
\[ \lambda_T \coloneqq \prod_{j \in T} \lambda_j, \quad \widetilde{\lambda}_T \coloneqq \prod_{j \in T}\widetilde{\lambda}_j. \] Thus, we can rewrite Eq.~\eqref{prop1proof1} as 
\begin{align*}
    p_{\varrho_1,\varrho_2}(z) &= \tr\left(\left[\prod_{j = 1}^n \frac{1}{2}(I - i\lambda_j \beta_j)\right] \left[\prod_{j' = 1}^n\frac{1}{2} (I - i\widetilde{\lambda}_{j'}\widetilde{\beta}_{j'} \right] \right) \\
    &= 2^n \gr{\frac{1}{2^{2n}}\sum_{T \subseteq [n]}(-i)^{|T|} \lambda_T \beta_T \sum_{T'\subseteq[n]}(-i)^{|T'|}\widetilde{\lambda}_{T'}\widetilde{\beta}_{T'}}_0,
\end{align*}
using Eq.~\eqref{trtogr0} in the second equality.

To illustrate the main ideas, we first consider the case where $\lambda_j \neq 0$ for all $j \in [n]$ (i.e., $C_{\varrho_1}$ is invertible) for simplicity, before extending the proof to the general case. Using the identities in Eqs.~\eqref{AB0} and~\eqref{pseudoscalarduality}, we have
\begin{align}
    p_{\varrho_1,\varrho_2}(z) &= \frac{1}{2^n} \sum_k \gr{\gr{\sum_{T \subseteq [n]}(-i)^{|T|} \lambda_T \beta_T}_k \gr{\sum_{T' \subseteq [n]}(-i)^{|T'|}\widetilde{\lambda}_{T'}\widetilde{\beta}_{T'}}_k}_0 \nonumber \\
    &= \frac{1}{2^n} \sum_{\ell =0}^n \gr{\sum_{T \in {[n]\choose \ell}} (-i)^\ell \lambda_T \beta_T \sum_{T' \in {[n]\choose \ell}} (-i)^{\ell}\widetilde{\lambda}_{T'}\widetilde{\beta}_{T'}}_0 \nonumber \\
    &= \frac{1}{2^n}\sum_{\ell=0}^n \gamma_{[2n]}^\dagger \gr{\gamma_{[2n]}(-1)^\ell \sum_{T \in {[n]\choose \ell}} \lambda_T \beta_T \sum_{T'\in {[n]\choose \ell}}\widetilde{\lambda}_{T'} \widetilde{\beta}_{T'}}_{2n}, \label{prop1proofbeforeflip}
\end{align}
where the second line follows from the fact that $\beta_T$ and $\widetilde{\beta}_T$ are of grade $2|T|$. Noting that $\gamma_{[2n]} = \beta_{[n]}$, we have 
\begin{equation} \label{prop1proofflip} \gamma_{[2n]} \beta_T = (-1)^{|T|} \beta_{[n] \setminus T},\end{equation}
since the $\beta_j$ commute and $\beta_j^2 = -1$. Also, $\lambda_{T} = \lambda_{[2n]}/\lambda_{[n]\setminus T}$, so relabelling $T \to [n] \setminus T$ in the second sum gives 
\begin{align}
    p_{\varrho_1,\varrho_2}(z) &= \frac{1}{2^n} \lambda_{[2n]} \gamma_{[2n]}^\dagger \sum_{\ell = 0}^n \gr{\sum_{T \in {[n]\choose n -\ell}} \frac{1}{\lambda_T} \beta_T \sum_{T' \in {[n]\choose \ell}}\widetilde{\lambda}_{T'}\widetilde{\beta}_{T'}}_{2n} \label{prop1proofinter0} \\
    &= \frac{1}{2^n}\lambda_{[2n]}\gamma_{[2n]}^\dagger \sum_{\ell =0}^n \sum_{T \in {[n]\choose n-\ell}} \sum_{T' \in {[n]\choose \ell}} \left(\frac{1}{\lambda_T} \beta_T\right) \wedge (\widetilde\lambda_{T'} \widetilde{\beta}_{T'}) \label{prop1proofinter}
\end{align}
by definition of the wedge product (Eq.~\eqref{wedgedef1}). Now, since any two bivectors commute with respect to the wedge product and $\beta_j \wedge \beta_j = \widetilde{\beta}_j \wedge \widetilde{\beta}_j = 0$, it follows from the multinomial formula that 
\begin{equation} \label{prop1proofmultinomial} \left(\sum_{j =1}^n \frac{1}{\lambda_j} \beta_j + \sum_{j = 1}^{n} \widetilde{\lambda}_j \widetilde{\beta}_j \right)^{\wedge n} =n! \sum_{\ell =0}^n \sum_{T \in {[n]\choose n-\ell}} \sum_{T' \in {[n]\choose \ell}} \left(\frac{1}{\lambda_T} \beta_T\right) \wedge (\widetilde\lambda_{T'} \widetilde{\beta}_{T'}). \end{equation}
Finally, we write the bivector on the LHS as a linear combination of $\gamma_\mu\gamma_\nu$, so that we can apply Fact~\ref{fact:bivectorPfaffian}. Letting 
\[ \Lambda_1 \coloneqq \bigoplus_{j=1}^n \begin{pmatrix} 0 &\lambda_j \\ -\lambda_j &0 \end{pmatrix}, \qquad \Lambda_2  \coloneqq \bigoplus_{j=1}^n \begin{pmatrix} 0 &\lambda_j'' \\ -\lambda_j'' &0 \end{pmatrix},  \]
we have (recalling $\widetilde{\lambda}_j = z\lambda_j''$)
\begin{align*}
    \left(\sum_{j =1}^n \frac{1}{\lambda_j} \beta_j + \sum_{j = 1}^{n} \widetilde{\lambda}_j \widetilde{\beta}_j\right)^{\wedge n} &= \left(\frac{1}{2}\sum_{\mu, \nu \in [2n]} (-\Lambda_1^{-1} + zQ' Q''^{\mathrm{T}} \Lambda_2 Q'' Q'^{\mathrm{T}})_{\mu\nu} \gamma_\mu \wedge \gamma_\nu\right)^{\wedge n} \\
    &= n! \Pf(-\Lambda_1^{-1} + zQ'Q''^{\mathrm{T}} \Lambda_2 Q''Q'^{\mathrm{T}}) \gamma_{[2n]} 
\end{align*}
using Fact~\ref{fact:bivectorPfaffian}. We also have $\lambda_{[2n]} = \Pf(\Lambda_1)$. Inserting these into Eq.~\eqref{prop1proofinter} and noting that $Q'^{\mathrm{T}} \Lambda_1 Q' = C_{\varrho_1}$ and $Q''^{\mathrm{T}} \Lambda_2 Q'' = C_{\varrho_2}$, we arrive at 
\begin{align*}
    p_{\varrho_1,\varrho_2}(z) &= \frac{1}{2^n} \Pf\left(\Lambda_1\right) \Pf\left(-\Lambda_1^{-1} + zQ' Q''^{\mathrm{T}} \Lambda_2 Q'' Q'^{\mathrm{T}}\right)\gamma_{[2n]}^\dagger\gamma_{[2n]} \\
    &= \frac{1}{2^n} \Pf\left(Q'C_{\varrho_1}Q'^{\mathrm{T}}\right)\Pf\left(Q'(-C_{\varrho_1}^{-1} + z C_{\varrho_2})Q'^{\mathrm{T}} \right)
\end{align*}
for invertible $C_{\varrho_1}$. We can then write this as $p_{\varrho_1,\varrho_2}(z) = 2^{-n}\Pf(C_{\varrho_1}') \Pf(-C_{\varrho_1}'^{-1} + zQ_1 C_{\varrho_2}Q_1^{\mathrm{T}})$ for any $C_{\varrho_1}' = Q_1 C_{\varrho_1} Q_1^{\mathrm{T}}$ with $Q_1 \in \orth$ using the Pfaffian identity
\begin{equation} \label{PfBABT} \Pf(BAB^{\mathrm{T}}) = \det(B) \Pf(A). \end{equation}

For the general case, let $J \coloneqq \{j\in [n]: \lambda_j \neq 0\}$. From Eq.~\eqref{covariancegeneral}, $|J| = r$, where $2r$ is the rank of $C_{\varrho_1}$. The equations up to Eq.~\eqref{prop1proofbeforeflip} still hold, except we can sum over only subsets $T$ such that $\lambda_T \neq 0$---that is, subsets $T \subseteq J$:
\[ p_{\varrho_1,\varrho_2}(z) = \frac{1}{2^n} \sum_{\ell=0}^{|J|} \gamma_{[2n]}^\dagger \gr{ \gamma_{[2n]}(-1)^\ell \sum_{T \in {J\choose \ell}} \lambda_T \beta_T \sum_{T' \in {[n]\choose \ell}} \widetilde{\lambda}_{T'} \widetilde{\beta}_{T'}}_{2n}. \] 
Next, instead of using Eq.~\eqref{prop1proofflip}, we write $\gamma_{[2n]} = \beta_{[n] \setminus J}\beta_J$ and observe that $\beta_J \beta_T = (-1)^{|T|} \beta_{J\setminus T}$ for any $T \subseteq J$. Also, $\lambda_T = \lambda_J/\lambda_{J\setminus T}$ for $T \subseteq J$, so Eqs.~\eqref{prop1proofinter0} and~\eqref{prop1proofinter} change to
\begin{align}
    p_{\varrho_1,\varrho_2}(z) &= \frac{1}{2^n} \lambda_J \gamma_{[2n]}^\dagger \sum_{\ell = 0}^{|J|} \gr{\beta_{[n] \setminus J}\sum_{T \in {J\choose |J| -\ell}} \frac{1}{\lambda_T} \beta_T \sum_{T' \in {[n]\choose \ell}}\widetilde{\lambda}_{T'}\widetilde{\beta}_{T'}}_{2n} \nonumber\\
    &= \frac{1}{2^n}\lambda_J\gamma_{[2n]}^\dagger \beta_{[n]\setminus J} \wedge \sum_{\ell =0}^{|J|} \sum_{T \in {J\choose |J|-\ell}} \sum_{T' \in {[n]\choose \ell}}  \left(\frac{1}{\lambda_T} \beta_T\right) \wedge (\widetilde\lambda_{T'} \widetilde{\beta}_{T'}), \label{prop1proofgeneral1}
\end{align}
again using Eq.~\eqref{wedgedef1} to obtain the second equality. The generalisation of Eq.~\eqref{prop1proofmultinomial} is
\[ \left(\sum_{j \in J} \frac{1}{\lambda_j} \beta_j + \sum_{j = 1}^{n} \widetilde{\lambda}_j \widetilde{\beta}_j \right)^{\wedge |J|} =(|J|)! \sum_{\ell =0}^{|J|} \sum_{T \in {J\choose |J|-\ell}} \sum_{T' \in {[n]\choose \ell}} \left(\frac{1}{\lambda_T} \beta_T\right) \wedge (\widetilde\lambda_{T'} \widetilde{\beta}_{T'}), \]
and we can write the bivector on the LHS as $\frac{1}{2}\sum_{\mu,\nu \in \text{pairs}(J)} (-\Lambda_{1}^{-1})_{\mu\nu} \gamma_\mu\gamma_\nu + \frac{1}{2}\sum_{\mu,\nu \in [2n]} z(Q'Q''^{\mathrm{T}}\Lambda_2 Q''Q'^{\mathrm{T}})_{\mu\nu}\gamma_\mu\gamma_\nu$ (where $\text{pairs}(J) \coloneqq \bigcup_{j \in J}\{2j-1,2j\}$). Substituting this into Eq.~\eqref{prop1proofgeneral1}, we have
\[ p_{\rho_1,\rho_2}(z) = \frac{1}{2^n} \lambda_J \gamma_{[2n]}^\dagger \beta_{[n] \setminus J} \wedge \frac{1}{(|J|)!} \left(-\frac{1}{2}\sum_{\mu,\nu \in \text{pairs}(J)} (\Lambda_{1}^{-1})_{\mu\nu} \gamma_\mu\gamma_\nu + \frac{1}{2}z \sum_{\mu,\nu \in [2n]} (Q'Q''^{\mathrm{T}}\Lambda_2 Q''Q'^{\mathrm{T}})_{\mu\nu}\gamma_\mu \gamma_\nu \right)^{\wedge |J|} \]
Now, the key additional step for this general case is to observe that $\beta_{[n]\setminus J} \wedge (\gamma_\mu \gamma_\nu) = \beta_{[n]\setminus J} \wedge \gamma_\mu \wedge \gamma_\nu$ is equal to zero whenever $\mu$ and/or $\nu$ are in $\text{pairs}([n]\setminus J)$, due to the antisymmetry of the wedge product. This implies that we can restrict the sum in the second term to over $\mu, \nu \in \text{pairs}(J)$ as well. We can then apply Fact~\ref{fact:bivectorPfaffian} to obtain
\begin{align*}
    p_{\varrho_1,\varrho_2}(z) &= \frac{1}{2^n} \lambda_J \gamma_{[2n]}^\dagger \beta_{[n]\setminus J} \wedge \left[\Pf\left(-\Lambda_1^{-1} \Big|_{\text{pairs}(J)} + zQ'Q''^{\mathrm{T}} \Lambda_2 Q''Q'^{\mathrm{T}}\Big|_{\text{pairs}(J)}\right) \gamma_{\text{pairs}(J)} \right] \\
    &= \frac{1}{2^n} \Pf\left(\Lambda_1\Big|_{\text{pairs}(J)}\right)\Pf\left(-\Lambda_1^{-1} \Big|_{\text{pairs}(J)} + zQ'C_{\varrho_2}Q'^{\mathrm{T}}\Big|_{\text{pairs}(J)}\right),
\end{align*}
since $\beta_{[n] \setminus J} \wedge \gamma_{\text{pairs}(J)} = \beta_{[n]\setminus J}\wedge \beta_J  = \gamma_{[2n]}$. Theorem~\ref{prop:densityoperators} follows by noting that any $C_{\varrho_1}'$ satisfying Eq.~\eqref{Cvarrho1'} is related to $\Lambda_1\big|_{\text{pairs}(J)}$ by an orthogonal basis change, and using the Pfaffian identity in Eq.~\eqref{PfBABT}.
\end{proof}

As a side remark, taking $z = 1$ in Theorem~\ref{prop:densityoperators} yields the general formula
\begin{equation} \label{trvarrho1varrho2} \tr(\varrho_1 \varrho_2) = \frac{1}{2^n} \Pf(-C_{\varrho_1}')\Pf\left(C_{\varrho_1}'^{-1} + (Q_1 C_{\varrho_2}Q_1^{\mathrm{T}})\Big|_{[r]}\right)\end{equation}
for the overlap between two arbitrary Gaussian states (since $p_{\varrho_1,\varrho_2}(1) = \sum_{\ell=0}^{2n} \tr(\varrho_1\mathcal{P}_{2\ell}(\varrho_2))$ and $\sum_{\ell=0}^{2n}\mathcal{P}_{2\ell}(\varrho_2) = \mathcal{P}_{\text{even}}(\varrho_2) = \varrho_2$). A special case of this formula, where $\varrho_1$ and $\varrho_2$ are both pure, is given in Ref.~\citenum{Bravyi2017-pw}. 

\subsection{Estimating overlaps with Slater determinants} \label{sec:compute_overlaps}

In this subsection, we prove Theorem~\ref{prop:overlaps}. As discussed in subsection~\ref{sec:summary_Slater}, Theorem~\ref{prop:overlaps} provides a method for efficiently estimating the expectation value of $\ket{\varphi}\bra{\mathbf{0}}$, where $\ket{\varphi}$ is an arbitrary Slater determinant, which in turn allows us to estimate the overlap of a pure state with $\ket{\varphi}$. 

The proof of Theorem~\ref{prop:overlaps} uses some of the ideas in the proof of Theorem~\ref{prop:densityoperators} (in the case where $C_{\varrho_1}$ is not invertible). The basic Clifford algebra properties we use in the proof are summarised in the previous subsection.

\begin{proof}[Proof of Theorem~\ref{prop:overlaps}] 
Let $\varrho = \prod_{j=1}^n \frac{1}{2}(I - i\lambda_j\gamma'_{2j-1}\gamma'_{2j})$, where $\lambda_j \in [-1,1]$ and $\gamma_{\mu}' = U_{Q'}^\dagger \gamma_\mu U_{Q'}$ for some $Q' \in \orth$. By the same reasoning as in the proof of Theorem~\ref{prop:overlaps}, $\tr(\ket{\varphi}\bra{\mathbf{0}}\mathcal{P}_{2\ell}(\varrho))$ is the coefficient of $z^\ell$ in 
\[ \tr\left(\ket{\varphi}\bra{\mathbf{0}} \left[\prod_{j=1}^n \frac{1}{2}(I - iz\lambda_j \gamma_{2j-1}'\gamma_{2j}') \right] \right) \eqqcolon q_{\ket{\varphi},\varrho}(z). \] With $\widetilde{Q}$ the orthogonal matrix defined in Eq.~\eqref{QSlater}, we can write $\ket{\varphi}\bra{\mathbf{0}}$ as
\begin{align} \label{phi0eq}
    \ket{\varphi}\bra{\mathbf{0}} &= \widetilde{a}_1^\dagger\dots \widetilde{a}_\zeta^\dagger \ket{\mathbf{0}}\bra{\mathbf{0}} = U_{\widetilde{Q}}^\dagger a_1^\dagger \dots a_\zeta^\dagger U_{\widetilde{Q}} \ket{\mathbf{0}}\bra{\mathbf{0}} = U_{\widetilde{Q}}^\dagger a_1^\dagger \dots a_\zeta^\dagger \ket{\mathbf{0}}\bra{\mathbf{0}}U_{\widetilde{Q}},
\end{align}
where in the last equality, we use the fact that $U_{\widetilde{Q}}$ conserves the number operator ($\sum_{j=1}^n \widetilde{a}_j^\dagger\widetilde{a}_j = \sum_{j=1}^n a_j^\dagger a_j$), so $U_{\widetilde{Q}}\ket{\mathbf{0}} \propto \ket{\mathbf{0}}$. Hence, 
defining $\widetilde{\gamma}_\mu \coloneqq U_{\widetilde{Q}}\gamma_{\mu}' U_{\widetilde{Q}}^\dagger$ and $\widetilde{\lambda}_j \coloneqq z\lambda_j$, we can use the cyclic property of the trace along with Eq.~\eqref{bb} to obtain
\begin{align*} q_{\ket{\varphi},\varrho}(z) = \tr\left(a_1^\dagger \dots a_\zeta^\dagger \left[\prod_{j=1}^n\frac{1}{2}(I - i\gamma_{2j-1}\gamma_{2j}) \right]\left[\prod_{j'=1}^n \frac{1}{2}(I - i\widetilde{\lambda}_{j} \widetilde{\gamma}_{2j'-1}\widetilde{\gamma}_{2j'})\right]\right).
\end{align*}

We will use the notation $\beta_j$, $\widetilde{\beta}_j$, $\beta_T$, $\widetilde{\beta}_T$ defined in Eq.~\eqref{betadef}, as well as $\widetilde{\lambda}_T \coloneqq \prod_{j \in T}\widetilde{\lambda}_j$. Note that $a_j^\dagger$ commutes with $\frac{1}{2}(I- i\beta_k)$ for all $k \neq j$ while $a_j^\dagger \frac{1}{2}(I - i\beta_j) = a_j^\dagger(I - a_j^\dagger a_j) = a_j^\dagger$. Thus,  
\begin{align*}
    q_{\ket{\varphi},\varrho}(z) &= \tr\left(a_1^\dagger \dots a_\zeta^\dagger \left[\prod_{j=\zeta + 1}^n \frac{1}{2}(I - i\gamma_{2j-1}\gamma_{2j}) \right] \left[\prod_{j'=1}^n \frac{1}{2}(I - i\widetilde{\lambda}_{j} \widetilde{\gamma}_{2j'-1}\widetilde{\gamma}_{2j'})\right] \right) \\ 
    &= \frac{1}{2^{n-\zeta}}\sum_k \gr{\gr{a_1^\dagger \dots a_\zeta^\dagger \sum_{T \subseteq [\zeta + 1:n]} (-i)^{|T|} \beta_T }_{k} \gr{\sum_{T' \subseteq[n]}(-i)^{|T'|}\widetilde{\lambda}_{T'} \widetilde{\beta}_{T'} }_k }_0,
\end{align*}
where $[\zeta + 1:n] \coloneqq \{\zeta + 1, \dots, n\}$ and we apply Eqs.~\eqref{trtogr0} and~\eqref{AB0} in the second line. Since $a_1^\dagger, \dots, a_\zeta^\dagger, \gamma_{2\zeta + 1},\dots \gamma_{2\zeta}$ mutually anticommute, $a_1^\dagger \dots a_\zeta^\dagger \beta_T$ is a blade of grade $\zeta + 2|T|$, so $\gr{a_1^\dagger \dots a_\zeta^\dagger \beta_T}_k = \delta_{k, \zeta + 2|T|} a_1^\dagger \dots a_\zeta^\dagger \beta_T$. Using this along with Eq.~\eqref{pseudoscalarduality} gives
\[ q_{\ket{\varphi},\varrho}(z) =  \frac{1}{2^{n-\zeta}} \sum_{\ell = \zeta/2}^n \gamma_{[2n]}^\dagger\gr{\gamma_{[2n]} a_1^\dagger \dots a_\zeta^\dagger \sum_{T \in {[\zeta +1:n] \choose \ell - \zeta/2}} (-i)^{\ell - \zeta/2} \beta_T \sum_{T' \in {[n]\choose \ell}}(-i)^\ell \widetilde{\lambda}_{T'} \widetilde{\beta}_{T'} }_{2n}.  \]
Next, $\gamma_{[2n]}a_1^\dagger \dots a_\zeta^\dagger = (-ia_1)^\dagger \dots (-ia_\zeta)^\dagger \gamma_{[2\zeta + 1:2n]}$ and $\gamma_{[2\zeta + 1:2n]} \beta_T = (-1)^{|T|} \beta_{[\zeta + 1:n] \setminus T}$ for any $T \subseteq [\zeta + 1:n]$, so
\begin{align*}
    q_{\ket{\varphi},\varrho}(z) &= \frac{1}{2^{n-\zeta}}  \gamma_{[2n]}^\dagger \sum_{\ell=\zeta/2}^n \gr{(-i)^\zeta a_1^\dagger \dots a_\zeta^\dagger \sum_{T \in {[\zeta + 1:n] \choose \ell - \zeta/2}}(-i)^{\ell - \zeta/2} (-1)^{\ell - \zeta/2}\beta_{[\zeta + 1:n]\setminus T} \sum_{T' \in {[n]\choose \ell}} (-i)^\ell\widetilde{\lambda}_{T'} \widetilde{\beta}_{T'} }_{2n} \\
    &= \frac{i^{\zeta/2}}{2^{n-\zeta}}  \gamma_{[2n]}^\dagger \sum_{\ell = \zeta/2}^n \gr{a_1^\dagger \dots a_\zeta^\dagger \sum_{T\in {[\zeta + 1:n] \choose n - \zeta/2 - \ell}} \beta_{T} \sum_{T' \in {[n]\choose \ell}}\widetilde{\lambda}_{T'}\widetilde{\beta}_{T'} }_{2n} \\
    &= \frac{i^{\zeta/2}}{2^{n-\zeta}}  \gamma_{[2n]}^\dagger (a_1^\dagger \dots a_\zeta^\dagger) \wedge \sum_{\ell = \zeta/2}^n \sum_{T\in {[\zeta + 1:n] \choose n - \zeta/2 - \ell}}\sum_{T' \in {[n]\choose \ell}} \beta_T \wedge (\widetilde{\lambda}_{T'} \widetilde{\beta}_{T'}) \\
    &= \frac{i^{\zeta/2}}{2^{n-\zeta}}  \gamma_{[2n]}^\dagger (a_1^\dagger \dots a_\zeta^\dagger) \wedge \frac{1}{(n -\zeta/2)!}\left(\sum_{j=\zeta + 1}^n \beta_j + \sum_{j=1}^n \widetilde{\lambda}_j \widetilde{\beta}_j \right)^{\wedge (n - \zeta/2)},
\end{align*}
where we have proceeded along the same lines as in the proof of Theorem~\ref{prop:densityoperators}---we change variables $T \to [\zeta + 1:n] \setminus T$ in the second line, apply the definition of the wedge product [Eq.~\eqref{wedgedef1}] in the third, and use the multinomial theorem in the fourth. 

Now, consider the basis of $1$-vectors 
\[ \{e_1, \dots, e_{2n}\} \coloneqq \{\sqrt{2}a_1^\dagger, \sqrt{2} a_1, \dots, \sqrt{2}a_\zeta^\dagger, \sqrt{2}a_\zeta, \gamma_{2\zeta + 1},\dots, \gamma_{2n} \}. \] It can be seen that the unitary matrix $W$ defined in Eq.~\eqref{Wdef} changes between $\{e_\mu\}_{\mu \in [2n]}$ and $\{\gamma_\mu\}_{\mu \in [2n]}$: we have $e_\mu = \sum_{\mu =1}^{2n} W_{\mu\nu}\gamma_\nu$, so $\gamma_\mu = \sum_{\mu=1}^{2n} W^\dagger_{\mu\nu} e_\nu$. Writing $q_{\ket{\varphi},\varrho}(z)$ in terms of $\{e_\mu\}_{\mu \in [2n]}$, we have 
\begin{align*}
    q_{\ket{\varphi},\varrho}(z) &= \frac{i^{\zeta/2}}{2^{n - \zeta}}   \gamma_{[2n]}^\dagger \left(\frac{e_1}{\sqrt{2}} \frac{e_3}{\sqrt{2}}  \dots \frac{e_{2\zeta -1}}{\sqrt{2}}\right) \\
    &\quad \wedge \frac{1}{(n-\zeta/2)!}\left(\frac{1}{2} \sum_{\mu,\nu \in \overline{S}_\zeta} (C_{\ket{\mathbf{0}}})_{\mu\nu} e_\mu e_\nu + \frac{1}{2} \sum_{\mu,\nu \in [2n]} z(W^* \widetilde{Q} C_{\varrho} \widetilde{Q}^{\mathrm{T}} W^\dagger)_{\mu\nu} e_\mu
    e_\nu \right)^{\wedge (n - \zeta/2)},
\end{align*}
where $\overline{S}_\zeta \coloneqq [2n] \setminus \{1,3,\dots, 2\zeta - 1\}$. 
Finally, by the antisymmetry of the wedge product, $(e_1 e_3 \dots e_{2\zeta -1}) \wedge (e_\mu e_\nu) = e_1 \wedge e_3 \wedge \dots \wedge e_{2\zeta - 1} \wedge e_\mu \wedge e_\nu$ vanishes whenever $\mu$ and/or $\nu$ are in $\{1,3,\dots, 2\zeta - 1\}$, so we can restrict the second sum to $\mu, \nu \in \overline{S}_\zeta$. We can then apply Fact~\ref{fact:bivectorPfaffian}, obtaining
\begin{align*}
    q_{\ket{\varphi},\varrho}(z) &= \frac{i^{\zeta/2}}{2^{n - \zeta/2}}   \gamma_{[2n]}^\dagger (e_1 e_3 \dots e_{2\zeta - 1}) \wedge \Pf\left(\left(C_{\ket{\mathbf{0}}} + zW^* \widetilde{Q} C_{\varrho}\widetilde{Q}^{\mathrm{T}} W^\dagger\right)\Big|_{\overline{S}_\zeta} \right) \bigwedge_{\mu \in \overline{S}_\zeta} e_\mu   \\
    &= \frac{i^{\zeta/2}}{2^{n - \zeta}} \Pf\left(\left(C_{\ket{\mathbf{0}}} + zW^* \widetilde{Q} C_{\varrho}\widetilde{Q}^{\mathrm{T}} W^\dagger\right)\Big|_{\overline{S}_\zeta} \right)\gamma_{[2n]}^\dagger (-1)^{\zeta(\zeta - 1)/2} \bigwedge_{\mu \in [2n]} e_\mu 
\end{align*}
Theorem~\ref{prop:overlaps} follows by noting that $\bigwedge_{\mu \in [2n]} e_\mu = (-i)^{\zeta} \gamma_{[2n]}$ and $(-1)^{\zeta(\zeta-1)/2} (-i)^\zeta = (-1)^{\zeta^2/2} = 1$ since $\zeta$ is even.
\end{proof}

\subsection{Efficient estimation of more general observables} \label{sec:compute_general}

In this subsection, we give a method for efficiently estimating the expectation values of a larger class of observables using our matchgate shadows, extending beyond the Gaussian density operators treated in Theorem~\ref{prop:densityoperators} and the $\ket{\varphi}\bra{\mathbf{0}}$ operators in Theorem~\ref{prop:overlaps}, which allow us to measure overlaps with Slater determinants. This class comprises arbitrary products of certain kinds of operators, including Gaussian density operators, Gaussian unitaries, and linear combinations of Majorana operators (i.e., $\sum_{\mu =1}^{2n} c_\mu \gamma_\mu$ for any $c_\mu \in \mathbb{C})$. A useful example of such a product is any operator of the form $\ket{\phi}\bra{\mathbf{0}} = U_{Q_\phi}\ket{\mathbf{0}}\bra{\mathbf{0}}$, where $\ket{\phi}$ is an arbitrary pure Gaussian state (not necessarily a Slater determinant) and $U_{Q_\phi}$ is a Gaussian unitary that prepares $\ket{\phi}$. The ability to efficiently estimate the expectation value of such operators allows us to efficiently estimate the overlaps between any pure state and any pure Gaussian states (see Appendix~\ref{app:overlaps_Gaussian}). 

Our method builds on ideas in Refs.~\citenum{Bravyi2008-lf} and~\citenum{Bravyi2017-pw}, which use Grassmann integration (also known as Berezin integration) to evaluate certain quantities involving fermionic operations, as well as identities in Ref.~\citenum{Caracciolo} for evaluating Grassmann integrals. Specifically, we generalise results in Appendix~A of \citen{Bravyi2017-pw}, and consider what is essentially a special case of the problem in \citen{Bravyi2008-lf}, and are therefore able to give a more explicit and streamlined protocol.\footnote{Moreover, using, e.g., Theorem~1 of \citen{Bravyi2008-lf} would lead to an asymptotically worse complexity than the one we obtain, when considering products of a constant number of observables of the kinds described above. The high-level techniques are closely related, but we benefit from directly treating a special case of their more general problem.} We also extend Theorem~A.15(d) of Ref.~\citenum{Caracciolo} to obtain an efficient method for evaluating Grassmann integrals of a certain form. 

We begin by reviewing the basics of Grassmann algebra and Grassmann integration (subsection~\ref{sec:Grassmann_review}). We then present our method, which can be outlined at a high level as follows. The objective is to evaluate $\tr(A^{(1)}\dots A^{(m)})$ where $A^{(i)}$ are elements of $\LHn$ (equivalently, the Clifford algebra $\mathcal{C}_{2n}$). 
The first main step is to express $\tr(A^{(1)}\dots A^{(m)})$ (for arbitrary $A^{(1)}, \dots, A^{(m)}$), as a Grassmann integral by associating to each Clifford algebra element $A^{(i)}$ a particular element of a $2^{2mn}$-dimensional Grassmann algebra (subsection~\ref{sec:CtoG}). We then show that the integral reduces to a fairly simple form when each $A^{(i)}$ belongs to a certain class of operators; importantly, operators that allow us to evaluate expectation values with respect to matchgate shadow samples fall into this class (subsection~\ref{sec:efficientexamples}). Finally, we construct a procedure for efficiently evaluating any integral of this form (Algorithm~\ref{alg:gMB}). 

After presenting the general method, we work through an example of applying it to efficiently evaluate the expectation value of $\ket{\phi}\bra{\mathbf{0}}$, where $\ket{\phi}$ is a Gaussian state, with respect to any matchgate shadow sample (subsection~\ref{sec:workedexample}). In conjunction with the procedure in Appendix~\ref{app:overlaps_Gaussian}, this enables us to efficiently estimate the overlap between any pure state and $\ket{\phi}$. 

\subsubsection{Review of Grassmann algebra and Grassmann integration} \label{sec:Grassmann_review}

We give a brief overview of Grassmann algebras and Grassmann integration, and refer the reader to Ref.~\citenum{Di_Francesco_1997} for a more detailed treatment. Our definitions and notational conventions are mainly based on Refs.~\citenum{Bravyi2004-fs} and~\citenum{Caracciolo}. 

For any $n \in \mathbb{Z}_{\geq 0}$, a $2^{2n}$-dimensional Grassmann algebra $\mathcal{G}_{2n}$ (over the field $\mathbb{C}$) is generated by $2n$ Grassmann variables $\theta_1,\dots, \theta_{2n}$ (for our purposes, it will suffice to consider Grassmann algebras with an even number of generators, although the same definitions hold for odd numbers of generators). The multiplication operation of the Grassmann algebra is associative and antisymmetric under the exchange of generators:
\begin{equation} \label{Grassmannproduct} \theta_\mu\theta_\nu = -\theta_\nu\theta_\mu \end{equation}
for any $\mu, \nu \in [2n]$; this implies that $\theta_\mu^2 = 0$. Any element $f(\theta) \in \mathcal{G}_{2n}$ can be written as
\begin{equation} \label{f(theta)}
f(\theta) = c_0 + \sum_{k =1}^{2n} \sum_{1 \leq \mu_1 < \dots < \mu_k \leq 2n} c_{\mu_1,\dots, \mu_k} \theta_{\mu_1}\dots \theta_{\mu_k}
\end{equation}  
for some coefficients $c_0, c_{\mu_1,\dots, \mu_k} \in \mathbb{C}$.  We say that  $f(\theta)$ is \textit{even} if $c_{\mu_1,\dots, \mu_k} = 0$ for all odd $k$. Even elements commute with all other elements in $\mathcal{G}_{2n}$. 
From these definitions, it can be observed that $\mathcal{G}_{2n}$ is equivalent to a $2^{2n}$-dimensional Clifford algebra $\mathcal{C}_{2n}$, except where the multiplication operation is taken to be the wedge product, defined by Eq.~\eqref{wedgedef1} and~\eqref{wedgedef2}. (Indeed, the properties of Grassmann integration, some of which we review below and in Appendix~\ref{app:Grassmann}, can all be interpreted in terms of Clifford algebra operations, and proved using Clifford algebra identities.)

The ``integral'' $\int d\theta_\mu$ is defined by its action 
\begin{equation} \label{Grassmannintegraldef1} \int d\theta_\mu \, 1 = 0,\qquad \int d\theta_\mu \, \theta_\nu = \delta_{\mu\nu}, \end{equation}
on scalars and generators, together with linearity over $\mathbb{C}$ and the rule
\begin{equation} \label{Grassmannintegraldef2} \int d\theta_\mu  \left(\theta_\nu f(\theta)\right) = \delta_{\mu\nu} f(\theta) - \theta_\nu \int d\theta_\mu \, f(\theta) \end{equation}
for any $f(\theta) \in \mathcal{G}_{2n}$. 
We will use the notation $\int d\theta_{\mu_1} \dots d\theta_{\mu_k} \equiv \int d\theta_{\mu_1} \dots \int d\theta_{\mu_k}$ for any $\mu_1,\dots, \mu_k \in [2n]$, and $D\theta \equiv d\theta_{2n} \dots d{\theta}_{1}$, so that $\int D\theta \equiv \int d\theta_{2n} \dots d\theta_1$. Then, note that $\int D\theta \, \theta_{\mu_1}\dots \theta_{\mu_k}$ is only nonzero if $k = 2n$ and $\{\theta_{\mu_1},\dots, \theta_{\mu_k}\} = [2n]$, in which case it is equal to $\mathrm{sgn}(\pi)$ where $\pi$ is the permutation that maps $(\mu_1,\dots, \mu_k)$ to 
$(1,\dots, 2n)$. Thus, when applied to an arbitrary element of $\mathcal{G}_{2n}$, $\int D\theta$ picks out the coefficient of $\theta_1\dots \theta_{2n}$, i.e., for $f(\theta)$ expanded as in Eq.~\eqref{f(theta)}, 
\begin{equation} \label{Dtheta} \int D\theta \, f(\theta) = c_{1,\dots, 2n}.  \end{equation}

Also, for any Grassmann variables $\theta_1,\dots, \theta_{2n}, \eta_1,\dots, \eta_{2n}$, we use $\theta \equiv \begin{pmatrix} \theta_1 &\dots &\theta_{2n}\end{pmatrix}^{\mathrm{T}}$ to denote the vector with $\theta_1, \dots, \theta_{2n}$ as its components, and likewise, $\eta \equiv \begin{pmatrix} \eta_1 &\dots &\eta_{2n}\end{pmatrix}^{\mathrm{T}}$. Then, we have, for instance,
\[ \theta^{\mathrm{T}} \eta \equiv \sum_{\mu =1}^{2n}\theta_\mu \eta_\mu, \qquad B\theta \equiv \begin{pmatrix} \sum\limits_{\mu = 1}^{2n} B_{1\mu} \theta_\mu \\ \vdots \\ \sum\limits_{\mu = 1}^{2n} B_{k\mu} \theta_\mu  \end{pmatrix}, \qquad \theta^{\mathrm{T}} M\theta \equiv \sum_{\mu,\nu=1}^{2n} M_{\mu\nu}\theta_{\mu}\theta_\nu  \]
for any $k \times 2n$ matrix $B$ and $2n\times 2n$ matrix $M$. Note from Eq.~\eqref{Grassmannproduct} that $\theta^{\mathrm{T}} \eta = -\eta^{\mathrm{T}} \theta$.

Finally, following \citen{Bravyi2004-fs}, we define the ``Grassmann representation'' of $\LHn \cong \mathcal{C}_{2n}$ as follows. 
Recall that $A \in \LHn$ can be expanded uniquely in terms of the (canonical) Majorana operators $\{\gamma_\mu\}_{\mu \in [2n]}$ as 
\begin{equation} \label{expandA} A = c_0 I + \sum_{k =1}^{2n} \sum_{1 \leq \mu_1 < \dots < \mu_k \leq 2n} c_{\mu_1,\dots, \mu_k} \gamma_{\mu_1}\dots \gamma_{\mu_k} \end{equation}
for $c_0, c_{\mu_1,\dots, \mu_k} \in \mathbb{C}$ (by Hilbert-Schmidt orthogonality, $c_0 = \tr(A), c_{\mu_1,\dots, \mu_k} = \tr(\gamma_{\mu_k}\dots\gamma_{\mu_1}A)$). The Grassmann representation $\omega(A; \theta)$ of $A$ in terms of Grassmann variables $\theta_1,\dots, \theta_{2n}$ is then given by
\begin{equation} \label{Grassman_rep} \omega(A; \theta) = c_0 + \sum_{k =1}^{2n} \sum_{1 \leq \mu_1 < \dots < \mu_k \leq 2n} c_{\mu_1,\dots, \mu_k} \theta_{\mu_1}\dots \theta_{\mu_k}, \end{equation}
which we can regard as an element of $\mathcal{G}_{2n}$, or of any Grassmann algebra whose generators include $\theta_1,\dots, \theta_{2n}$. Note that for $A, B \in \LHn$, we have $\omega(AB; \theta) = \omega(A;\theta)\omega(B;\theta)$ if and only if $AB = A \wedge B$ (where $\wedge$ denotes the wedge product in $\mathcal{C}_{2n})$. 

\subsubsection{Expressing the trace of any product of operators as a Grassmann integral} \label{sec:CtoG}

Let $A^{(1)},\dots, A^{(m)}$ be arbitrary operators in $\LHn \cong \mathcal{C}_{2n}$. To express the scalar quantity $\tr(A^{(1)}\dots A^{(m)})$ in terms of a Grassmann integral, we represent each $A^{(i)}$ in terms of a different set of $2n$ independent Grassmann variables $\theta_1^{(i)},\dots, \theta_{2n}^{(i)}$, then consider a Grassmann algebra that includes all of the $\{\theta_{\mu}^{(i)}\}_{\mu \in [2n],i \in [m]}$ as generators. We prove the following identity in Appendix~\ref{app:tr(ABCD)}.
\begin{theorem} \label{prop:tr(ABCD)}
For any even integer $m \in \mathbb{Z}_{> 0}$ and any $A^{(1)}, \dots, A^{(m)} \in \LHn$, let $\theta_1^{(1)}, \dots, \theta_{2n}^{(1)},\dots, \theta_1^{(m)}, \dots, \theta_{2n}^{(m)}$ be generators of a $2^{2mn}$-dimensional Grassmann algebra $\mathcal{G}_{2mn}$. Then, we have
\begin{align} \label{CliffordtoGrassmann}
    \tr\left(A^{(1)}\dots A^{(m)}\right) = 2^{n} (-1)^{n m(m-1)/2} \int D\theta^{(m)}\dots D\theta^{(1)}\, \omega(A^{(1)};\theta^{(1)})\dots \omega(A^{(m)}; \theta^{(m)}) \exp\Bigg(\sum_{\substack{i,j \in [m]\\i <j}} s_{ij} \theta^{(i) T} \theta^{(j)} \Bigg),
\end{align}
where $\omega(A^{(i)}; \theta^{(i)}) \in \mathcal{G}_{2mn}$ is the Grassmann representation of $A^{(i)}$ in terms of $\theta_1^{(i)},\dots, \theta_{2n}^{(i)}$ $\mathrm{(}$Eq.~\eqref{Grassman_rep}$\mathrm{)}$, and $s_{ij} \coloneqq (-1)^{i+j+1}$ for $i < j$.
\end{theorem}
Here, we use the same notational conventions as those defined in subsection~\ref{sec:Grassmann_review}, e.g., $D\theta^{(i)} \equiv d\theta_{2n}^{(i)}\dots d\theta_1^{(i)}$ and $\theta^{(i)}$ denotes the vector $\begin{pmatrix} \theta_1^{(i)} &\dots &\theta_{2n}^{(i)} \end{pmatrix}^{\mathrm{T}}$. To apply Theorem~\ref{prop:tr(ABCD)} to products of an odd number $m$ of operators, we can, for instance, express one of the operators as a product of two operators, or add in $A^{(m+1)} = I$. 

Theorem~\ref{prop:tr(ABCD)} generalises Equations~(140) and~(144) of \citen{Bravyi2017-pw}, which give the expressions for $m = 2$ and for $m=4$ in the special case where one of the operators is $(-i)^{n} \gamma_{[2n]}$.\footnote{Indeed, when $m$ is even and one of the operators (say, $A^{(1)}$) is the parity operator $P = (-i)^{n} \gamma_{[2n]}$, Eq.~\eqref{CliffordtoGrassmann} simplifies somewhat, reducing to an expression involving only $2(m-1)n$ Grassmann generators; see Appendix~\ref{app:tr(ABCD)} for details. Hence, if we have an odd number of operators, one of which can be rewritten as a product of $P$ and some operator (e.g., $\ket{\psi}\bra{\mathbf{0}} \propto P\ket{\psi}\bra{\mathbf{0}}$ for any eigenstate $\ket{\psi}$ of $P$), we can obtain a slightly simpler expression from Theorem~\ref{prop:tr(ABCD)}.} 

The purpose of Theorem~\ref{prop:tr(ABCD)} is to translate the trace of any product of operators into a Grassmann integral. The basic reason for doing so is that the multiplication operation of the Grassmann algebra, which is equivalent to the wedge product in the Clifford algebra, is easier to work with than the multiplication operation of the Clifford algebra. In this translation, we go from a $2^{2n}$-dimensional Clifford algebra to a larger Grassmann algebra, with $2n$ separate generators for each of the $m$ operators $A^{(i)}$ appearing in the trace. To see why, consider expanding each $A^{(i)}$ on the left-hand side of Eq.~\eqref{CliffordtoGrassmann} into products of Majorana operators, as in Eq.~\eqref{expandA}, and observe that the terms that end up contributing to the trace are those for which each $\gamma_\mu$ appears in an even number of the $A^{(i)}$. Now, if we were to use only $2n$ generators for \textit{all} of the different operators (i.e., work with an expression like $\omega(A^{(1)};\theta)\dots \omega(A^{(m)};\theta)$ rather than have a different $\theta^{(i)}$ for each $A^{(i)}$), these terms would not be properly accounted for in the Grassmann algebra, as the wedge product is antisymmetric---note $\gamma_\mu^2 = I$, whereas $\theta_\mu^2 = 0$. Using independent sets of generators instead preserves all of these terms when we move to the Grassmann algebra. Then, the exponential on the right-hand side of Eq.~\eqref{CliffordtoGrassmann} correctly compensates for the fact that the Grassmann integral picks out only terms that contain all of the generators (as shown by the proof in Appendix~\ref{app:tr(ABCD)}).

\subsubsection{Efficiently evaluating certain Grassmann integrals} \label{sec:efficientexamples}

We now consider Grassmann integrals of a certain, rather general form that can be efficiently evaluated, and show that for many operators of potential interest, Theorem~\ref{prop:tr(ABCD)} yields integrals that can be written in this form. 

Specifically, for any $N,K \in \mathbb{Z}_{\geq 0}$, any $K \times 2N$ matrix $B$, and any antisymmetric $2N \times 2N$ matrix $M$, we define the Grassmann integral
\begin{equation} \label{gBM} g(B, M) \coloneqq \int D\chi \, (B\chi)_1\dots (B\chi)_K \exp\left(\frac{1}{2}\chi^{\mathrm{T}} M\chi \right), \end{equation}
where $\chi_1,\dots, \chi_{2N}$ are generators of a $2^{2N}$-dimensional Grassmann algebra, $D\chi \equiv d\chi_{2N}\dots d\chi_1$, and $\chi \equiv \begin{pmatrix} \chi_1 \dots \chi_{2N}\end{pmatrix}^{\mathrm{T}}$ (so $(B\chi)_j \equiv \sum_{\mu=1}^{2N}B_{j\mu}\chi_\mu$, $\chi^{\mathrm{T}} M\chi = \sum_{\mu,\nu=1}^{2N}M_{\mu\nu}\chi_\mu\chi_\nu$). In the case where $M$ is invertible, Theorem~A.15(d) of Ref.~\citenum{Caracciolo} gives an efficiently computable expression for $g(B,M)$, namely, 
\begin{equation} \label{A.15(d)}
    g(B, M) = 
    \Pf(M) \Pf(-B M^{-1} B^{\mathrm{T}})
\end{equation}
where $\Pf(-BM^{-1}B^{\mathrm{T}}) \equiv 1$ if $K = 0$ (and the Pfaffian for any matrix of odd dimension is zero by definition, so $\Pf(-BM^{-1}B^{\mathrm{T}}) \equiv 0$ if $K$ is odd). 
We provide a recursive algorithm that efficiently evaluates $g(B,M)$ in the general case, where $M$ is not necessarily invertible, later in this subsection.
We first give several broad examples where $\tr(A^{(1)}\dots A^{(m)})$ can be expressed as such an integral, via Theorem~\ref{prop:tr(ABCD)}. 

Consider the integral on the RHS of Eq.~\eqref{CliffordtoGrassmann} for some $A^{(1)},\dots, A^{(m)} \in \LHn$. We take $N = mn$ and identify $\theta^{(i)}_{\mu}$ with $\chi_{2n(i-1) + \mu}$ for $\mu \in [2n], i \in [m]$, so that $\chi = \begin{pmatrix} \theta^{(1)T} &\dots & \theta^{(m)T}\end{pmatrix}^{\mathrm{T}} = \begin{pmatrix} \theta_1^{(1)} &\dots &\theta_{2n}^{(1)} &\dots  &\theta_1^{(m)} &\dots &\theta_{2n}^{(m)}\end{pmatrix}^{\mathrm{T}}.$ Then, note that we can write the term $\exp(\sum_{i<j} s_{ij}\theta^{(i)T}\theta^{(j)})$ in the form
\begin{equation} \label{sijpart}
\exp\Bigg(\sum_{i,j \in [m]:i <j} s_{ij} \theta^{(i) T} \theta^{(j)} \Bigg)= \exp\left(\frac{1}{2}\chi^{\mathrm{T}} S \chi\right) \end{equation} for an antisymmetric $2N \times 2N$ matrix $S$ (more explicitly, $S$ has a simple block form, with $2n\times 2n$ zero blocks along its diagonal, blocks of the form $s_{ij}\mathbbm{1}$ above the diagonal, and appropriately matching blocks below the diagonal, where $\mathbbm{1}$ is the $2n\times 2n$ identity matrix). It remains to consider the Grassmann representations $\omega(A^{(i)};\theta^{(i)})$ and show how to encode them appropriately into an integral of the form $g(M,B)$.

The simplest example is where $A^{(i)}$ is an arbitrary linear combination of Majorana operators: $A^{(i)} = \sum_{\mu = 1}^{2n} c_{\mu} \gamma_\mu$ for some $c_{\mu} \in \mathbb{C}$. Then, by Eq.~\eqref{Grassman_rep}, $\omega(A^{(i)}; \theta^{(i)}) = \sum_{\mu =1}^{2n}c_\mu \theta_\mu^{(i)}$. By putting the $2n$ coefficients $c_\mu$ in columns $2n(i-1) + 1$ through $2ni$ of the $k$th row of $B$ for some $k$, we then have $\omega(A^{(i)};\theta^{(i)}) = (B\chi)_k$, which clearly fits into the form of $g(B,M)$. Similarly, if $A^{(i)}$ is a product of linear combinations of Majorana operators that mutually anticommute, i.e., if $A^{(i)} = L_1\dots L_\ell$ where $L_p = \sum_{\mu =1}^{2n}c_{p\mu} \gamma_\mu$ and $L_pL_q = -L_q L_p$ for all $p \neq q$ (for instance, this arises when $A^{(i)} = \widetilde{\gamma}_{\mu_1}\dots \widetilde{\gamma}_{\mu_\ell}$ for some Majorana operators $\widetilde{\gamma}_{\mu}$, or $A^{(i)} = \widetilde{a}^\dagger_{j_1}\dots \widetilde{a}^\dagger_{j_\ell}$ for some creation operators $\widetilde{a}^\dagger_j$), then $A^{(i)} = L_1 \wedge \dots \wedge L_\ell$, so $\omega(A^{(i)}; \theta^{(i)}) = (\sum_{\mu_1} c_{1\mu_1} \theta_{\mu_1}^{(i)})\dots (\sum_{\mu_\ell} c_{\ell\mu_\ell} \theta_{\mu_\ell}^{(i)})$. Hence, by putting the coefficients $c_{p\mu}$ in appropriate positions in some rows $k,\dots,k + \ell-1$ of the matrix $B$, we have $\omega(A^{(i)};\theta^{(i)}) = (B\chi)_{k}\dots (B\chi)_{k+ \ell-1}$. 

Next, suppose $A^{(i)}$ is a fermionic Gaussian density operator $\varrho = \prod_{j=1}^n \frac{1}{2}(I - i\lambda_{j} \widetilde{\gamma}_{2j-1}\widetilde{\gamma}_{2j})$, where $\widetilde{\gamma}_\mu = \sum_{\nu=1}^{2n} Q_{\mu\nu}\gamma_\mu$ for some $Q \in \orth$ (see Eq.~\eqref{varrhodef}). Since the $\widetilde{\gamma}_\mu$ mutually anticommute, $\varrho = 2^{-n}\bigwedge_{j \in [n]}^{2n}(I - i \lambda_j \widetilde{\gamma}_{2j-1} \wedge \widetilde{\gamma}_{2j})$, so 
\begin{align} \label{omegavarrho}
    \omega(\varrho; \theta^{(i)}) &= \frac{1}{2^n} \prod_{j=1}^{n} \left(1 - i \lambda_j \widetilde{\theta}_{2j-1}^{(i)} \widetilde{\theta}_{2j}^{(i)}\right) = \frac{1}{2^n} \prod_{j=1}^n \exp\left(-i \lambda_j \widetilde{\theta}_{2j-1}^{(i)} \widetilde{\theta}_{2j}^{(i)}\right) = \frac{1}{2^n}\exp \left(-\frac{i}{2}\theta^{(i)T} C_{\varrho} \theta^{(i)} \right),
\end{align}
where we let $\widetilde{\theta}^{(i)}_\mu \coloneqq \sum_{\nu=1}^{2n}Q_{\mu\nu} \theta^{(i)}_\nu$. The second inequality follows from the fact that $(\widetilde{\theta}^{(i)}_{2j-1}\widetilde{\theta}_{2j}^{(i)})^k = 0$ for any $k > 1$, so $\exp(-i \lambda_j \widetilde{\theta}_{2j-1}^{(i)} \widetilde{\theta}_{2j}^{(i)}) = 1 - i\lambda_j \widetilde{\theta}_{2j-1}^{(i)} \widetilde{\theta}_{2j}^{(i)}$, and the third uses the fact that $\widetilde{\theta}_{2j-1}^{(i)} \widetilde{\theta}_{2j}^{(i)}$ are even and hence mutually commute, along with the definition of the covariance matrix $C_\varrho$ of $\varrho$ (Eq.~\eqref{covariancematrix}). Thus, we can write $\omega(\varrho;\theta^{(i)}) = \exp(\frac{1}{2} \chi^{\mathrm{T}} C^{(i)} \chi)$ where $C^{(i)}$ is a $2N \times 2N$ matrix with $-iC_\varrho$ as one of its blocks (rows and columns $2n(i-1) + 1$ through $2ni$) and all other entries $0$. Furthermore, observe that Eq.~\eqref{omegavarrho} does not require $\lambda_j \in [-1,1]$ (as is the case for a Gaussian density operator), so we can still write $A^{(i)} = \exp(\frac{1}{2}\chi^{\mathrm{T}} C^{(i)}\chi)$ for some $C^{(i)}$ whenever it has the form $\prod_{j=1}^n \frac{1}{2}(I - i\lambda_j \widetilde{\gamma}_{2j-1}\widetilde{\gamma}_{2j})$ for \textit{some} $\lambda_j \in \mathbb{C}$, even if it is not a Gaussian density operator per se. This observation will be useful for, for instance, evaluating expectation values with respect to matchgate shadow samples (see the following subsection). 

A final example we consider is where one or more of the operators in the expression $\tr(A^{(1)}\dots A^{(m)})$ we want to evaluate is a fermionic Gaussian unitary $U_Q \in \mathrm{M}_n$. We suppose that $U_Q$ is specified by the orthogonal matrix $Q \in \orth$ that it corresponds to via Eq.~\eqref{UQdef}, along with $\det(U_Q)$. Hence, if $\det(Q) = 1$, i.e., $Q \in \mathrm{SO}(2n)$, we can write $U_Q = e^{i\alpha} U_R$, where $\det(U_R) = 1$, $R = Q$, and $e^{i\alpha}$ is a global phase, determined by $\det(U_Q)$. If $\det(Q) = -1$, we can write $U_Q = e^{i\alpha}\gamma_1 U_R$, where $R = R_1 Q$ with $R_1 = \text{diag}(1,-1,\dots, -1)$ the $2n \times 2n$ reflection matrix that $\gamma_1$ corresponds to via Eq.~\eqref{UQdef}, $\det(U_R) = 1$, and $e^{i\alpha}$ is a global phase. In both cases, $R \in \mathrm{SO}(2n)$, and it can be shown (see e.g., \citen{Jozsa2008-cm}) that $U_R$ is generated by a quadratic Hamiltonian: $U_R = \exp(-iH)$ with 
\begin{equation} \label{quadraticHamiltonian} H = \frac{i}{2}\sum_{\mu,\nu=1}^{2n} h_{\mu\nu} \gamma_\mu \gamma_\nu, \end{equation}
where $h$ is the $2n\times 2n$ antisymmetric matrix such that $R = e^{2h}$, i.e., $h = \ln(R)/2$. Block-diagonalising $h$ using an orthogonal matrix $Q' \in \orth$ as $h = Q'^{\mathrm{T}}\bigoplus_{j=1}^n \begin{pmatrix} 0 &\sigma_j \\ -\sigma_j &0 \end{pmatrix}Q'$ for some $\sigma_j \in \mathbb{R}$, and defining the Majorana operators $\gamma_\mu' \coloneqq \sum_{\nu=1}^{2n} Q_{\mu\nu}' \gamma_\nu$, we then have 
\[ H = i \sum_{j=1}^n \sigma_j \gamma_{2j-1}' \gamma_{2j}', \]
so 
\begin{align*}
    U_R = \exp\left( \sum_j \sigma_j \gamma_{2j-1}'\gamma_{2j}'\right) = \prod_{j=1}^n \exp\left(\sigma_j  \gamma_{2j-1}'\gamma_{2j}'\right) = \prod_{j=1}^n \left(\cos(\sigma_j) I + \sin(\sigma_j)\gamma_{2j-1}'\gamma_{2j}' \right),
\end{align*}
where the second equality follows from the fact that $\gamma_{2j-1}'\gamma_{2j}'$ mutually commute, and the third from $(\gamma_{2j-1}'\gamma_{2j}')^2 = -I$. Since the $\gamma_\mu'$ anticommute, we can also write this as $U_R = \wedge_{j=1}^n [\cos(\sigma_j)I + \sin(\sigma_j)\gamma_{2j-1}'\gamma_{2j}']$, from which we see that $\theta_{2n}$ has the form
\[ \omega(U_R; \theta) = \prod_{j=1}^n \left(\cos(\sigma_j) + \sin(\sigma_j) \theta_{2j-1}'\theta_{2j}'\right), \]
with $\theta'_\mu \coloneqq \sum_{\nu = 1}^{2n} Q'_{\mu\nu} \theta_\nu$. Thus, if $\cos(\sigma_j) \neq 0$ for all $j \in [n]$, we have 
\begin{equation} \label{omegaUR} \omega(U_R;\theta) = c \prod_{j=1}^n \left[1 + \tan(\sigma_j) \theta'_{2j-1}\theta'_{2j}\right] = c\prod_{j=1}^n \exp\left(\tan(\sigma_j)\theta'_{2j-1}\theta'_{2j} \right) = c\exp\left(\frac{1}{2}\theta^{\mathrm{T}} T_R \theta \right), \end{equation}
where $c \coloneqq \prod_{j=1}^n \cos(\sigma_j)$ and $T_R\coloneqq Q'^{\mathrm{T}} \bigoplus_{j=1}^n \begin{pmatrix} 0 &\tan(\sigma_j) \\ -\tan(\sigma_j) &0\end{pmatrix}Q'$. If, on the other hand, $\cos(\sigma_j) = 0$ for some $j$, say, for $j \in J \subseteq [n]$, then we have 
\[ \omega(U_R; \theta) = c_J \left[\prod_{j \in J}\left(Q'\theta\right)_{2j-1}\left(Q'\theta\right)_{2j} \right]\exp\left( \frac{1}{2}\theta^{\mathrm{T}} T_{R,J} \theta\right), \]
where $c_J \coloneqq \prod_{j=1}^n \cos(\sigma_j)$ and $T_{R,J}$ is the same as $T_R$ except with $\tan(\sigma_j)$ replaced by $0$ for all $j \in J$. Hence, by placing $T_R$ or $T_{R,J}$ in the appropriate block of the larger matrix $M$ in Eq.~\eqref{gBM}, and some of the rows of $Q'$ in the appropriate positions of the larger matrix $B$ if $J \neq \varnothing$, we can incorporate $\omega(U_R; \theta)$ into a Grassmann integral of the form of $g(B,M)$ (Eq.~\eqref{gBM}). When applying Theorem~\ref{prop:tr(ABCD)} to $U_Q$, we take $A^{(i)} = U_Q = e^{i\alpha}U_R$ for some $i$ if $\det(Q) = 1$, and if $\det(Q) = -1$ we take $A^{(i)} = \gamma_1$ (then $\omega(A^{(i)};\theta^{(i)})$ is simply $\theta^{(i)}_1$) and $A^{(i+1)} = e^{i\alpha}U_R$. 

Therefore, when each $A^{(i)}$ in $\tr(A^{(1)}\dots A^{(m)})$ belongs to one of the three example categories we have considered above (products of linear combinations of Majorana operators, Gaussian density operators, and Gaussian unitaries), we can write the RHS of Eq.~\eqref{CliffordtoGrassmann} as a Grassmann integral over a product of some linear combinations of the Grassmann variables $\chi_\mu$, as well as some exponentials of the form $\exp(\frac{1}{2}\theta^{\mathrm{T}} M^{(j)} \theta)$ for some antisymmetric $2N \times 2N$ matrices $M^{(j)}$. Since these exponentials are even elements, we can commute them past any of linear combinations of $\chi_\mu$ to move them all to the right, and then combine them with the $\exp(\frac{1}{2}\chi^{\mathrm{T}} S \chi)$ from Eq.~\eqref{sijpart} as $\exp(\frac{1}{2}\theta^{\mathrm{T}} \sum_i (M^{j} + S)\theta)$. Then, setting $M = \sum_i M^{j} +S$ and constructing the matrix $B$ appropriately, we obtain a Grassmann integral $g(B,M)$ of the form in Eq.~\eqref{gBM}. We work through an explicit example of this procedure in the following subsection. 

It remains to show how to efficiently evaluate the Grassmann integral $g(B,M)$ for arbitrary $K \times 2N$ matrices $B$ and antisymmetric $2N \times 2N$ matrices $M$. We present an algorithm for doing so as Algorithm~\ref{alg:gMB}, and prove its correctness in Appendix~\ref{app:gMB}. The algorithm builds on ideas in \citen{Bravyi2017-pw}, and extends Theorem~A.15(d) in Ref.~\citenum{Caracciolo} (see Eq.~\eqref{A.15(d)}) to an efficient procedure for evaluating $g(B,M)$ even when $M$ is not invertible.

\begin{myalgorithmfloat}[tb]
\begin{mdframed}[roundcorner=3pt]
\vspace{1em}
{\small

\begin{center} 
\textsc{Algorithm 2: Evaluating the Grassmann integral $g(B,M)$}
\end{center}

\begin{tabbing}
\textbf{Inputs:} \= $K \times 2N$ matrix $B$, antisymmetric $2N \times 2N$ matrix $M$ (with entries in $\mathbb{C}$)
\end{tabbing}
\vspace{-1em}
\noindent \textbf{Output:} value of $g(B,M)$, where $g(B,M)$ is the Grassmann integral defined in Eq.~\eqref{gBM}

\medskip

\texttt{evaluateIntegral}($B,M,K,2N$):

\qquad 1. If $K$ is odd or $K > 2N$,

\qquad \qquad Return $0$

\qquad 2. Else if $K = 2N$,

\qquad\qquad Return $\det(B)$

\qquad 3. Else

\qquad \qquad a. If $M$ is invertible,

\qquad \qquad \qquad Return $\Pf(M) \Pf(-B M^{-1} B^{\mathrm{T}})$ \hfill \textit{// with $\Pf(-BM^{-1}B^{\mathrm{T}}) \equiv 1$ if $K = 0$}

\qquad \qquad b. Else
 \begin{enumerate}[\qquad \qquad \qquad i.]
\item Let $2r$ be the rank of $M$

\item Let $R \in \mathrm{SO}(2N)$ be any special orthogonal matrix and $M'$ any invertible $2r \times 2r$ such that \[ M = R^{\mathrm{T}} \begin{pmatrix} M' &0 \\
0 &0 \end{pmatrix} R \]
\hfill \textit{// this is always possible since $M$ is antisymmetric}

\item Let $BR = \begin{pmatrix} B' & B'' \end{pmatrix}$, where $B'$ has size $K \times 2r$ and $B''$ has size $K \times 2N -2r$ 

\smallskip

\item Return $(-1)^{N+ K/2}\Pf(-M') \times \texttt{evaluateIntegral}(B''^{\mathrm{T}}, B' M'^{-1} B'^{\mathrm{T}}, 2N - 2r, K)$ 
\end{enumerate}
}
\vspace{1em}
\end{mdframed}
\parbox{0.98\linewidth}{
\caption{Our algorithm for efficiently evaluating $g(B,M)$ (Eq.~\eqref{gBM}); see Appendix~\ref{app:gMB} for the proof of correctness. When $M$ is invertible, the runtime scales with $N$ as $\mathcal{O}(N^3)$. When $M$ is not invertible, the algorithm is recursive. Importantly, the fourth argument of the function $\texttt{evaluateIntegral}$ always decreases (by at least $2$) when it is called recursively, so the number of recursion steps is at most $N$ (note that whenever the third argument is larger than the fourth, the algorithm returns in $\mathcal{O}(1)$ time). Therefore, the total runtime is at most $\mathcal{O}(N^4)$ for general $M$. Note that further optimisations are possible in some cases, by using, for instance, Pfaffian identities for block matrices.}\label{alg:gMB}} 
\end{myalgorithmfloat}

\subsubsection{Worked example: Overlaps with pure Gaussian states via matchgate shadows} \label{sec:workedexample}

Having presented our general method, we now apply it to show how to use our matchgate shadows to estimate the overlaps between a pure state (accessed via a state preparation circuit) and arbitrary pure Gaussian states. As shown in Appendix~\ref{app:overlaps_Gaussian}, we can evaluate any such overlap by evaluating the expectation value of $\ket{\phi}\bra{\mathbf{0}}$ for a Gaussian state $\ket{\phi}$ of even parity ($P\ket{\phi} = \ket{\phi}$ where $P = (-i)^n \gamma_{[2n]}$). To illustrate the general applicability of the method, we will actually consider evaluating the expectation value of $\widetilde{\gamma}_S\ket{\phi}\bra{\mathbf{0}}$ for some Majorana operators $\{\widetilde{\gamma}_\mu\}_{\mu \in [2n]}$ and $S \subseteq [2n]$.\footnote{In fact, such quantities naturally arise when constructing the energy estimator for certain generalisations of QC-AFQMC.} This reduces to $\ket{\phi}\bra{0}$ in the special case where $S = \varnothing$, so $\gamma_S = I$.   

By Eq.~\eqref{Minverse}, the expectation value of $\widetilde{\gamma}_S \ket{\phi}\bra{\mathbf{0}}$ with respect to a matchgate shadow sample $\mathcal{M}^{-1}(U_Q^\dagger\ket{b}\bra{b}U_Q)$ is 
\begin{equation} \label{estimateblahblah} \tr\left(\widetilde{\gamma}_S \ket{\phi}\bra{\mathbf{0}} \mathcal{M}^{-1}(U_Q^\dagger\ket{b}\bra{b}U_Q)\right) = \sum_{\ell=0}^n {2n\choose 2\ell}{n\choose \ell}^{-1} \tr\left(\widetilde{\gamma}_S \ket{\phi}\bra{\mathbf{0}} \mathcal{P}_{2\ell}(U_Q^\dagger\ket{b}\bra{b}U_Q) \right),\end{equation}
so it suffices to efficiently evaluate the terms $\tr(\widetilde{\gamma}_S \ket{\phi}\bra{\mathbf{0}} \mathcal{P}_{2\ell}(U_Q^\dagger\ket{b}\bra{b}U_Q))$ for each $\ell \in \{0,\dots, n\}$.  We consider more generally $\tr(\widetilde{\gamma}_S \ket{\phi}\bra{\mathbf{0}} \mathcal{P}_{2\ell}(\varrho))$ for an arbitrary Gaussian density operator $\varrho = \prod_{j=1}^n \frac{1}{2}(I - i\lambda_j \gamma'_{2j-1}\gamma'_{2j})$, with $\lambda_j \in [-1,1]$ and $\gamma_{\mu}' \coloneqq U_{Q}^\dagger \gamma_\mu U_{Q}$ for any $Q \in \orth$.
We use the same generating function trick as in Theorems~\ref{prop:densityoperators} and~\ref{prop:overlaps}, putting the variable $z$ in front of the bivectors $\gamma_{2j-1}'\gamma_{2j}'$ so that $\tr(\widetilde{\gamma}_S \ket{\phi}\bra{\mathbf{0}} \mathcal{P}_{2\ell}(\varrho))$ is the coefficient of $z^\ell$ in the polynomial
\[ p(z) \coloneqq \tr\left(\widetilde{\gamma}_S \ket{\phi}\bra{\mathbf{0}} \prod_{j=1}^n \frac{1}{2}\left(I - iz \lambda_j \gamma_{2j-1}'\gamma_{2j}'\right) \right). \]
$p(z)$ has degree at most $n$, so we can find its coefficients via polynomial interpolation if we are able to evaluate $p(z)$ at arbitrary values of $z$. We can use the method described above to do this efficiently. 

First, we fix some notation. Let $\widetilde{Q} \in \orth$ be the orthogonal matrix such that $\widetilde{\gamma}_\mu = \sum_{\nu=1}^{2n} \widetilde{Q}_{\mu\nu}\gamma_\mu$, and let $S = \{\mu_1,\dots, \mu_{|S|}\}$ with $\mu_1 < \dots < \mu_{|S|}$. We assume that $\ket{\phi}$ is specified as $\ket{\phi} = e^{i\alpha} U_R\ket{\mathbf{0}}$, where $e^{i\alpha}$ is a global phase, $\det(U_R) = 1$, and $R \in \mathrm{SO}(n)$ (we know $\det(R) = 1$ because $\ket{\phi}$ and $\ket{\mathbf{0}}$ have the same parity). For convenience, we denote $\varrho(z) \coloneqq \prod_{j=1}^{n} \frac{1}{2}(I - iz\lambda_j \gamma_{2j-1}'\gamma_{2j}')$.

Hence, we want to evaluate $p(z) = e^{i\alpha} \tr(\widetilde{\gamma}_S U_R \ket{\mathbf{0}}\bra{\mathbf{0}} \varrho(z))$. We start by writing down the Grassmann representations of $\widetilde{\gamma}_S$, $U_R$, $\ket{\mathbf{0}}\bra{\mathbf{0}}$, and $\varrho(z)$, using subsection~\ref{sec:efficientexamples}. Next, we apply Theorem~\ref{prop:tr(ABCD)} from subsection~\ref{sec:CtoG} to to express $p(z)$ as a Grassmann integral. Finally, we rewrite this integral in the form of Eq.~\eqref{gBM}, by assembling the matrices $B$ and $M$. The integral can then be evaluated by putting $B$ and $M$ into Algorithm~\ref{alg:gMB}. 

As shown in subsection~\ref{sec:efficientexamples}, the Grassmann representation of $\widetilde{\gamma}_S = \widetilde{\gamma}_{\mu_1}\dots \widetilde{\gamma}_{\mu_{|S|}}$ in terms of $\theta_1,\dots, \theta_{2n}$ is simply \[ \omega(\widetilde{\gamma}_S; \theta) = (\widetilde{Q}\theta)_{\mu_1}\dots (\widetilde{Q}\theta)_{\mu_{|S|}}\]
(recalling that $\theta$ represents the vector $\begin{pmatrix} \theta_1 &\dots &\theta_{2n}\end{pmatrix}^{\mathrm{T}})$). We find the Hamiltonian that generates $U_R$ (Eq.~\eqref{quadraticHamiltonian}), and block-diagonalise its antisymmetric coefficient matrix $h$ as $h = Q'^{\mathrm{T}} \bigoplus_{j=1}^n \begin{pmatrix} 0&\sigma_j \\ -\sigma_j &0 \end{pmatrix} Q'$. Then, considering the generic case where $\cos(\sigma_j) \neq 0$ for all $j\in [n]$ for simplicity, we have \[ \omega(U_R;\theta) = \Bigg(\prod_{j=1}^n \cos(\sigma_j)\Bigg) \exp\left(\frac{1}{2}\theta^{\mathrm{T}} T_R\theta\right), \qquad \text{where } T_R\coloneqq Q'^{\mathrm{T}} \bigoplus_{j=1}^n \begin{pmatrix} 0 &\tan(\sigma_j) \\ -\tan(\sigma_j) &0\end{pmatrix}Q' \]
by Eq.~\eqref{omegaUR}. Using Eq.~\eqref{omegavarrho} (noting that it holds for arbitrary $\lambda_j \in \mathbb{C}$), we obtain
\[ \omega(\ket{\mathbf{0}}\bra{\mathbf{0}}; \theta) = \frac{1}{2^n}\exp\left(-\frac{i}{2} \theta^{\mathrm{T}} C_{\ket{\mathbf{0}}} \theta \right), \qquad \omega(\varrho(z)) = \frac{1}{2^n}\exp\left(-\frac{i}{2} \theta^{\mathrm{T}} (zC_{\varrho}) \theta \right) \]
for any $z$,
where $C_{\ket{\mathbf{0}}}$ is the covariance matrix of the vacuum state $\ket{\mathbf{0}}\bra{\mathbf{0}}$ (see Eq.~\eqref{covarianceb}) and $C_{\varrho}$ is the covariance matrix of $\varrho$ (see Eq.~\eqref{covariancegeneral}). 

Then, applying Theorem~\ref{prop:tr(ABCD)} with $m = 4$, $A^{(1)} = \widetilde{\gamma}_S$, $A^{(2)} = U_R$, $A^{(3)} = \ket{\mathbf{0}}\bra{\mathbf{0}}$, and $A^{(4)} = \varrho(z)$ yields
\begin{align} 
p(z) &= e^{i\alpha}\tr\left(\widetilde{\gamma}_S U_R \ket{\mathbf{0}}\bra{\mathbf{0}}\varrho(z)\right) \nonumber \\
&= e^{i\alpha} 2^n \int D\theta^{(1)}\,D\theta^{(2)}\,D\theta^{(3)}\,D\theta^{(4)}\, (\widetilde{Q}\theta^{(1)})_{\mu_1}\dots (\widetilde{Q}\theta^{(1)})_{\mu_{|S|}}\nonumber\\
&\qquad \qquad\times \Bigg(\prod_{j=1}^n \cos(\sigma_j)\Bigg) \exp\left(\frac{1}{2}\theta^{(2)T} T_R\theta^{(2)}\right) \frac{1}{2^n} \exp\left(-\frac{i}{2}\theta^{(3)T} C_{\ket{\mathbf{0}}}\theta^{(3)}\right)\frac{1}{2^n} \exp\left(-\frac{i}{2}\theta^{(4)T} (zC_{\varrho})\theta^{(4)} \right) \nonumber\\
&\qquad\qquad \times \exp\left(\theta^{(1)T}\theta^{(2)} -\theta^{(1)T}\theta^{(3)} +\theta^{(1)T}\theta^{(4)} + \theta^{(2) T}\theta^{(3)} - \theta^{(2)T}\theta^{(4)} + \theta^{(3)T}\theta^{(4)}\right). \label{workedexample}
\end{align}
Since the arguments of the exponentials are all even elements of $\mathcal{G}_{8n}$, they mutually commute, so we can rewrite the product of exponentials as a single exponential. 

Now, we relabel $\theta_1^{(1)},\dots, \theta^{(1)}_{2n}, \dots, \theta_1^{(4)}, \dots, \theta_{2n}^{(4)}$ as $\chi_1, \dots, \chi_{8n}$, so 
\[ \chi \equiv \begin{pmatrix} \theta^{(1)} \\ \theta^{(2)} \\ \theta^{(3)} \\ \theta^{(4)}\end{pmatrix}. \]
Let $\widetilde{Q}\big|_{S*}$ denote the $|S| \times 2n$ matrix formed by restricting $\widetilde{Q}$ to rows in $S$, and define the $|S| \times 8n$ matrix $B$ by
\[ B = \begin{pmatrix} \widetilde{Q}\big|_{S*} &0 &0 &0
\end{pmatrix}\]
(where each $0$ denotes an $|S| \times 2n$ block). $B$ is constructed such that the $k$th entry of $B\chi$ is $(\widetilde{Q}\theta^{(1)})_{\mu_k}$. Also define the $8n \times 8n$ matrix $M(z)$ by
\begin{equation}
\label{M(z)} 
M(z) = \begin{pmatrix} 
0 &\mathbbm{1} &-\mathbbm{1} &\mathbbm{1} \\
-\mathbbm{1} &T_R &\mathbbm{1} &-\mathbbm{1} \\
\mathbbm{1} & -\mathbbm{1} &-iC_{\ket{\mathbf{0}}} &\mathbbm{1} \\
-\mathbbm{1} &\mathbbm{1} &-\mathbbm{1} &-izC_{\varrho}
\end{pmatrix} \end{equation}
(where $\mathbbm{1}$ denotes the $2n\times 2n$ identity matrix). $M(z)$ is constructed such that the product of exponentials in Eq.~\eqref{workedexample} is equal to $\exp(\frac{1}{2}\chi^{\mathrm{T}} M\chi)$. 
Therefore,
\[ p(z) = e^{i\alpha} \frac{1}{2^n}\Bigg(\prod_{j=1}^n \cos(\sigma_j) \Bigg)\int D\chi\,  (B\chi)_{1} \dots (B\chi)_{|S|} \exp\left(\frac{1}{2}\chi^{\mathrm{T}} M(z)\chi \right),  \]
and the Grassmann integral is equal to $g(B, M(z))$ in the notation of Eq.~\eqref{gBM}. This can be evaluated for any fixed value of $z$ in at most $\mathcal{O}(n^4)$ time using Algorithm~\ref{alg:gMB}.

To evaluate the expectation value of $\widetilde{\gamma}_S\ket{\phi}\bra{\mathbf{0}}$ with respect to any classical shadow sample $\mathcal{M}^{-1}(U_Q^\dagger\ket{b}\bra{b}U_Q)$ (taking $S =\varnothing$ if we wish to estimate the overlap with $\ket{\phi}$), we substitute $C_{U_Q^\dagger\ket{b}\bra{b}U_Q} = Q^{\mathrm{T}} C_{\ket{b}}Q$ for $C_\varrho$ in Eq.~\eqref{M(z)}, and find the coefficients of $p(z)$ via polynomial interpolation, evaluating $p(z)$ at $n + 1$ values of $z$ using Algorithm~\ref{alg:gMB}. We then substitute the coefficients into Eq.~\eqref{estimateblahblah}.

\section{Variance bounds} \label{sec:variance}

In subsections~\ref{sec:compute_densityoperators} and~\ref{sec:compute_overlaps}, we presented efficient methods for extracting unbiased estimates of expectation values of fermionic Gaussian density operators and of overlaps with Slater determinants from our matchgate shadows. We now analyse the variances of these estimates in this section. We leave investigating the variance of the more general classes of observables considered in subsection~\ref{sec:compute_general} to future work.

\subsection{Variance for expectation values of Gaussian density operators} \label{sec:variance_densityoperators} 

We begin in this subsection by proving Eq.~\eqref{variancedensityoperators}, which is an efficiently computable upper bound on the variance for estimating the expectation value of any Gaussian density operator $\varrho$. We then show that this bound scales at most polynomially with the number of fermionic modes $n$ (using a straightforward analysis to bound it loosely by $\mathcal{O}(n^3)$). We provide a more careful analysis in Appendix~\ref{app:variancebound} to demonstrate that the variance is $\mathcal{O}(\sqrt{n}\log n)$.

\begin{proof}[Proof of Eq.~\eqref{variancedensityoperators}] Let $\varrho$ be the density operator of any Gaussian state. 
By Eq.~\eqref{varrhodef}, $\varrho = \prod_{j =1}^n \frac{1}{2}(I - i\lambda_j \widetilde\gamma_{2j-1}\widetilde\gamma_{2j})$ for some $\lambda_j \in [-1,1]$ and Majorana operators $\widetilde{\gamma}_\mu$. We apply Eq.~\eqref{variancebound2} with this set of Majorana operators (recall we are free to choose the Majorana basis in Eq.~\eqref{variancebound2}) to $\varrho$, yielding
\[ \mathrm{Var}[\hat{o}]\Big|_{O = \varrho} \leq  \frac{1}{2^{2n}} \sum_{\substack{\ell_1,\ell_2,\ell_3 \geq 0 \\ \ell_1 + \ell_2 + \ell_3 \leq n}}\alpha_{\ell_1,\ell_2,\ell_3} \sum_{\substack{A_1, A_2, A_3 \subseteq [2n] \text{ disjoint} \\ |A_1| = 2\ell_1, |A_2|=2\ell_2,|A_3|=2\ell_3}} \big|\tr(\varrho\widetilde\gamma_{A_2\cup A_3})\tr(\varrho \widetilde\gamma_{A_3 \cup A_1})\big|, \]
with $\alpha_{\ell_1,\ell_2,\ell_3}$ given in Eq.~\eqref{alphadef}. We can write $\varrho$ as 
\[ \varrho = \frac{1}{2^n}\sum_{T \subseteq [n]} \lambda_T \widetilde{\gamma}_{\text{pairs}(T)} \]
where $\lambda_T \coloneqq \prod_{j \in T} \lambda_j$ and $\text{pairs}(T) \coloneqq \bigcup_{j \in T} \{2j-1,2j\}$. Thus, we see that $\tr(\varrho\widetilde{\gamma}_{A_2 \cup A_3})\tr(\varrho\widetilde{\gamma}_{A_3 \cup A_1})$ is nonzero only if $A_2 \cup A_3 = \text{pairs}(T_1')$ and $A_3 \cup A_1 = \text{pairs}(T_2')$ for some $T_1',T_2' \subseteq [n]$. For mutually disjoint $A_1, A_2, A_3$, this condition is equivalent to $A_1 = \text{pairs}(T_1)$, $A_2 = \text{pairs}(T_2)$, and $A_3 = \text{pairs}(T_3)$ for some mutually disjoint subsets $T_1,T_2, T_3 \subseteq [n]$, in which case $|\tr(\varrho\widetilde{\gamma}_{A_2 \cup A_3})\tr(\varrho\widetilde{\gamma}_{A_3 \cup A_1})| =  |2^{-n}\lambda_{A_2 \cup A_3}\tr(\widetilde{\gamma}_{A_2\cup A_3}^2) \cdot 2^{-n} \lambda_{A_3 \cup A_1}\tr(\widetilde\gamma_{A_3 \cup A_1}^2)| = |\lambda_{A_2\cup A_3}\lambda_{A_3 \cup A_1}|$.
Hence, using Eq.~\eqref{alphadef} and $|\lambda_T| \leq 1$ for any $T \subseteq [n]$, 
\begin{align*}
    \mathrm{Var}[\hat{o}]\Big|_{O = \varrho} &\leq \frac{1}{2^{2n}}\sum_{\substack{\ell_1,\ell_2,\ell_3 \geq 0 \\ \ell_1 + \ell_2 + \ell_3 \leq n}} \alpha_{\ell_1,\ell_2,\ell_3} \sum_{\substack{T_1,T_2,T_3 \subseteq [n] \text{ disjoint} \\ |T_1| = \ell_1, |T_2| = \ell_2, |T_3| = \ell_3}} |\lambda_{A_2 \cup A_3}\lambda_{A_3 \cup A_1}| \\
    &\leq \frac{1}{2^{2n}}\sum_{\substack{\ell_1,\ell_2,\ell_3\geq 0\\ \ell_1 + \ell_2 + \ell_3 \leq n}} \frac{{n\choose \ell_1,\ell_2,\ell_3, n-\ell_1-\ell_2-\ell_3}}{{2n\choose 2\ell_1,2\ell_2,2\ell_3, 2(n-\ell_1-\ell_2-\ell_3)}}\frac{{2n\choose 2(\ell_1 + \ell_3)}}{{n\choose \ell_1 + \ell_3}} \frac{{2n\choose 2(\ell_2 + \ell_3)}}{{n\choose \ell_2 + \ell_3}} \sum_{\substack{T_1,T_2,T_3 \subseteq [n] \text{ disjoint} \\ |T_1| = \ell_1, |T_2| = \ell_2, |T_3| = \ell_3}}1 \\
    &= \frac{1}{2^{2n}}\sum_{\substack{\ell_1,\ell_2,\ell_3\geq 0\\ \ell_1 + \ell_2 + \ell_3 \leq n}} \frac{{n\choose \ell_1,\ell_2,\ell_3, n-\ell_1-\ell_2-\ell_3}^2}{{2n\choose 2\ell_1,2\ell_2,2\ell_3, 2(n-\ell_1-\ell_2-\ell_3)}}\frac{{2n\choose 2(\ell_1 + \ell_3)}}{{n\choose \ell_1 + \ell_3}} \frac{{2n\choose 2(\ell_2 + \ell_3)}}{{n\choose \ell_2 + \ell_3}}.
\end{align*}
\end{proof}
We further analyse the bound in Eq.~\eqref{variancedensityoperators} using Stirling's approximation---more precisely, we consider the bounds
\begin{equation} \label{Stirling} \sqrt{2\pi}\sqrt{k} \left(\frac{k}{e}\right)^k < k! \leq e\sqrt{k}\left(\frac{k}{e}\right)^k, \end{equation}
which hold for all $k \in \mathbb{Z}_{>0}$~\cite{Robbins_1955}. By writing the multinomial coefficients in terms of factorials then applying Eq.~\eqref{Stirling}, it is straightforward to show that
\begin{equation} \label{binbound}
    \frac{{2n\choose 2k}}{{n\choose k}} \leq 2^{nH_{\mathrm{b}}(k/n)},
\end{equation}
for any $k \in \{0,\dots, n\}$, where $H_{\mathrm{b}}$ denotes the binary entropy function $H_{\mathrm{b}}(x) = -x\log_2 x -(1-x) \log_2(1-x)$, and that for integers $k_1 + k_2 + k_3 + k_4 = n$, 
\begin{equation} \label{multbound} \frac{{n\choose k_1,k_2,k_3,k_4}^2}{{2n\choose 2k_1,2k_2,2k_3,2k_4}} \leq \begin{dcases}
\sqrt{\frac{n}{k_1k_2k_3k_4}} \quad &k_1,k_2,k_3,k_4 > 0\\
\sqrt{\frac{n}{k_1k_2k_3}} \quad &k_1,k_2,k_3 > 0, k_4 = 0 \\
\sqrt{\frac{n}{k_1k_2}} \quad &k_1,k_2 >0, k_3= k_4 = 0 \\
1 \quad &k_1 > 0, k_2=k_3=k_4 = 0
\end{dcases}\end{equation}
(where all other cases follow by symmetry between $k_1, k_2, k_3, k_4$). We can obtain a loose bound by upper-bounding ${2n\choose 2k}{n\choose k}^{-1}$ by $2^n$ for all $k$, using $\max_x H_\mathrm{b}(x) = 1$, and noting that the RHS of Eq.~\eqref{multbound} is upper bounded by a constant in every case, so ${n\choose k_1,k_2,k_3,k_4}^2 {2n\choose 2k_1,2k_2,2k_3,2k_4}^{-1} \leq c$ for some constant $c$. Thus,
\[ \Var[\hat{o}]\Big|_{O = \varrho} \leq \frac{1}{2^{2n}} \sum_{\substack{\ell_1,\ell_2,\ell_3 \geq 0\\ \ell_1 + \ell_2 + \ell_3 \leq n}} c (2^n)^2 = \mathcal{O}(n^3).  \]

In Appendix~\ref{app:variancebound}, we tighten this to $\mathcal{O}(\sqrt{n}\log n)$, using Eqs.~\eqref{binbound} and~\eqref{multbound}. In Figure.~\ref{fig:bnzetaplot}, the bound in Eq.~\eqref{variancedensityoperators} is plotted as in red, while the dashed black line is $y = \sqrt{n}\ln n$.

\subsection{Variance for overlaps with Slater determinants} \label{sec:variance_overlaps}

In this subsection, we prove Eqs.~\eqref{varianceoverlaps} and~\eqref{kappa}, which constitute an efficiently computable upper bound on the variance for estimating the expectation value of $\ket{\varphi}\bra{\mathbf{0}}$, for any Slater determinant $\ket{\varphi}$ with an even number $\zeta$ of fermions. As explained in Appendix~\ref{app:overlaps}, the ability to estimate these expectation values allow us to estimate the overlaps between a pure state and arbitrary Slater determinants (with any number of fermions). 

Note that for $\zeta = 0$, we have $\ket{\varphi}\bra{\mathbf{0}} = \ket{\mathbf{0}}\bra{\mathbf{0}}$, which is a Gaussian density operator, and indeed Eqs.~\eqref{varianceoverlaps} and~\eqref{kappa} reduce in this case to our variance bound for Gaussian density operators, Eq.~\eqref{variancedensityoperators}. The key observation in the proof of Eq.~\eqref{variancedensityoperators}, given in the previous subsection, is that for a Gaussian density operator of the form $\varrho = \prod_{j=1}^n \frac{1}{2}(I -i\lambda_j \widetilde{\gamma}_{2j-1}\widetilde{\gamma}_{2j})$, the terms $\tr(\varrho \widetilde{\gamma}_{A_2 \cup A_3})\tr(\varrho \widetilde{\gamma}_{A_3 \cup A_1})$ (in the inner sum of Eq.~\eqref{variancebound2}) are nonzero only for certain combinations of disjoint subsets $A_1,A_2, A_3 \subseteq [2n]$. It was then straightforward to see that the number of such combinations is at most ${n \choose \ell_1,\ell_2,\ell_3, n -\ell_1-\ell_2-\ell_3}$. We use a similar approach to prove Eqs.~\eqref{varianceoverlaps} and~\eqref{kappa}, expanding $\ket{\varphi}\bra{\mathbf{0}}$ in a particular basis of Majorana operators, then counting the number of combinations $(A_1,A_2,A_3)$ that contribute to the sum in Eq.~\eqref{variancebound2}. For $\zeta > 0$, this requires a more involved combinatorial argument.

\begin{proof}[Proof of Eqs.~\eqref{varianceoverlaps} and~\eqref{kappa}] Let $\ket{\varphi} = \widetilde{a}_1^\dagger \dots \widetilde{a}_\zeta^\dagger\ket{\mathbf{0}}$ be any $\zeta$-fermion Slater determinant with $\zeta$ even, and let $\{\widetilde{\gamma}_\mu\}_{\mu \in [2n]}$ be the Majorana operators corresponding to $\{\widetilde{a}_j\}_{j \in [n]}$ (via Eq.~\eqref{Majoranadef}). Then, $\widetilde{\gamma}_{2j-1}\ket{\mathbf{0}} = (\widetilde{a}_j + \widetilde{a}_j^\dagger)\ket{\mathbf{0}} = \widetilde{a}_j^\dagger\ket{\mathbf{0}}$, which together with the anticommutation relations implies that $\widetilde{a}_1^\dagger\dots \widetilde{a}_\zeta^\dagger\ket{\mathbf{0}} = \widetilde{\gamma}_1\widetilde{\gamma}_3 \dots \widetilde{\gamma}_{2\zeta - 1}\ket{\mathbf{0}}$. Hence, defining $S_\zeta \coloneqq \{1,3,\dots, 2\zeta - 1\}$,
\[ \ket{\varphi}\bra{\mathbf{0}} = \widetilde{\gamma}_1\widetilde{\gamma}_3 \dots \widetilde{\gamma}_{2\zeta - 1}\ket{\mathbf{0}}\bra{\mathbf{0}} = \widetilde{\gamma}_{S_\zeta} \ket{\mathbf{0}} \bra{\mathbf{0}}. \]
Note that since the basis transformation from $a_j$ to $\widetilde{a}_j$ (Eq.~\eqref{ajtildeSlater}) is number-conserving, we can also express $\ket{\mathbf{0}}\bra{\mathbf{0}}$ in terms of $\widetilde{\gamma}_{\mu}$ as $\ket{\mathbf{0}}\bra{\mathbf{0}} = \prod_{j=1}^n \frac{1}{2}(I - i\widetilde{\gamma}_{2j-1}\widetilde{\gamma}_{2j})$.
Applying Eq.~\eqref{variancebound2} with $\widetilde{\gamma}_\mu$ as the Majorana basis to $\ket{\varphi}\bra{\mathbf{0}}$, we have
\begin{align*}
    \Var[\hat{o}]\Big|_{O = \ket{\varphi}\bra{\mathbf{0}}} \leq \frac{1}{2^{2n}}\sum_{\substack{\ell_1,\ell_2,\ell_3 \geq 0 \\ \ell_1 + \ell_2 + \ell_3 \leq n}} \alpha_{\ell_1,\ell_2,\ell_3} \sum_{\substack{A_1,A_2, A_3 \subseteq [2n] \text{ disjoint} \\ |A_1| = 2\ell_1, |A_2| = 2\ell_2, |A_3| = 2\ell_3}} \big|\tr\left(\ket{\varphi}\bra{\mathbf{0}} \widetilde{\gamma}_{A_2 \cup A_3}\right) \tr\left(\ket{\varphi}\bra{\mathbf{0}}\widetilde{\gamma}_{A_3 \cup A_1}\right) \big|,
\end{align*}
with 
\begin{align*}
    \big|\tr\left(\ket{\varphi}\bra{\mathbf{0}} \widetilde{\gamma}_{A_2 \cup A_3}\right) \tr\left(\ket{\varphi}\bra{\mathbf{0}}\widetilde{\gamma}_{A_3 \cup A_1}\right) \big| &= \big|\tr\left(\ket{\mathbf{0}}\bra{\mathbf{0}} \widetilde{\gamma}_{A_2 \cup A_3} \widetilde{\gamma}_{S_\zeta} \right) \tr\left(\ket{\mathbf{0}}\bra{\mathbf{0}}\widetilde{\gamma}_{A_3 \cup A_1}\widetilde{\gamma}_{S_\zeta}\right) \big| \\
    &= \big|\tr\left(\ket{\mathbf{0}}\bra{\mathbf{0}} \widetilde{\gamma}_{(A_2 \cup A_3) \triangle S_\zeta}\right) \tr\left(\ket{\mathbf{0}}\bra{\mathbf{0}}\widetilde{\gamma}_{(A_3 \cup A_1) \triangle S_\zeta}\right) \big|,
\end{align*}
where $A \triangle B$ denotes the symmetric difference of the sets $A$ and $B$. Writing 
\[ \ket{\mathbf{0}}\bra{\mathbf{0}} = \frac{1}{2^n} \sum_{T \subseteq [n]} \widetilde{\gamma}_{\text{pairs}(T)}, \]
where $\text{pairs}(T) \coloneqq \bigcup_{j \in T} \{2j-1,2j\}$, we see that $|\tr\left(\ket{\varphi}\bra{\mathbf{0}} \widetilde{\gamma}_{A_2 \cup A_3}\right) \tr\left(\ket{\varphi}\bra{\mathbf{0}}\widetilde{\gamma}_{A_3 \cup A_1}\right)|$ is nonzero only if $(A_2 \cup A_3) \triangle S_\zeta = \text{pairs}(T_1)$ and $(A_3 \cup A_1)\triangle S_\zeta = \text{pairs}(T_2)$ for some $T_1, T_2 \subseteq [n]$, in which case it takes on the value $1$. Thus, 
\[ \Var[\hat{o}] \Big|_{O = \ket{\varphi}\bra{\mathbf{0}}} \leq \frac{1}{2^{2n}} \sum_{\substack{\ell_1,\ell_2,\ell_3 \geq 0 \\ \ell_1 + \ell_2 + \ell_3 \leq n}} \alpha_{\ell_1,\ell_2,\ell_3} \, \kappa(n, \zeta, \ell_1,\ell_2,\ell_3), \]
where $\kappa(n,\zeta,\ell_1,\ell_2,\ell_3)$ is the number of tuples $(A_1,A_2, A_3)$ such that $A_1, A_2, A_3$ are mutually disjoint subsets of $[2n]$ of cardinalities $2\ell_1$, $2\ell_2$, and $2\ell_3$, respectively, and $(A_2 \cup A_3) \triangle S_\zeta = \text{pairs}(T_1)$ and $(A_3 \cup A_1)\triangle S_\zeta = \text{pairs}(T_2)$ for some $T_1, T_2 \subseteq [n]$. This last condition means that for each $j \in [n]$, $2j-1 \in (A_2 \cup A_3) \triangle S_\zeta$ if and only if $2j \in (A_2\cup A_3) \triangle S_\zeta$, and likewise for $(A_3 \cup A_1) \triangle S_\zeta$.

For mutually disjoint $A_1, A_2, A_3$,  $(A_2 \cup A_3) \triangle S_\zeta = A_2 \triangle (A_3 \triangle S_\zeta)$ and $(A_3 \cup A_1) \triangle S_\zeta = A_1 \triangle (A_3 \triangle S_\zeta)$. For convenience, let $A_3' \coloneqq A_3 \triangle S_\zeta$, and for each $i \in \{1,2,3\}$ and $j \in [n]$, we write
\[
A_i^{(j)} \coloneqq \begin{dcases} 
00 \qquad &\text{if $2j - 1, 2j \not\in A_i$} \\
01 \qquad &\text{if $2j-1 \not\in A_i$ and $2j \in A_i$} \\
10 \qquad &\text{if $2j-1 \in A_i$ and $2j\not\in A_i$} \\
11 \qquad &\text{if $2j-1,2j \in A_i$},
\end{dcases}
\]
and define $A_3'^{(j)}$ analogously.
% we write $A_1^{(j)} = 00$ if $2j-1,2j \not\in A_1$, $A_1^{(j)} = 01$ if $2j-1 \notin A_1$ and $2j \in A_1$, $A_1^{(j)} = 10$ if $2j-1 \in A_1$ and $2j \not\in A_1$, and $A_1^{(j)} = 11$ if $2j -1, 2j \in A_1$, and similarly for $A_2, A_3, A_3'$. 
The above condition then translates to $A_2^{(j)} \oplus A_3'^{(j)} = 00$ or $11$, and $A_1^{(j)} \oplus A_3'^{(j)} = 00$ or $11$, for every $j \in [n]$, where $\oplus$ denotes the bitwise XOR. Equivalently,
\[ A_1^{(j)} = A_3'^{(j)} \oplus 00 \text{ or } A_3'^{(j)} \oplus 11, \qquad A_2^{(j)} = A_3'^{(j)} \oplus 00 \text{ or } A_3'^{(j)} \oplus 11. \]
We assume that $A_1,A_2,A_3$ satisfy this condition and are mutually disjoint, and consider the different possibilities. 
For $j \in \{1,\dots, \zeta\}$,
\begin{itemize}
    \item if $A_3^{(j)} = 00$, then $A_3'^{(j)} = 10$, so either $A_1^{(j)} = 01$ and $A_2^{(j)} = 10$, or $A_1^{(j)} = 10$ and $A_2^{(j)} = 01$ (note that e.g., $A_1^{(j)}$ and $A_2^{(j)}$ cannot both be $01$, since this would mean that $2j$ is in both $X_1$ and $A_2$, contradicting the assumption that they are disjoint);
    \item if $A_3^{(j)} = 01$, then $A_3'^{(j)} = 11$, so $A_1^{(j)} = 00$ and $A_2^{(j)} = 00$ (note that we cannot have $A_1^{(j)} = 11$ since then $A_1$ and $A_3$ would not be disjoint, and likewise for $A_2$);
    \item if $A_3^{(j)} = 10$, then $A_3'^{(j)} = 00$, so $A_1^{(j)} = 00$ and $A_2^{(j)} = 00$;
    \item if $A_3^{(j)} = 11$, then $A_3'^{(j)} = 01$, and there are no possibilities for $A_1$ and $A_2$ that satisfy the disjointness assumption.
\end{itemize}
For $j \in \{\zeta + 1,\dots, n\}$, we have $A_3'^{(j)} = A_3^{(j)}$, so
\begin{itemize}
    \item if $A_3^{(j)} = 00$, then $A_1^{(j)} = 00$ and $A_2^{(j)} = 00$, or $A_1^{(j)} = 00$ and $A_2^{(j)} = 00$, or $A_1^{(j)} = 11$ and $A_2^{(j)} = 00$;
    \item if $A_3^{(j)} = 01$ or $A_3^{(j)} = 10$, there are no possibilities for $A_1$ and $A_2$;
    \item if $A_3^{(j)} = 11$, then $A_1^{(j)} = 00$ and $A_2^{(j)} = 00$.
\end{itemize}

Now, let us count the number of such combinations for which $|A_i| = 2\ell_i$ for $i \in \{1,2,3\}$. We imagine building up the sets $A_1,A_2, A_3$, starting with empty sets and then adding to them by choosing one of the viable cases above for each $j \in [n]$.
Note that for $j \in \{\zeta + 1,\dots, n\}$, we either add $0$ or $2$ elements to $A_3$ (choosing $A_3^{(j)} = 00$ adds no elements, whereas choosing $A_3^{(j)} = 11$ adds two). Hence, since we ultimately want an even number of elements in $A_3$, the number of elements we add to $A_3$ from the $j \in \{1,\dots, \zeta\}$ cases must be even. Let this number be $2k$. Then, $2k$ can be any integer between $0$ and $2\ell_3$. For each $j \in \{1,\dots, \zeta\}$, there are two ways to add $0$ elements to $A_3$ ($A_3^{(j)} = 00$, then $A_1 = 01$ and $A_2 = 10$ or vice versa), and both of these ways add $1$ element each to $A_1$ and $A_2$; there are also two ways to add $1$ element to $A_3$ ($A_3^{(j)} = 01$ or $A_3^{(j)} = 10$), and both of these ways add $0$ elements to $A_1$ and $A_2$. Thus, for a fixed $k$, there are ${\zeta \choose 2k} 2^{2k} 2^{\zeta- 2k} = {\zeta \choose 2k} 2^\zeta$ ways to add $2k$ elements to $A_3$ from the $j \in \{1,\dots,\zeta\}$ cases, and all of these ways add $\zeta - 2k$ elements to each of $A_1$ and $A_2$. Thus, it remains to add $2\ell_3 - 2k$ elements to $A_3$, $2\ell_1 - \zeta + 2k$ elements to $A_1$, and $2 \ell_2 - \zeta +2k$ elements to $A_2$ from the $j \in \{\zeta + 1,\dots, n\}$ cases. For each $j \in \{\zeta + 1, \dots n\}$, we can add no elements to $A_1,A_2,A_3$, or $2$ elements to $A_1$ and no elements to $A_2, A_3$, or $2$ elements to $A_2$ and no elements to $A_1,A_3$, or $2$ elements to $A_3$ and no elements to $A_1, A_2$. Thus, there are ${n - \zeta \choose (2\ell_1 - \zeta + 2k)/2, \, (2\ell_2 - \zeta + 2k)/2, (2\ell_3 - 2k)/2}'$ ways to add the remaining elements to $A_1, A_2, A_3$ (see Eq.~\eqref{multinomial_shorthand}). Summing over the possible values of $k$ gives
\[ \kappa(n,\zeta,\ell_1,\ell_2,\ell_3) = \sum_{k = 0}^{\ell_3} {\zeta \choose 2k} 2^{\zeta}  {n-\zeta \choose \ell_1 - \zeta/2 + k,\, \ell_2 -\zeta/2+k, \, \ell_3 -k}'. \]
This is the same as Eq.~\eqref{kappa} since ${\zeta \choose 2k} = 0$ if $2k > \zeta$. 
\end{proof}

The bound given by Eqs.~\eqref{varianceoverlaps} and~\eqref{kappa} is computable in $\text{poly}(n)$ time, and we plot it for values of $n$ up to $1000$ and various values of $\zeta \leq n$ in Fig.~\ref{fig:bnzetaplot}. The plot strongly suggests that the bound for every $\zeta$ scales sublinearly for all $n$. We know that the bound for $\zeta = 0$ is $\mathcal{O}(\sqrt{n}\log n)$ from Appendix~\ref{app:variancebound}, and from the plot, the bounds for $\zeta > 0$ seem to all  be lower than that for $\zeta = 0$. 

\section{Conclusion} \label{sec:discussion}

In this paper, we investigated the classical shadows obtained via two different ensembles of random matchgate circuits, one continuous and one discrete. 
In Theorem~\ref{thm:moments}, we analysed the first three moments of the uniform distribution over all matchgate circuits ($\mathrm{M}_n$) and of the uniform distribution over Clifford matchgate circuits ($\mathrm{M}_n\cap \mathrm{Cl}_n$), and found that they match, establishing that Clifford matchgate circuits form a ``matchgate 3-design.'' We then used these results to derive expressions for the classical measurement channel corresponding to these distributions (Eq.~\eqref{matchgateshadowschannel}) and its pseudo-inverse (Eq.~\eqref{Minverse}), and to establish bounds on the variance of our matchgate shadows estimator for the expectation value of an arbitrary observable (Eq.~\eqref{variancebound2}). Importantly, the 3-design property allowed us to easily bound the variance of estimators arising from the discrete ensemble for general classes of fermionic observables. (This is reminiscent of how Ref.~\citenum{Huang2020} uses the result that the Clifford group forms a unitary $3$-design~\cite{Zhu2017-xf} as a key step in analysing their Clifford classical shadows).

We then developed techniques to efficiently extract various kinds of information about a quantum state of interest from its matchgate shadow, and placed bounds on the variance of the associated estimators in some cases. For local fermionic operators, we showed that our matchgate shadows straightforwardly lead to efficient measurement schemes, matching the performance of prior work (\citen{Zhao2021-fv}) and generalising it to handle local operators in arbitrary single-particle bases.
We then demonstrated that our matchgate shadows can efficiently estimate not only local observables, but quantities like \(\langle \varrho \rangle\), where \(\varrho\) is the density operator of an arbitrary fermionic Gaussian state. For $\varrho$ pure, this gives the fidelity between an arbitrary unknown quantum state and the fermionic Gaussian state. We showed that these estimates can be obtained from the matchgate shadows in cubic time, and bounded the variance in terms of the system size $n$ as $\mathcal{O}(\sqrt{n} \log n)$.
This provides an interesting contrast with classical shadows derived from the Clifford group, where one is forced to choose between being able to efficiently (in terms of sample complexity) treat local qubit observables (using random single-qubit Clifford circuits) and being able to efficiently treat global properties such as fidelities and the expectation values of low-rank operators (using random $n$-qubit Clifford circuits), but not both simultaneously~\cite{Huang2020}. Our results show that randomising over a certain strict subset of the $n$-qubit Clifford group (that is, $\mathrm{M}_n \cap \mathrm{Cl}_n$), we can efficiently handle both local fermionic observables and global properties.  

One of the original motivations of our work was the need to efficiently estimate quantities of the form \(\langle \psi | \varphi \rangle\), where \(\ket{\psi}\) is an arbitrary pure state (accessed via a state preparation circuit) and \(\ket{\varphi}\) is an arbitrary Slater determinant. These overlaps are required in, for instance, auxiliary-field quantum Monte Carlo (AFQMC) methods, and a protocol for estimating them using Clifford shadows was recently implemented as a core subroutine in the quantum-classical hybrid AFQMC algorithm of \citen{Huggins2022}. However, this protocol involved a classical post-processing step whose complexity scales exponentially with the system size $n$. Our matchgate shadows approach (explicitly described in Algorithm~\ref{alg:overlaps}) to this problem removes this exponential bottleneck, enabling us to process each each sample in $\mathcal{O}(n^4)$ time. We also bounded the variance of the estimates (and hence the number of samples needed) by an expression that we can efficiently evaluate. This bound is plotted in Fig.~\ref{fig:bnzetaplot} for values of $n$ up to $1000$, showing that the variance is reasonable at these values, with a growth rate that suggests sublinear scaling. 

In addition, we constructed a more general framework for efficiently evaluation the expectation values of arbitrary products of certain commonly encountered fermionic operators with respect to matchgate shadows. We applied this framework to generalise our overlap estimation procedure to arbitrary fermionic Gaussian states, and we expect it will be a useful tool in the development of further applications of our matchgate shadows in the future.

Before concluding, we highlight some open questions raised by our results. 

We were able to improve the naive quartic-time algorithm for classically evaluating the expectation value of $\varrho$ with respect to a matchgate shadow sample (see subsection~\ref{sec:summary_densityoperators} and Appendix~\ref{app:linearPfaffian}), by exploiting the structure of the Pfaffians of certain linear matrix functions. Can a similar improvement be achieved for our overlap estimation algorithms? Due to the large number of overlap computations in QC-AFQMC, even shaving off this one factor of $n$ would be valuable for this application. More generally, we did not delve into the optimisations that may be available for our methods for evaluating other quantities, and these may manifest themselves when the methods are applied to concrete use cases. 

As for the quantum sample complexity, we were able to provide closed-form expressions for bounds on the variance of our matchgate shadow estimators in some cases, and efficiently computable bounds in others. Is is possible to provide a simpler or more intuitive characterisation of the variance for arbitrary operators, than our Eqs.~\eqref{varianceboundtilde} and~\eqref{variancebound2}?  How would our approach perform when used to estimate quantities that are nonlinear in the density operator of the unknown state? 

More fundamentally, we proved that $\mathrm{M}_n \cap \mathrm{Cl}_n$ is a $3$-design for $\mathrm{M}_n$. Is there an analogous result for the subgroup $\mathrm{M}_n^*$ of $\mathrm{M}_n$ consisting only of fermionic parity-conserving Gaussian unitaries, i.e., $U_Q$ such that $Q \in \mathrm{SO}(2n)$? We know that the $k$-th moment of the (Haar-)uniform distribution over $\mathrm{M}_n^*$ differs from that over $\mathrm{M}_n$, for any $k \geq 1$. To see this, observe that for any $Q \in \mathrm{SO}(2n)$, $\mathcal{U}_Q(\gamma_{[2n]}) = \gamma_{[2n]}$ (recall $\gamma_{[2n]}$ is proportional to the parity operator), so the 1-fold twirl for $\mathrm{M}_n^*$ maps $\gamma_{[2n]}$ to itself, whereas the $1$-fold twirl for $\mathrm{M}_n$ maps $\gamma_{[2n]}$ to 0 by Theorem~\ref{thm:moments}(i). Thus, the $1$-fold twirls are different, which implies that the $k$-fold twirls are different for all $k$.  

Finally, for Clifford shadows, \citen{Chen2021-cq} drew on the connection between classical shadows and randomised benchmarking to design a noise-robust classical shadow protocol.
Can we make use of this connection and the prior work of \citen{Helsen2022} on the randomised benchmarking of matchgates to design a similarly robust version of our matchgate shadows? 

We leave these explorations to future work.

\section*{Acknowledgements}

We thank Dave Bacon, David Gosset, Jarrod McClean, Bryan O'Gorman, Nicholas Rubin, and Andrew Zhao for helpful discussions. We thank Benchen Huang and Sam McArdle for pointing out various typos in the first draft. We thank anonymous reviewers for insightful comments and suggestions, and for pointing out further typos (special thanks to Reviewer 4 for spotting no fewer than 57 typos). JL thanks David Reichman for support. This research was supported in part by the National Science Foundation under Grant No. NSF PHY-1748958, as some of the discussions occurred during a workshop at the Kavli Institute for Theoretical Physics.

\section*{Author contributions}

WH and RB proposed studying classical shadows associated with matchgate circuits for the purpose of estimating Slater determinant overlaps, and WH provided the reduction of overlap estimation to expectation value estimation in Appendix~\ref{app:overlaps}. JL constructed the cubic-time algorithm in Appendix~\ref{app:linearPfaffian} for computing Pfaffians of certain linear matrix functions. KW developed the remaining methods and proofs, and wrote most of the manuscript. RB managed the collaboration. All authors contributed to discussions and paper-writing.

\noindent\textbf{}
\appendix 

\section{Reducing overlap estimation to expectation value estimation} \label{app:overlaps}

In the main text, one of the applications we considered was estimating the overlap between a general pure state $\ket{\psi}$ and a Slater determinant $\ket{\varphi}$, given a quantum circuit for preparing $\ket{\psi}$ and a classical description of $\ket{\varphi}$. In this appendix, we describe methods for reducing the problem of evaluating the overlap to the problem of finding the expectation value of an appropriate observable that can be handled by our matchgate shadows protocol. We generalise this reduction to the case where we want to estimate the overlap between $\ket{\psi}$ and an arbitrary pure Gaussian state.

The basic procedure we will use for reducing overlap estimation to expectation value estimation is as follows. Let $\ket{\Psi}$, $\ket{\Phi}$, and $\ket{\!\perp}$ be pure states (on the same number of qubits), and let $\rho$ be the density operator of $\frac{1}{\sqrt{2}}(\ket{\!\perp} + \ket{\Psi})$: 
\begin{equation} \label{overlaprho} \rho \coloneqq \frac{1}{2}(\ket{\!\perp} + \ket{\Psi})(\bra{\perp\!} + \bra{\Psi}). \end{equation}
Now, if $\braket{\perp\!}{\Psi} = 0$ and $\braket{\perp\!}{\Phi} = 0$, then we have
\[ \tr(\ket{\Phi}\bra{\perp\!} \rho) = \frac{1}{2}\braket{\Psi}{\Phi}. \] Hence, the overlap between $\ket{\Psi}$ and $\ket{\Phi}$ can be estimated using the classical shadows framework, by taking $\rho$ to be the initial state and $\ket{\Phi}\bra{\perp \!}$ to be an observable. Recall from subsection~\ref{sec:shadows_channel} that our matchgate classical shadows protocol gives unbiased estimates of $\tr(O\rho)$ if $\rho \in \Gamma_{\text{even}}$ or $O \in \Gamma_{\text{even}}$. Thus, to estimate overlaps using our matchgate shadows, we have to choose $\ket{\Psi}$, $\ket{\Phi}$, and $\ket{\!\perp}$ such that 1.~$\braket{\Psi}{\Phi}$ gives the overlap we want, 2.~$\braket{\perp\!}{\Psi} =\braket{\perp\!}{\Phi} =0$, and 3.~either $\rho \in \Gamma_{\text{even}}$ or $\ket{\Phi}\bra{\perp\!} \in \Gamma_{\text{even}}$ (or both). In the main text (subsection~\ref{sec:summary_Slater}), we assumed for simplicity that $\ket{\psi}$ has no overlap with the vacuum state $\ket{\mathbf{0}}$, and that $\ket{\varphi}$ is a Slater determinant with an even number of fermions $\zeta > 0$, which implies that $\braket{\mathbf{0}}{\varphi} = 0$ and $\ket{\varphi}\bra{\mathbf{0}} \in \Gamma_{\text{even}}$. Thus, taking $\ket{\Psi} = \ket{\psi}$, $\ket{\Phi} = \ket{\varphi}$, and $\ket{\!\perp} = \ket{\mathbf{0}}$ satisfies all three conditions above, allowing us to estimate $\braket{\psi}{\varphi}$. In the following, we show how choose $\ket{\Psi}$, $\ket{\Phi}$, and $\ket{\!\perp}$ to estimate $\braket{\psi}{\varphi}$ in more general situations, such as when these assumptions are not met. 

\subsection{Overlaps with arbitrary Slater determinants} \label{app:overlaps_Slater}

Let $\ket{\psi}$ be an arbitrary $n$-qubit pure state (possibly with nonzero overlap with the vacuum state $\ket{\mathbf{0}}$) and $\ket{\varphi}$ an arbitrary $n$-mode, $\zeta$-fermion Slater determinant. 

First, consider the case where $\zeta$ is odd. Then, we introduce an ancilla qubit (extending our system to $n+1$ fermionic modes) and take 
\begin{equation} \label{oddcase} \ket{\Psi} = \ket{\psi}\ket{1}, \qquad \ket{\Phi} = \ket{\varphi}\ket{1}, \qquad \ket{\!\perp} = \ket{\mathbf{0}}\ket{\mathbf{0}}. \end{equation}
Then, $\braket{\Psi}{\Phi} = \braket{\psi}{\varphi}$ and $\braket{\Psi}{\!\perp} = \braket{\Phi}{\!\perp} = 0$ (regardless of whether $\ket{\psi}$ and $\ket{\varphi}$ are orthogonal to $\ket{\mathbf{0}}$), so conditions 1 and 2 above are satisfied. Also, $\ket{\Phi}$ is a Slater determinant with an even number ($\zeta + 1$) of fermions, and $\ket{\!\perp}$ is the vacuum state of $n + 1$ modes, so $\ket{\Phi}\bra{\perp \!}$ is an even operator, as can be seen from the fact that it commutes with the parity operator $P = (-i)^n \gamma_1 \dots \gamma_{2n}$. Thus, we can use our matchgate classical shadows protocol to estimate the overlap $\braket{\psi}{\varphi}$ by estimating the expectation value of $\ket{\Phi}\bra{\perp\!}$ with respect to the state $\rho$ defined as in Eq.~\eqref{overlaprho}. Moreover, the fact that $\ket{\Phi}$ is also a Slater determinant allows us to directly apply Theorem~\ref{prop:overlaps} to efficiently compute the expectation value estimates from our matchgate shadows, as well as Eq.~\eqref{varianceoverlaps} to bound the variance of these estimates.

Now consider the case where $\zeta$ is even. Then, we introduce two ancilla qubits (extending our system to $n + 2$ fermionic modes) and take
\begin{equation} \label{evencase} \ket{\Psi} = \ket{\psi}\ket{1}\ket{1}, \qquad \ket{\Phi} = \ket{\varphi}\ket{1}\ket{1}, \qquad \ket{\!\perp} = \ket{\mathbf{0}}\ket{0}\ket{0}. \end{equation}
Again, all three conditions above are satisfied, and since $\ket{\Phi}$ is a Slater determinant with an even number ($\zeta + 1$) of fermions, we can use our matchgate classical shadows protocol, directly applying Theorem~\ref{prop:overlaps} and Eq.~\eqref{varianceoverlaps}. 

Furthermore, note that in both cases, the initial state $\ket{\rho}$ can be easily prepared given access to a controlled quantum circuit that prepares $\ket{\psi}$. For instance, in the case where $\zeta$ is odd, we can prepare $\frac{1}{\sqrt{2}}(\ket{\!\perp} + \ket{\Psi}) = \frac{1}{\sqrt{2}}(\ket{\mathbf{0}}\ket{0} + \ket{\psi}\ket{1})$ by starting with the all-zeros state on $n+1$ qubits, applying a Hadamard gate to the last qubit, then implementing the unitary that prepares $\ket{\psi}$, controlled on the last qubit. 
Finally, we remark that further optimisations are possible to reduce the number of ancilla qubits used in both cases, but here we have presented the most straightforward schemes that give observables $\ket{\Phi}\bra{\perp\!}$ to which we can directly apply the results of the main text. 

\subsection{Overlaps with arbitrary pure Gaussian states} \label{app:overlaps_Gaussian}

We extend the above procedures to estimating the overlap between $\ket{\psi}$ and an arbitrary pure Gaussian state $\ket{\phi}$. Note that any pure Gaussian state has a fixed parity, i.e., $\tr(P \ket{\phi}\bra{\phi}) = \pm 1$. Hence, since an operator is even if and only if it commutes with the parity operator, $\ket{\phi}\ket{1}\bra{\mathbf{0}}\bra{0} \in \Gamma_{\text{even}}$ if $\ket{\phi}$ has odd parity, and $\ket{\phi}\ket{1}\ket{1}\bra{\mathbf{0}}\bra{0}\bra{0} \in \Gamma_{\text{even}}$ if $\ket{\phi}$ has even parity. Therefore, we choose the states $\ket{\Psi}$, $\ket{\Phi}$, and $\ket{\! \perp}$ as in Eq.~\eqref{oddcase} (with $\varphi \to \phi$) if $\ket{\phi}$ is an odd-parity state, and we choose them as in Eq.~\eqref{evencase} if $\ket{\phi}$ is an even-parity state. In both cases, the initial state $\frac{1}{\sqrt{2}} (\ket{\!\perp} + \ket{\Psi})$ can be prepared given a preparation for $\ket{\psi}$, as discussed above. Moreover, $\ket{\phi}\ket{1}$ and $\ket{\phi}\ket{1}\ket{1}$ are both Gaussian states, so we can use our method described in subsection~\ref{sec:workedexample} to efficiently compute the expectation value of $\ket{\phi}\ket{1}\bra{\mathbf{0}}\ket{\mathbf{0}}$ or $\ket{\phi}\ket{1}\ket{1}\bra{\mathbf{0}}\bra{0}\bra{0}$ from our matchgate shadows.

\section{Generating random matchgate circuits}
\label{app:random_circuit_sampling}

In the main text, we analysed the classical shadows obtained from two distributions over matchgate circuits: the uniform distribution over the group $\mathrm{M}_n$ of all $n$-qubit matchgate circuits, and the uniform distribution over the group $\mathrm{M}_n \cap \mathrm{Cl}_n$ of matchgate circuits that are also in the Clifford group. In this appendix, we review how to efficiently sample unitaries from these distributions and implement them as linear-depth quantum circuits.

\subsection{Sampling from the uniform distribution over $\mathrm{M}_n$}

As discussed in subsection~\ref{sec:summary_random_matchgates}, we take the ``uniform'' distribution over \(\mathrm{M}_n\) to be the distribution induced by the Haar measure over \(\orth\). Recall from  subsection~\ref{sec:background_matchgates} that each $Q \in \orth$ corresponds uniquely to a matchgate circuit $U_Q$ if we ignore the global phase of $U_Q$ (for the classical shadows protocol, the global phase is irrelevant). Hence, we can sample from \(\mathrm{M}_n\) by sampling \(Q\) from \(\orth\) according to the Haar measure and finding the corresponding \(U_Q\). In this section, we briefly explain how to efficiently sample from \(\orth\) and implement \(U_Q\) as a quantum circuit.

There exists standard techniques for efficiently sampling from $\orth$~\cite{Mezzadri2006-sb, Stewart1980-rr}. The basic idea is to apply Gram-Schmidt orthogonalisation to a $2n \times 2n$ matrix whose entries are independent, normally distributed, real-valued random variables. This procedure yields a Haar-random element $Q$ of $\orth.$ 
With $Q$ in hand, we can compile \(U_Q\) using the techniques of \citen{Jiang2018-lu}, which decomposes $Q$ as a product of Givens rotations and reflections via straightforward operations.
The resulting circuits require \(\mathcal{O}(n^2)\) two-qubit gates on a linearly-connected array of qubits, and can be executed in depth \(\mathcal{O}(n)\).

\citen{Helsen2022} put forward an alternative approach to sampling from \(\orth\) that makes use of the probabilistic Hurwitz lemma~\cite{Diaconis2000-mf}. Their approach is based on directly sampling a sequence of random Givens rotations, which correspond to one- and two-qubit Pauli rotations as shown in subsection~\ref{sec:background_matchgates}, and adding in a reflection with probability $1/2$. Their method also achieves the asymptotically optimal \(\mathcal{O}(n^2)\) two-qubit gate count.

\subsection{Sampling from the uniform distribution over $\mathrm{M}_n \cap \mathrm{Cl}_n$}

As noted in \citen{Zhao2021-fv} and in subsection~\ref{sec:summary_random_matchgates}, a matchgate circuit \(U_Q \in \mathrm{M}_n\) is also an element of the $n$-qubit Clifford group \(\mathrm{Cl}_n\) if and only if \(Q\) is a signed permutation matrix, i.e., \(Q \in \mathrm{B}(2n)\). We can therefore sample from the discrete distribution \(\mathrm{M}_n \cap \mathrm{Cl}_n\) by sampling from \(\mathrm{B}(2n)\). We can straightforwardly sample a random permutation matrix. For example, we can sample the position of the \(1\) in each column (from the remaining set of valid choices) one column at a time. The signs can then be sampled independently. The resulting matchgate circuit/Clifford unitary can be compiled into a quantum circuit using the same strategy as suggested above for arbitrary matchgates (see \citen{Jiang2018-lu}), or by using strategies designed for compiling Clifford circuits (e.g., by using the canonical form of \citen{Bravyi2021-rn}).

\section{Proof details for Theorem~\ref{thm:moments}(iii)}  \label{app:E3}

In this appendix, we prove Lemma~\ref{lem:E3} and fill in the details for the proof of Theorem~\ref{thm:moments}(iii). As explained in the main text, we use $\mathcal{E}^{(3)}$ to denote $\mathcal{E}^{(3)}_{\mathrm{M}_n}$ or $\mathcal{E}^{(3)}_{\mathrm{M}_n\cap \mathrm{Cl}_n}$

\begin{proof}[Proof of Lemma~\ref{lem:E3}] \hfill
\begin{enumerate}[(a)]
    \item Let $S_1,S_2,S_3 \subseteq [2n]$. First, observe that if there exists an index $\mu \in [2n]$ such that $\mu \in S_i$ for some $i \in \{1,2,3\}$ and $\mu \not\in S_j$ for both $j \in \{1,2,3\}\setminus \{i\}$, or if there exists an index $\mu \in S_1 \cap S_2 \cap S_3$, then $\mathcal{E}^{(3)}\kett{\gamma_{S_1}}\kett{\gamma_{S_2}}\kett{\gamma_{S_3}} = 0$. To see this, consider the reflection matrix $Q \in \mathrm{B}(2n) \subset \mathrm{O}(2n)$ such that $\mU_Q\kett{\gamma_\mu} = -\kett{\gamma_\mu}$ and $\mu_Q\kett{\gamma_\nu} = \kett{\gamma_\nu}$ for all $\nu \neq \mu$. Then in either case, $\mU_Q\kett{\gamma_{S_1}}\kett{\gamma_{S_2}}\kett{\gamma_{S_3}} = -\kett{\gamma_{S_1}}\kett{\gamma_{S_2}}\kett{\gamma_{S_3}}$. Hence, using Eq.~\eqref{Einvariance} gives \[ \mathcal{E}^{(3)}\kett{\gamma_{S_1}}\kett{\gamma_{S_2}}\kett{\gamma_{S_3}} = \mathcal{E}^{(3)}\mU_Q \kett{\gamma_{S_1}}\kett{\gamma_{S_2}}\kett{\gamma_{S_3}} = -\mathcal{E}^{(3)} \kett{\gamma_{S_1}}\kett{\gamma_{S_2}}\kett{\gamma_{S_3}},\] so $\mathcal{E}^{(3)}\kett{\gamma_{S_1}}\kett{\gamma_{S_2}}\kett{\gamma_{S_3}} =0$ for both cases. 
    Thus, $\mathcal{E}^{(3)}\kett{\gamma_{S_1}}\kett{\gamma_{S_2}}\kett{\gamma_{S_3}}$ can only be nonzero if each index $\mu \in [2n]$ appears in either zero or two of the subsets $S_1,S_2, S_3$, in which case $S_1,S_2, S_3$ must be of the form $S_1 = A_1 \cup A_2$, $S_2 = A_2 \cup A_3$, $S_3 = A_3 \cup A_1$, for mutually disjoint subsets $A_1, A_2, A_3 \subseteq [2n]$.
    \item Suppose $A_1,A_2,A_3 \subseteq [2n]$ are mutually disjoint, $A_1',A_2',A_3' \subseteq [2n]$ are mutually disjoint, and $|A_i| = |A_i'|$ for all $i \in \{1,2,3\}$. Let $Q' \in \mathrm{B}(2n) \subseteq \orth$ be any permutation matrix such that $\mU_{Q'}\kett{\gamma_{A_i}} = \kett{\gamma_{A_i'}}$ for all $i \in \{1,2,3\}$ (it is clear that such a $Q'$ always exists, given the disjointness condition and the fact that $A_i$ and $A_i'$ have the same cardinality). Then, using Eq.~\eqref{Einvariance}, \[ \mathcal{E}^{(3)}\kett{\gamma_{A_1}\gamma_{A_2}} \kett{\gamma_{A_2}\gamma_{A_3}}\kett{\gamma_{A_3}\gamma_{A_1}} = \mathcal{E}^{(3)}\mU_{Q'}\kett{\gamma_{A_1}\gamma_{A_2}} \kett{\gamma_{A_2}\gamma_{A_3}}\kett{\gamma_{A_3}\gamma_{A_1}} = \mathcal{E}^{(3)}\kett{\gamma_{A_1'}\gamma_{A_2'}} \kett{\gamma_{A_2'}\gamma_{A_3'}}\kett{\gamma_{A_3'}\gamma_{A_1'}}. \]
\end{enumerate}
\end{proof}

\begin{proof}[Proof of Theorem~\ref{thm:moments}(iii)]
Considering the basis for $\LHn$ consisting of products of Majorana operators, Lemma~\ref{lem:E3}(a) gives the form of all basis elements that are not annihilated by $\mathcal{E}^{(3)}$. We can then expand $\mathcal{E}^{(3)}$ in terms of its matrix elements with respect to this basis as
\begin{align*}
\mathcal{E}^{(3)} = \sum_{A_1,A_2,A_3 \subseteq [2n] \text{ disjoint}}\sum_{A_1',A_2',A_3' \subseteq [2n] \text{ disjoint}} &\braa{\gamma_{A_1}\gamma_{A_2}}\braa{\gamma_{A_2}\gamma_{A_3}}\braa{\gamma_{A_3}\gamma_{A_1}} \mathcal{E}^{(3)}\kett{\gamma_{A_1'}\gamma_{A_2'}}\kett{\gamma_{A_2'}\gamma_{A_3'}}\kett{\gamma_{A_3'}\gamma_{A_1'}} \\
&\times \kett{\gamma_{A_1}\gamma_{A_2}}\kett{\gamma_{A_2}\gamma_{A_3}}\kett{\gamma_{A_3}\gamma_{A_1}}\braa{\gamma_{A_1'}\gamma_{A_2'}}\braa{\gamma_{A_2'}\gamma_{A_3'}}\braa{\gamma_{A_3'}\gamma_{A_1'}}, 
\end{align*}
where each sum is over all triplets of subsets of $[2n]$ that are mutually disjoint. Note that $\kett{\gamma_{A_i}\gamma_{A_j}}$ is equal to $\kett{\gamma_{A_i \cup A_j}}$ up to sign. Next, since $\mathcal{E}^{(3)}(\Gamma_{k_1} \otimes \Gamma_{k_2}\otimes \Gamma_{k_3}) = \Gamma_{k_1} \otimes \Gamma_{k_2}\otimes \Gamma_{k_3}$ by Fact~\ref{fact:mUQ}(d), the matrix element $\braa{\gamma_{A_1}\gamma_{A_2}}\braa{\gamma_{A_2}\gamma_{A_3}}\braa{\gamma_{A_3}\gamma_{A_1}} \mathcal{E}^{(3)}\kett{\gamma_{A_1'}\gamma_{A_2'}}\kett{\gamma_{A_2'}\gamma_{A_3'}}\kett{\gamma_{A_3'}\gamma_{A_1'}}$ can only be nonzero if $|A_1 \cup A_2| = |A_1' \cup A_2'|$, $|A_2 \cup A_3| = |A_2' \cup A_3'|$, and $|A_3 \cup A_1| = |A_3' \cup A_1'|$. This condition is equivalent to $|A_i| = |A_i'|$ for all $i \in \{1,2,3\}$, because the $A_i$ are mutually disjoint and the $A_i'$ are mutually disjoint. Then, using Eq.~\eqref{E3elements} (which follows from Lemma~\ref{lem:E3}(b)) gives 
\begin{align*}
\mathcal{E}^{(3)} &= \sum_{\substack{k_1,k_2,k_3 \in \{0,\dots, 2n\}\\k_1+k_2+k_3 \leq 2n}}c'_{k_1,k_2,k_3}\sum_{\substack{A_1,A_2,A_3 \subseteq [2n] \text{ disjoint} \\ |A_1| = k_1, |A_2| = k_2, |A_3| = k_3}}\kett{\gamma_{A_1}\gamma_{A_2}}\kett{\gamma_{A_2}\gamma_{A_3}}\kett{\gamma_{A_3}\gamma_{A_1}} \\
&\hspace{11em} \times \sum_{\substack{A_1',A_2',A_3' \subseteq [2n] \text{ disjoint} \\ |A_1'| = k_1, |A_2'| = k_2, |A_3'| = k_3}} \braa{\gamma_{A_1'}\gamma_{A_2'}}\braa{\gamma_{A_2'}\gamma_{A_3'}}\braa{\gamma_{A_3'}\gamma_{A_1'}}\\
&= \sum_{\substack{k_1,k_2,k_3 \in \{0,\dots, 2n\}\\k_1+k_2+k_3 \leq 2n}}c_{k_1,k_2,k_3} \kett{\Upsilon_{k_1,k_2,k_3}^{(3)}}\braa{\Upsilon_{k_1,k_2,k_3}^{(3)}} 
\end{align*}
with $c_{k_1,k_2,k_3} = c'_{k_1,k_2,k_3}{2n\choose k_1,k_2,k_3, 2n-k_1-k_2-k_3}$ and $\kett{\Upsilon^{(3)}_{k_1,k_2,k_3}}$ defined as in Eq.~\eqref{Psi3}, repeated here for convenience:
\[\kett{\Upsilon^{(3)}_{k_1,k_2,k_3}} \coloneqq {{2n \choose k_1,k_2, k_3,2n-k_1-k_2-k_3}}^{-1/2} \sum_{\substack{A_1,A_2, A_3 \subseteq [2n] \text{ } \mathrm{ disjoint} \\ |A_1| = k_1, |A_2| = k_2, |A_3| = k_3} }   \kett{\gamma_{A_1}\gamma_{A_2}}\kett{\gamma_{A_2}\gamma_{A_3}} \kett{\gamma_{A_3}\gamma_{A_1}}. \]

Finally, we show that for all $k_1,k_2,k_3 \in [2n]$ such that $k_1+k_2+k_3 \leq 2n$, we have 
\begin{equation} \label{UQ3Upsilon}
\mathcal{U}_Q^{\otimes 3}\kett{\Upsilon_{k_1,k_2,k_3}^{(3)}} = \kett{\Upsilon_{k_1,k_2,k_3}^{(3)}}
\end{equation} for all $Q \in \orth$, which implies that $\mathcal{E}^{(3)}\kett{\Upsilon_{k_1,k_2,k_3}^{(3)}}  = \kett{\Upsilon_{k_1,k_2,k_3}^{(3)}}$ and hence $c_{k_1,k_2,k_3} = 1$. To do this, we show that Eq.~\eqref{UQ3Upsilon} holds for a set of generators for $\orth$. In particular, we consider Givens rotations $G_{\mu}$ in the plane spanned by the $\gamma_\mu, \gamma_{\mu+1}$ axes, 
\begin{equation} \label{UGmu} \mU_{G_{\mu}(\theta)}: \begin{dcases} \gamma_\mu &\mapsto \enspace \cos\theta \gamma_\mu + \sin\theta \gamma_{\mu+1}, \qquad \\
\gamma_{\mu + 1} &\mapsto\enspace -\sin\theta \gamma_\mu + \cos\theta\gamma_{\mu+1}, \qquad  \\
\gamma_\nu &\mapsto \enspace \gamma_\nu \qquad \text{ for } \nu \not\in \{\mu,\mu + 1\}.\end{dcases}\end{equation}
for all $\mu \in [2n-1]$ and $\theta \in \mathbb{R}$,
along with any reflection $R$, say
\[\mU_R : \begin{dcases} \gamma_1 \mapsto -\gamma_1, \\
\gamma_\nu \mapsto \gamma_\nu \qquad \text{ for } \nu \neq 1.\end{dcases}\]

It is easy to see that $\mU_R^{\otimes 3}\kett{\gamma_{A_1}\gamma_{A_2}}\kett{\gamma_{A_2}\gamma_{A_3}}\kett{\gamma_{A_3}\gamma_{A_1}} = \kett{\gamma_{A_1}\gamma_{A_2}}\kett{\gamma_{A_2}\gamma_{A_3}}\kett{\gamma_{A_3}\gamma_{A_1}}$, and hence Eq.~\eqref{UQ3Upsilon} holds for $Q = R$. The index $1$ is either in one of the (disjoint) subsets $A_1,A_2, A_3$, or in none of them. In the latter case, clearly $\mU_R\kett{\gamma_{A_i}\gamma_{A_j}} = \kett{\gamma_{A_i}\gamma_{A_j}}$ for all $(i,j) \in \{(1,2),(2,3),(3,1)\}$. In the former, $\mU_R$ negates two of the $\kett{\gamma_{A_i}\gamma_{A_j}}$ and leaves the third unchanged. 

For the Givens rotations, our strategy is to show that \begin{equation} \label{dUGUpsilon} \frac{\partial}{\partial\theta} \mU_{G_{\mu}(\theta)}^{\otimes 3}\kett{\Upsilon_{k_1,k_2,k_3}^{(3)}}\Bigg|_{\theta = 0} = 0. \end{equation} Then, it will be clear that the exact same argument applies to the rotated set of Majorana operators $\{\widetilde{\gamma}_\xi\}_{\xi \in [2n]}$ given by $\kett{\widetilde{\gamma}_{\xi}} = \mU_{G_{\mu}(\widetilde{\theta})}\kett{\gamma_{\xi}}$ for any $\widetilde{\theta} \in \mathbb{R}$, leading to 
\begin{equation} \label{thetaderivative} \frac{\partial}{\partial\theta}\mU_{G_{\mu}(\theta)}^{\otimes 3} \sum_{\substack{A_1,A_2, A_3 \subseteq [2n] \text{ } \mathrm{ disjoint} \\ |A_1| = k_1, |A_2| = k_2, |A_3| = k_3} }   \kett{\widetilde\gamma_{A_1}\widetilde\gamma_{A_2}}\kett{\widetilde\gamma_{A_2}\widetilde\gamma_{A_3}} \kett{\widetilde\gamma_{A_3}\widetilde\gamma_{A_1}}\Bigg|_{\theta = 0} = 0. \end{equation}
The left-hand side is proportional to $\frac{\partial}{\partial\theta} \mU_{G_{\mu}(\theta)}^{\otimes 3}\kett{\Upsilon_{k_1,k_2,k_3}^{(3)}}\big|_{\theta = \widetilde{\theta}}$. Thus, we have $\frac{\partial}{\partial\theta} \mU_{G_{\mu}(\theta)}^{\otimes 3}\kett{\Upsilon_{k_1,k_2,k_3}^{(3)}} = 0$ for all $\theta$, and so $\mU_{G_{\mu}(\theta)}^{\otimes 3}\kett{\Upsilon_{k_1,k_2,k_3}^{(3)}} = \mU_{G_{\mu}(0)}^{\otimes 3}\kett{\Upsilon_{k_1,k_2,k_3}^{(3)}} = \kett{\Upsilon_{k_1,k_2,k_3}^{(3)}}$ for all $\theta$.
(One could also directly show that $\mU_{G_{\mu}(\theta)}\kett{\Upsilon_{k_1,k_2,k_3}^{(3)}} = \kett{\Upsilon_{k_1,k_2,k_3}^{(3)}}$ by considering the same cases as the ones below and applying Eq.~\eqref{UGmu}, but taking the derivative at $0$ substantially reduces the number of terms in the expansions and thereby simplifies the presentation.) 

Taking the $\theta$-derivative of Eq.~\eqref{UGmu} gives
\begin{equation} \label{UGmuderivative} \frac{\partial}{\partial \theta} \mU_{G_{\mu}(\theta)}\Bigg|_{\theta = 0} : \begin{dcases}
\gamma_\mu &\mapsto \enspace \gamma_{\mu + 1} \\
\gamma_{\mu + 1} &\mapsto\enspace -\gamma_\mu \\
\gamma_\nu &\mapsto \enspace 0 \qquad \text{ for } \nu \not\in \{\mu,\mu+1 \}.
\end{dcases} \end{equation}
Using Fact~\ref{fact:mUQ}(a), we have that for $S = \{\nu_1,\dots, \nu_{|S|}\} \subseteq [2n]$ with $\nu_1 < \dots < \nu_{|S|}$, 
\begin{align*}
    \frac{\partial}{\partial\theta}\mU_{G_\mu(\theta)}(\gamma_S)\Bigg|_{\theta = 0} = \sum_{i =1}^{|S|} \gamma_{\nu_1}\dots \gamma_{\nu_{i-1}} \frac{\partial}{\partial \theta}\mU_{G_\mu(\theta)}(\gamma_{\nu_i})\Bigg|_{\theta = 0}\gamma_{\nu_{i+1}}\dots \gamma_{\nu_{|S|}}.
\end{align*}
Substituting in Eq.~\eqref{UGmuderivative}, we obtain
\begin{equation} \label{dUGgammaS}
    \frac{\partial}{\partial\theta}\mU_{G_\mu(\theta)}(\gamma_S)\Bigg|_{\theta = 0} = \begin{dcases}
    0 \qquad &\text{if $\mu, \mu + 1 \in S$ or $\mu, \mu + 1 \not\in S$} \\
    \gamma_{S[\mu \to \mu+1]} \qquad &\text{if $\mu \in S$, $\mu + 1 \not\in S$} \\
    -\gamma_{S[\mu + 1 \to \mu]} \qquad &\text{if $\mu + 1 \in S, \mu \notin S$}
    \end{dcases},
\end{equation}
where $S[\mu \to \mu \pm 1]$ denotes the set $S$ with $\mu$ replaced by $\mu \pm 1$ (the $\mu,\mu+1\in S$ case follows from $\gamma_\mu^2 = \gamma_{\mu+1}^2 = I$). 

Then, looking at the basis elements $\kett{\gamma_{A_1}\gamma_{A_2}}\kett{\gamma_{A_2}\gamma_{A_3}}\kett{\gamma_{A_3}\gamma_{A_1}}$ (with $A_1$, $A_2$, $A_3$ mutually disjoint) that appear in $\kett{\Upsilon_{k_1,k_2,k_3}^{(3)}}$, we see from Eq.~\eqref{dUGgammaS} that $\frac{\partial}{\partial\theta} \mU_{G_{\mu}(\theta)}^{\otimes 3}\kett{\gamma_{A_1}\gamma_{A_2}}\kett{\gamma_{A_2}\gamma_{A_3}}\kett{\gamma_{A_3}\gamma_{A_1}}\big|_{\theta = 0} = 0$ whenever $\mu, \mu+1$ are both in $A_i$ for some $i \in \{1,2,3\}$ [Case 1], or $\mu, \mu + 1$ are not in $A_1$, $A_2$, nor $A_3$ [Case 2]. There are only two more cases left to consider (all others follow from symmetry): $\mu \in A_1$ with $\mu + 1$ not in $A_1$, $A_2$, nor $A_3$, and $\mu \in A_1$ with $\mu +1 \in A_2$. 

[Case 3] For any disjoint $(A_1,A_2,A_3)$ such that $\mu \in A_1$ and $\mu + 1 \notin A_1,A_2, A_3$, consider the corresponding triplet of disjoint subsets $(A_1[\mu \to \mu + 1], A_2, A_3)$. Using $\frac{\partial}{\partial\theta}\mU_{G_\mu(\theta)}^{\otimes 3}\big|_{\theta = 0} = \frac{\partial}{\partial\theta}\mU_{G_\mu(\theta)}\big|_{\theta = 0}\otimes \mathcal{I}\otimes \mathcal{I} + \mathcal{I} \otimes \frac{\partial}{\partial\theta}\mU_{G_\mu(\theta)}\big|_{\theta = 0}\otimes \mathcal{I} + \mathcal{I} \otimes \mathcal{I} \otimes \frac{\partial}{\partial\theta}\mU_{G_\mu(\theta)}\big|_{\theta = 0}$ and applying Eq.~\eqref{dUGgammaS}, we see that
\begin{align*}
    \frac{\partial}{\partial \theta}\mU_{G_\mu(\theta)}^{\otimes 3}\Bigg|_{\theta =0}\big( &\kett{\gamma_{A_1}\gamma_{A_2}}\kett{\gamma_{A_2}\gamma_{A_3}}\kett{\gamma_{A_3}\gamma_{A_1}} +\kett{\gamma_{A_1[\mu \to \mu+1]}\gamma_{A_2}}\kett{\gamma_{A_2}\gamma_{A_3}}\kett{\gamma_{A_3}\gamma_{A_1[\mu \to \mu + 1]}} \big) \\
    = \enspace &\kett{\gamma_{A_1[\mu\to\mu+1]}\gamma_{A_2}}\kett{\gamma_{A_2}\gamma_{A_3}}\kett{\gamma_{A_3}\gamma_{A_1}} + \kett{\gamma_{A_1}\gamma_{A_2}}\kett{\gamma_{A_2}\gamma_{A_3}}\kett{\gamma_{A_3}\gamma_{A_1[\mu \to \mu+1]}} \\
    - \enspace &\kett{\gamma_{A_1[\mu \to \mu + 1][\mu + 1\to \mu]}\gamma_{A_2}}\kett{\gamma_{A_2}\gamma_{A_3}}\kett{\gamma_{A_3}\gamma_{A_1[\mu \to \mu + 1]}} - \kett{\gamma_{A_1[\mu \to \mu + 1]}\gamma_{A_2}}\kett{\gamma_{A_2}\gamma_{A_3}}\kett{\gamma_{A_3}\gamma_{A_1[\mu \to \mu + 1][\mu + 1 \to \mu]}} \\
    = \enspace &0,
\end{align*}
observing that $A_1[\mu\to\mu + 1][\mu+1 \to \mu] = A_1$.

[Case 4] For any disjoint $(A_1,A_2, A_3)$ such that $\mu \in A_1$ and $\mu + 1 \in A_2$, consider the corresponding triplet of disjoint subsets $(A_1[\mu\to \mu+1], A_2[\mu+1 \to \mu], A_3)$. The superposition of the two corresponding basis elements is 
\begin{align*}
    &\kett{\gamma_{A_1}\gamma_{A_2}}\kett{\gamma_{A_2}\gamma_{A_3}}\kett{\gamma_{A_3}\gamma_{A_1}} + \kett{\gamma_{A_1[\mu \to \mu + 1]}\gamma_{A_2[\mu+1\to\mu]}}\kett{\gamma_{A_2[\mu+1\to\mu]}\gamma_{A_3}}\kett{\gamma_{A_3}\gamma_{A_1[\mu\to\mu+1]}} \\
    &\quad = \kett{\gamma_{A_1}\gamma_{A_2}}\big(\kett{\gamma_{A_2}\gamma_{A_3}}\kett{\gamma_{A_3}\gamma_{A_1}} - \kett{\gamma_{A_2[\mu+1\to\mu]}\gamma_{A_3}}\kett{\gamma_{A_3}\gamma_{A_1[\mu\to\mu+1]}}\big)
\end{align*}
since $\gamma_{A_1[\mu \to \mu+1]}\gamma_{A_2[\mu + 1 \to \mu]}$ is the same as $\gamma_{A_1}\gamma_{A_2}$ except with $\gamma_{\mu}$ and $\gamma_{\mu+1}$ exchanged, so $\gamma_{A_1[\mu \to \mu+1]}\gamma_{A_2[\mu + 1 \to \mu]} = -\gamma_{A_1}\gamma_{A_2}$. By Eq.~\eqref{dUGgammaS}, 
\[ \frac{\partial}{\partial \theta} \mU_{G_\mu(\theta)}\Bigg|_{\theta = 0} \kett{\gamma_{A_1}\gamma_{A_2}} = 0, \]
since $\gamma_{A_1}\gamma_{A_2} \propto \gamma_{A_1 \cup A_2}$ and $\mu, \mu+1$ are both in $A_1 \cup A_2$. Hence,
\begin{align*}
    \frac{\partial}{\partial \theta}\mU_{G_\mu(\theta)}^{\otimes 3}\Bigg|_{\theta =0}\big(&\kett{\gamma_{A_1}\gamma_{A_2}}\kett{\gamma_{A_2}\gamma_{A_3}}\kett{\gamma_{A_3}\gamma_{A_1}} + \kett{\gamma_{A_1[\mu \to \mu + 1]}\gamma_{A_2[\mu+1\to\mu]}}\kett{\gamma_{A_2[\mu+1\to\mu]}\gamma_{A_3}}\kett{\gamma_{A_3}\gamma_{A_1[\mu\to\mu+1]}} \big) \\
    = \enspace &\kett{\gamma_{A_1}\gamma_{A_2}}\big[-\kett{\gamma_{A_2[\mu+1\to\mu]}\gamma_{A_3}} \kett{\gamma_{A_3}\gamma_{A_1}} + \kett{\gamma_{A_2}\gamma_{A_3}} \kett{\gamma_{A_3} \gamma_{A_1[\mu\to\mu+1]}} \\
    &\hspace{4.35em} - \big(\kett{\gamma_{A_2[\mu+1\to\mu][\mu \to \mu+1]}\gamma_{A_3}}\kett{\gamma_{A_3}\gamma_{A_1[\mu\to\mu+1]}} - \kett{\gamma_{A_2[\mu+1\to\mu]}\gamma_{A_3}}\kett{\gamma_{A_3}\gamma_{A_1[\mu\to\mu+1][\mu+1\to\mu]}} \big)\big] \\
    = \enspace &0.
\end{align*}

Eq.~\eqref{dUGUpsilon} now follows by expanding $\kett{\Upsilon_{k_1,k_2,k_3}^{(3)}}$ into basis elements [Eq.~\eqref{Psi3}] and applying the above cases. (In particular, note that the basis elements that are not handled by Case 1 or 2 are all paired off according to Case 3 or 4 such that there are no overlapping pairs. Explicitly, consider the map $(A_1,A_2,A_3) \mapsto (A_1',A_2', A_3')$ where
\[
\begin{dcases}
A_i' = A_i[\mu \to \mu + 1], \enspace A_j' = A_j \text{ for } j \neq i \qquad &\text{if $\mu \in A_i$, $\mu + 1 \not\in A_1,A_2,A_3$} \\
A_i' = A_i[\mu+ 1 \to \mu], \enspace A_j' = A_j \text{ for } j \neq i \qquad &\text{if $\mu+1 \in A_i$, $\mu \not\in A_1,A_2,A_3$} \\
A_i' = A_i[\mu \to \mu + 1], \enspace A_j' = A_j[\mu +1\to \mu], \enspace A_k' = A_k \text{ for } k \not\in\{i,j\} \qquad &\text{if $\mu \in A_i$, $\mu + 1\in A_j$ $(i \neq j)$}.
\end{dcases}
\]
This map is defined on every triplet $(A_1,A_2,A_3)$ of disjoint subsets that does not fall under Case~1 or 2. The map takes each such triplet to a different triplet and is clearly self-inverse, so it partitions its domain into pairs, and the calculations above for Cases 3 and 4 show that $\frac{\partial}{\partial \theta}\mU_{G_\mu(\theta)}^{\otimes 3}\Big|_{\theta =0} (\kett{\gamma_{A_1}\gamma_{A_2}}\kett{\gamma_{A_2}\gamma_{A_3}}\kett{\gamma_{A_3}\gamma_{A_1}} + \kett{\gamma_{A_1'}\gamma_{A_2'}}\kett{\gamma_{A_2'}\gamma_{A_3'}}\kett{\gamma_{A_3'}\gamma_{A_1'}}) =~0$.
\end{proof}

\section{Computing Pfaffians of linear matrix functions in cubic time} \label{app:linearPfaffian}

In this appendix, we consider the problem of computing the coefficients of the polynomial $\Pf(A(z))$, where $A(z) = B + zC$ for $z$-independent, antisymmetric $2r\times 2r$ matrices $B$ and $C$. Using polynomial interpolation, this can be done in $\mathcal{O}(r^4)$ time, since the polynomial has degree at most $r$. Here, we present a different procedure that computes all of the $r+1$ coefficients in $\mathcal{O}(r^3)$ time in total, for the case where $A(0) = B$ is invertible. Since the matrix $C_{\varrho_1}'^{-1}$ in Theorem~\ref{prop:densityoperators} is always invertible by construction, we can apply our procedure to compute the coefficients of the polynomial $p_{\varrho_1,\varrho_2}(z)$ in Theorem~\ref{prop:densityoperators} in $\mathcal{O}(r^3)$ time, for arbitrary Gaussian states $\varrho_1$ and $\varrho_2$. 

The strategy is to recursively compute the derivatives of $\Pf(A(z))$ evaluated at $z = 0$. Denoting the $\ell$th derivative of a function $F$ by $\partial^\ell F$, the coefficient of $z^\ell$ in $\Pf(A(z))$ is 
\[ \partial^\ell \Pf(A) \Big|_{z = 0} = \partial^\ell \Pf(A(0)). \]
From $\Pf(A)^2 = \det(A)$ and Jacobi's formula, the first derivative of $\Pf(A)$ is
\begin{equation} \label{PfJacobi} \partial^1 \Pf(A) = \frac{1}{2}\Pf(A)\tr(A^{-1} \,\partial^1 A) \end{equation}
if $A(z)$ is invertible. Since $A$ is a smooth function of $z$, and $A(0)$ is invertible by assumption, $A(z)$ is invertible in a neighbourhood around $z = 0$, so we can apply Eq.~\eqref{PfJacobi} to compute the derivatives at $z =0$. Define the functions 
\[ f(z) = \Pf(A(z)), \qquad g(z) = \frac{1}{2}\tr\left(A^{-1}(z)\,  \partial^1A(z)\right) \]
so that Eq.~\eqref{PfJacobi} becomes $\partial^1 f = fg$.
Then, from the product rule, the $\ell$th derivative of $f$ is
\begin{align} \label{partialellf}
    \partial^\ell f = \sum_{j = 0}^{\ell - 1} {\ell - 1 \choose j} \left(\partial^{\ell - 1 - j} f\right) \left(\partial^{j} g\right)
\end{align}
for $\ell \geq 1$. For $A(z) = B + zC$, we have $\partial^1 A(z) = C$ and $g(z) = \frac{1}{2}\tr(A^{-1}(z)C)$. Hence, since $\partial^1 A^{-1} = -A^{-1} (\partial^1 A)A^{-1}$,
\[ \partial^j g = \frac{1}{2}(-1)^j j! \, \tr\left((A^{-1}C)^{j+1}\right). \]

The procedure for computing $\partial^\ell \Pf(A(0)) = \partial^\ell f(0)$ for $\ell \in \{0,\dots, r\}$ is as follows. First, we compute and store the eigenvalues $\lambda_1,\dots, \lambda_{2r}$ of $A^{-1}(0)C = B^{-1}C$, which takes $\mathcal{O}(r^3)$ time. Then, we compute $\partial^j g(0)$ using the fact that $\tr((B^{-1}C)^{j+1}) = \sum_{i = 1}^{2r} \lambda_i^{j+1}$. By storing the powers of the eigenvalues as they are computed, we can obtain $\partial^j g(0)$ for every $j \in \{0,\dots, r-1\}$ in $\mathcal{O}(r^2)$ time. Finally, we iteratively compute $\partial^\ell f(0)$ via Eq.~\eqref{partialellf}, using these values of $\partial^j g(0)$ as well as previously computed values of lower-order derivatives of $f$ at $z=0$ (starting by computing $f(0) = \Pf(B)$, which uses $\mathcal{O}(r^3)$ time). Since Eq.~\eqref{partialellf} gives $\partial^\ell f$ as a linear combination of $\ell$ terms of the form $(\partial^{\ell-1-j} f)(\partial^j g)$, computing $\partial^\ell f(0)$ for every $\ell \in [r]$ takes $\mathcal{O}(r^3)$ time if we compute the binomial coefficients independently, or $\mathcal{O}(r^2)$ if we iteratively compute them (using e.g., Pascal's triangle). Thus, the total runtime is $\mathcal{O}(r^3)$.

\section{Proofs for the efficient estimation of more general observables} \label{app:Grassmann}

In this appendix, we prove Theorem~\ref{prop:tr(ABCD)} and the correctness of Algorithm~\ref{alg:gMB} in subsection~\ref{sec:compute_general}. As preliminaries, we review some standard identities for Grassmann integrals that we will use. 

For arbitrary $N \in \mathbb{Z}_{>0}$, let $\chi_1,\dots,\chi_{2N}$ be generators of a $2^{2N}$-dimensional Grassmann algebra $\mathcal{G}_{2N}$. It follows from the definition in Eqs.~\eqref{Grassmannintegraldef1} and~\eqref{Grassmannintegraldef2} that Grassmann integrals anticommute,
\begin{equation} \label{integralantisymmetric} \int d\chi_\mu \int d\chi_\nu = -\int d\chi_\nu \int d\chi_\mu.
\end{equation}

For any antisymmetric $2N\times 2N$ matrix $M$, 
\begin{equation} \label{GrassmannGaussian} \int D\chi\, \exp\left(\frac{1}{2}\chi^{\mathrm{T}} M\chi\right) = \Pf(M). \end{equation}
This is essentially the analogue of Fact~\ref{fact:bivectorPfaffian} for Grassmann variables. Expanding the exponential as a power series, $\exp(\frac{1}{2}\chi^{\mathrm{T}} M\chi) = \sum_{j=0}^{\infty} \frac{1}{j!} (\frac{1}{2}\sum_{\mu,\nu=1}^{2N} M_{\mu\nu}\chi_{\mu}\chi_\nu)^j$, we see that the only power that can contain a term proportional to $\chi_1\dots \chi_N$ (and hence contribute to the integral) is that for $j = N$. Now, note that Fact~\ref{fact:bivectorPfaffian} relies only on the antisymmetry of the wedge product under the exchange of $1$-vectors in the Clifford algebra. The multiplication operation of the Grassmann algebra is likewise antisymmetric under the exchange of generators (Eq.~\eqref{Grassmannproduct}), so we can infer that $\frac{1}{N!} (\frac{1}{2}\sum_{\mu,\nu=1}^{2N} M_{\mu\nu}\chi_{\mu}\chi_\nu)^N = \Pf(M) \chi_1\dots \chi_N$. A straightforward generalisation of Eq.~\eqref{GrassmannGaussian} is
\begin{equation} \label{A.15(b)}
\int D\chi \, \chi_{j_1}\dots \chi_{j_k} \exp\left(\frac{1}{2}\chi^{\mathrm{T}} M\chi\right)  = \epsilon(J)\, \Pf\big(M\big|_{\overline{J}}\big),
\end{equation}
where $J = \{j_1,\dots, j_k\} \subseteq [2N]$ with $j_1 < \dots < j_k$, $\epsilon(J) \coloneqq (-1)^{|J|(|J|-1)/2 + \sum_{j \in J}j}$, and $M\big|_{\overline{J}}$ is the submatrix of $M$ consisting of rows and columns in $\overline{J} \coloneqq [2N] \setminus J$. This is Theorem~A.15(b) of Ref.~\citenum{Caracciolo}, and follows from the antisymmetry of Grassmann integrals (Eq.~\eqref{integralantisymmetric}) combined with Eq.~\eqref{GrassmannGaussian}. 

Also, under a change of variables $\widetilde{\chi}_\mu = \sum_{\nu = 1}^{2N} T_{\mu\nu} \chi_\nu$ for $\mu \in [2N]$, where $T$ is any invertible $2N \times 2N$ matrix,
\begin{equation} \label{Dchitilde} \int D\chi = \det(T) \int D\widetilde{\chi}\end{equation}
(where $D\widetilde{\chi} \equiv d\widetilde{\chi}_{2N} \dots d\widetilde{\chi}_{1}$). This also follows from antisymmetry (Eq.~\eqref{integralantisymmetric}). 

Now let $\chi_1,\dots, \chi_{2N}, \eta_1,\dots, \eta_K$ (for $K \in \mathbb{Z}_{\geq 0}$) be generators of some extended Grassmann algebra, where the $\eta_\mu$ do not involve $\chi_1,\dots, \chi_{2N}$. Then, for any $K \times N$ matrix $B$ and any invertible antisymmetric $2N \times 2N$ matrix $M$, 
\begin{equation} \label{A.15(a)}
    \int D\chi\, \exp\left(\frac{1}{2}\chi^{\mathrm{T}} M\chi + \eta^{\mathrm{T}} B \chi \right) = \Pf(M) \exp\left(\frac{1}{2}\eta^{\mathrm{T}} BM^{-1} B^{\mathrm{T}} \eta \right).
\end{equation}
This follows from Eq.~\eqref{GrassmannGaussian} by ``completing the square,'' defining $\chi' = \chi - M^{-1} B^{\mathrm{T}} \eta$ so that $\exp(\frac{1}{2}\chi^{\mathrm{T}} M\chi + \eta^{\mathrm{T}} B \chi) = \exp(\frac{1}{2}\chi'^{\mathrm{T}} M \chi') \exp(\frac{1}{2}\eta^{\mathrm{T}} B M^{-1} B^{\mathrm{T}} \eta)$. See Theorem~A.15(a) of Ref.~\citenum{Caracciolo} for a more general statement and proof.

\subsection{Proof of Theorem~\ref{prop:tr(ABCD)}} \label{app:tr(ABCD)}

\begin{proof}[Proof of Theorem~\ref{prop:tr(ABCD)}] 
First, note that it suffices to prove Eq.~\eqref{CliffordtoGrassmann} for the case where each of $A^{(i)}$ is a product $\gamma_{S^{(i)}}$ of Majorana operators, for arbitrary $S^{(i)} \subseteq [2n]$. Then, Eq.~\eqref{CliffordtoGrassmann} extends to arbitrary $A^{(i)} \in \LHn$, by virtue of the linearity of the Grassmann representation $\omega$ (Eq.~\eqref{Grassman_rep}) and the fact that products of Majorana operators span $\LHn$. 

We start by considering the LHS of Eq.~\eqref{CliffordtoGrassmann} for $A^{(i)} = \gamma_{S^{(i)}}$:
\[ \text{LHS} \coloneqq \tr\left(\gamma_{S^{(1)}} \dots \gamma_{S^{(m)}}\right). \]
For each $\mu \in [2n]$, define $J_\mu \subseteq [m]$ to be the set of indices $i$ for which $\gamma_\mu$ appears in $\gamma_{S^{(i)}}$: 
\[ J_\mu \coloneqq \{i \in [m]: \mu \in S^{(i)}\}. \] We now reorder the $\gamma_\mu$'s in the product $\gamma_{S^{(1)}} \dots \gamma_{S^{(m)}}$ such that any $\gamma_1$'s all appear to the left, followed by any $\gamma_2$'s, and so on. Since distinct $\gamma_\mu$'s  anticommute, this reordering incurs a sign, which we will denote by $\varsigma \in \{\pm 1\}$. Hence, 
\begin{align*}
    \text{LHS} &= \varsigma \, \tr\left( \gamma_1^{|J_1|} \dots \gamma_{2n}^{|J_{2n}|} \right) \\
    &= \varsigma\, 2^n \, \mathbf{1}_{|J_1|,\dots, |J_{2n}| \text{ even}} 
\end{align*}
since $\tr(\gamma_1^{|J_1|} \dots \gamma_{2n}^{|J_{2n}|}) = \tr(I) = 2^n$ if $|J_\mu|$ is even for all $\mu \in [2n]$, whereas if any $|J_\mu|$ has odd cardinality, then $\gamma_\mu^{|J_\mu|} = \gamma_\mu$ so $\gamma_1^{|J_1|} \dots \gamma_{2n}^{|J_{2n}|} = \gamma_S$ for $S \neq \varnothing$, which is traceless.

Next, we consider the RHS of Eq.~\eqref{CliffordtoGrassmann} for $A^{(i)} = \gamma_{S^{(i)}}$ and show that it is also equal to $\varsigma\, 2^n \mathbf{1}_{|J_1|,\dots, |J_{2n}| \text{ even}}$. By Eq.~\eqref{Grassman_rep}, $\omega(\gamma_S;\theta)= \theta_S$ for any $S \subseteq [2n]$, where we define $\theta_S \coloneqq \theta_{\mu_1}\dots \theta_{\mu_{|S|}}$ for $S = \{\mu_1,\dots,\mu_{|S|}\}$ with $\mu_1 < \dots < \mu_{|S|}$. Thus, the RHS of Eq.~\eqref{CliffordtoGrassmann} is 
\begin{align*}
    \text{RHS} \coloneqq 2^n (-1)^{nm(m-1)/2} \int D\theta^{(m)}\dots D\theta^{(1)} \, \theta^{(1)}_{S^{(1)}}\dots \theta^{(m)}_{S^{(m)}} \exp\Bigg(\sum_{\substack{i,j \in [m]\\i <j}} s_{ij} \theta^{(i) T} \theta^{(j)} \Bigg).
\end{align*}
We reorder the integrals $\int D\theta^{(m)} \dots D\theta^{(1)} \equiv \int d\theta^{(m)}_{2n} \dots d\theta^{(m)}_{1} \dots d\theta^{(1)}_{2n} \dots d\theta^{(1)}_1$ to $\int d\theta^{(m)}_{2n} \dots d\theta^{(1)}_{2n} \dots d\theta^{(m)}_{1} \dots d\theta^{(1)}_{1}$. It follows from antisymmetry (Eq.~\eqref{integralantisymmetric}) that this incurs a sign of $(-1)^{nm(m-1)/2}$, i.e., 
\[ (-1)^{nm(m-1)/2}\int D\theta^{(m)} \dots D\theta^{(1)} = \int d\theta^{(m)}_{2n} \dots d\theta^{(1)}_{2n} \dots d\theta^{(m)}_{1} \dots d\theta^{(1)}_{1}. \]
Next, we reorder the $\theta^{(i)}_{\mu}$'s in $\theta^{(1)}_{S^{(1)}}\dots \theta^{(m)}_{S^{(m)}}$ such that any $\theta_1^{(i)}$'s all appear to the left, followed by any $\theta^{(i)}_2$'s, and so on. Importantly, we do not move $\theta_\mu^{(i)}$ to the left of $\theta_\mu^{(j)}$ if $i > j$, i.e., for any fixed $\mu\in [2n]$, we do not reorder the $\theta_\mu^{(i)}$ among themselves. Consequently, this reordering incurs the same sign as that of reordering the $\gamma_\mu$'s in $\gamma_{S^{(1)}}\dots \gamma_{S^{(m)}}$ we considered above, so
\[ \theta^{(1)}_{S^{(1)}} \dots \theta^{(m)}_{S^{(m)}} = \varsigma \Bigg(\prod_{j_1 \in J_1} \theta^{(j_1)}_1 \Bigg)\dots \Bigg(\prod_{j_{2n} \in J_{2n}} \theta^{(j_{2n})}_{2n} \Bigg), \]
where for each $\mu\in [2n]$, the product $\prod_{j_\mu \in J_\mu}$ is taken in increasing order. 
Finally, we write the exponential as 
\begin{align*}
    \exp\Bigg(\sum_{\substack{i,j \in [m]\\i <j}} s_{ij} \theta^{(i) T} \theta^{(j)} \Bigg) &= \exp\Bigg(\sum_{\substack{i,j \in [m]\\i <j}} s_{ij} \sum_{\mu =1}^{2n} \theta^{(i)}_\mu \theta^{(j)}_\mu \Bigg) \\
    &= \exp\Bigg(\sum_{\substack{i,j \in [m]\\i <j}} s_{ij} \theta^{(i)}_1 \theta^{(j)}_1 \Bigg)\dots \exp\Bigg(\sum_{\substack{i,j \in [m]\\i <j}} s_{ij} \theta^{(i)}_{2n} \theta^{(j)}_{2n} \Bigg)
\end{align*}
using the fact that the $\theta_\mu^{(i)}\theta_\mu^{(j)}$ mutually commute. Thus, 
\begin{align*} \text{RHS} &= 2^n \varsigma  \int d\theta^{(m)}_{2n} \dots d\theta^{(1)}_{2n} \dots d\theta^{(m)}_{1} \dots d\theta^{(1)}_{1}\Bigg(\prod_{j_1 \in J_1} \theta^{(j_1)}_1 \Bigg)\dots \Bigg(\prod_{j_{2n} \in J_{2n}} \theta^{(j_{2n})}_{2n}  \Bigg) \\
&\qquad \times \exp\Bigg(\sum_{\substack{i,j \in [m]\\i <j}} s_{ij} \theta^{(i)}_1 \theta^{(j)}_1 \Bigg)\dots \exp\Bigg(\sum_{\substack{i,j \in [m]\\i <j}} s_{ij} \theta^{(i)}_{2n} \theta^{(j)}_{2n} \Bigg) 
\end{align*}
which we can rewrite using Eq.~\eqref{Grassmannintegraldef2} (recalling that $m$ is even) and the fact that the exponentials are even elements as
\[ \text{RHS} = 2^n \varsigma \, \Omega_{J_1}\dots \Omega_{J_{2n}}, \]
where for each $\mu \in [2n]$, $\Omega_{J_\mu}$ denotes the integral
\[ \Omega_{J_\mu} \coloneqq \int d\theta_\mu^{(m)}\dots d\theta^{(1)}_{\mu}\, \left(\prod_{j_\mu \in J_\mu} \theta^{(j_\mu)}_\mu \right)\exp\Bigg(\sum_{\substack{i,j \in [m]\\i <j}} s_{ij} \theta^{(i)}_{\mu} \theta^{(j)}_{\mu} \Bigg). \]

We now show that $\Omega_{J_\mu} = \mathbf{1}_{|J_\mu| \text{ even}}$ for all $\mu \in [2n]$. Let $M$ be the antisymmetric $m \times m$ matrix with entries $s_{ij}$ above its diagonal, e.g., for $m = 4$,
\[ M = \begin{pmatrix} 0 &1 &-1 &1 \\ 
-1 & 0 &1 &-1 \\
1 &-1 & 0 &1 \\
-1 &1 &-1 &0 \end{pmatrix} \]
and the structure is similar for all $m$. Then, the exponential term can be written as $\exp(\sum_{i < j} s_{ij} \theta_\mu^{(i)}\theta_\mu^{(j)}) = \exp(\frac{1}{2}\sum_{i, j\in [m]} M_{ij} \theta_\mu^{(i)}\theta_\mu^{(j)})$, so applying Eq.~\eqref{A.15(b)} with $(\chi_1,\dots, \chi_{2N}) \to (\theta^{(1)}_{\mu},\dots, \theta^{(m)}_\mu)$ and $J \to J_\mu$ gives
\begin{align*} \Omega_{J_\mu} 
&= \epsilon(J_\mu) \Pf\big(M\big|_{\overline{J}_\mu}\big).
\end{align*}
Since $m$ is even, $\overline{J}_\mu \coloneqq [m] \setminus J_\mu$ has even cardinality if and only if $J_\mu$ has even cardinality, so $\Pf(M\big|_{\overline{J}_\mu}) = \mathbf{1}_{|J_\mu| \text{ even}}\,\Pf(M\big|_{\overline{J}_\mu})$ (recall that the Pfaffian of any matrix of odd dimension is $0$). It remains to prove that $\epsilon(J_\mu) \Pf\big(M\big|_{\overline{J}_\mu}) = 1$ when $|J_\mu|$ is even. It follows from the definition $\epsilon(J) \coloneqq (-1)^{|J|(|J|-1)/2 + \sum_{j\in J} j}$ that $\epsilon(J_\mu) = \epsilon(\overline{J}_\mu)$ for $m$ and $|J_\mu|$ even, so we show that $\epsilon(J_\mu) \Pf(M\big|_{J_\mu}) = 1$ for arbitrary $J_\mu \subseteq[m]$ of even cardinality. Let $J = \{j_1,\dots, j_k\}$, with $k = |J_\mu|$. First, it is easy to verify that the Pfaffian of the top left $k \times k$ submatrix of $M$ is~$1$, i.e., $\Pf(M\big|_{\{1,\dots,k\}}) = 1$. Now, observe that for $i \in [k]$, replacing the $i$th row and column by the row and column indexed by $j_i$ leaves the $i$th row and column unchanged if $j_i - i$ is even, and multiplies the $i$th row and column by $-1$ if $j_i - i$ is odd. Hence, since multiplying a row and column by a constant multiplies the Pfaffian by the same constant, replacing the $i$th row and column in $M\big|_{\{1,\dots, k\}}$ by row and column $j_i$ multiplies the Pfaffian by $(-1)^{j_i - i}$. Doing this for all $i \in [k]$ to change $M\big|_{\{1,\dots, k\}}$ to $M\big|_{\{j_1,\dots, j_k\}} = M\big|_{J_\mu}$ gives
$\Pf(M\big|_{J_\mu}) = (-1)^{\sum_{i =1}^{k} (j_i - i)} \Pf(M\big|_{\{1,\dots, k\}}) = (-1)^{\sum_{j \in J_\mu} j - |J_\mu|(|J_\mu| - 1)/2} = \epsilon(J_\mu)$. Therefore, we have $\Omega_{J_\mu} = \mathbf{1}_{|J_\mu| \text{ even}}$ for any $J_\mu \subseteq [m]$, so
\[ \mathrm{RHS} = 2^n \varsigma \,  \mathbf{1}_{|J_1| \text{ even}}\dots \mathbf{1}_{|J_{2n}| \text{ even}} = \text{LHS}. \]
\end{proof}

\smallskip

In the case where one of the operators is the parity operator $P = \gamma_{[2n]}$, say, $A^{(m)} = P$, Eq.~\eqref{CliffordtoGrassmann} becomes
\begin{align*}
&\tr(A^{(1)}\dots A^{(m-1)} P) \\
&\quad = (-i)^n 2^n (-1)^{nm(m-1)/2} \int D\theta^{(m-1)} \dots D\theta^{(1)} \, \omega(A^{(1)};\theta^{(1)})\dots \omega(A^{(m-1)};\theta^{(m-1)}) \exp\Bigg(\sum_{\substack{i, j \in [m-1] \\ i < j}}s_{ij} \theta^{(i)^{\mathrm{T}}} \theta^{(j)} \Bigg), \end{align*}
which involves only $m - 1$ sets of Grassmann variables, $\theta^{(1)}, \dots, \theta^{(m-1)}$. This generalises Equation~144 of Ref.~\citenum{Bravyi2017-pw} to arbitrary even $m$, and follows from that fact that $\omega(\gamma_{[2n]}; \theta^{(m)}) = \theta^{(m)}_1 \dots \theta^{(m)}_{2n}$, so we can integrate out $\theta^{(m)}$ as
\[ \int D\theta^{(m)}\, \theta^{(m)}_1\dots \theta^{(m)}_{2n} \exp\Bigg(\sum_{i < m} s_{ij}\theta^{(i)T}\theta^{(m)}\Big) = 1. \]

\subsection{Proof of correctness for Algorithm~\ref{alg:gMB}} \label{app:gMB}

\begin{proof}[Proof of correctness for Algorithm~\ref{alg:gMB}] Let $B$ be an arbitrary $K \times 2N$ matrix, and $M$ an arbitrary $2N \times 2N$ matrix. Let $\chi_1,\dots, \chi_{2N}$ be the generators of a $2^{2N}$-dimensional Grassmann algebra. 

First, note that
\[ g(B,M) \coloneqq \int D\chi \, (B\chi)_1\dots (B\chi)_K \exp\left(\frac{1}{2}\chi^{\mathrm{T}} M\chi\right) \]
vanishes whenever $K$ is odd. The exponential $\exp(\frac{1}{2}\chi^{\mathrm{T}} M\chi)$ is a linear combination of even elements, so in this case the integrand cannot contain terms proportional to $\chi_1\dots \chi_{2N}$ (recall from Eq.~\eqref{Dtheta} that these are the only terms that contribute to the integral over $D\chi$). $g(B,M)$ is also zero whenever $K > 2N$, because $(B\chi)_1\dots (B\chi)_K$ expands into a linear combination of products of $K$ $\chi_\mu$'s. Every product must vanish for $K > 2N$, since $\chi_\mu^2 = 0$ from Eq.~\eqref{Grassmannproduct}. Consequently, Algorithm~\ref{alg:gMB} returns $g(B,M) =0$ if $K$ is odd or $K > 2N$ (Step 1). 

Next, if $K = 2N$ (so $B$ is a square matrix), Algorithm~\ref{alg:gMB} returns $g(B,M) = \det(B)$ in Step 2. This is correct by Equation A.86 of Ref.~\citenum{Caracciolo}, but is also easy to see directly. When $K = 2N$, only the scalar term from the exponential contributes, so $g(B,M) = \int D\chi (B\chi)_1\dots (B\chi)_K$. Then, changing variables to $\widetilde{\chi}_\mu = \sum_{\nu=1}^{2N} B_{\mu\nu}\chi_\mu $ gives $g(B,M) = \det(B)$ by Eq.~\eqref{Dchitilde}. 

Henceforth, assume that $K$ is even and $K < 2N$. Consider an extended Grassmann algebra, generated by $\chi_1,\dots, \chi_{2N}, \eta_1,\dots, \eta_K$ (where the $\eta_\mu$ do not involve $\chi_1,\dots, \chi_{2N}$). We will use the fact that 
\begin{equation} \label{Bchi} (B\chi)_1 \dots (B\chi)_K = (-1)^{K/2} \int D\eta \, \exp(\eta^{\mathrm{T}} B \chi).  \end{equation} To see this, note that $\exp(\eta^{\mathrm{T}} B\chi) = (1 + \eta_1(B\chi)_1) \dots (1 + \eta_K (B\chi)_K)$, since the $\eta_k (B\chi)_k$ mutually commute and $(\eta_k (B\chi)_k)^j = 0$ for any $j > 1$, so $\exp(\eta_k (B\chi)_k) = 1 + \eta_k (B\chi)_k$. As a result, we have
\begin{align*}
    \int d\eta_1\dots d\eta_K \exp(\eta^{\mathrm{T}} B\chi) &= \int d\eta_1\dots  d\eta_K \, \left(1 + \eta_1 (B\chi)_1\right) \dots \left(1 + \eta_K (B\chi)_K\right) \\
    &= \left[\int d\eta_1 \left(1 + \eta_1 (B\chi)_1\right) \right]  \dots \left[\int d\eta_K \left(1 + \eta_K (B\chi)_K\right) \right] \\
    &= (B\chi)_1 \dots (B\chi)_K,
\end{align*}
applying the definition of $\int d\eta_\mu$, Eqs.~\eqref{Grassmannintegraldef1} and~\eqref{Grassmannintegraldef2}, to obtain the second and third lines. Finally observe that antisymmetry (Eq.~\eqref{integralantisymmetric}) implies that $\int D\eta \equiv \int d\eta_K\dots d\eta_1 = (-1)^{K(K-1)/2} \int d\eta_1\dots d\eta_K$, and we have $(-1)^{K(K-1)/2} = (-1)^{K/2}$ since $K^2/2$ is even for even $K$.

Thus, we have
\begin{align*}
    g(B,M) &= (-1)^{K/2} \int D\chi\, D\eta \, \exp(\eta^{\mathrm{T}} B\chi) \exp\left(\frac{1}{2}\chi^{\mathrm{T}} M\chi\right) \\
    &= (-1)^{K/2} \int D\eta \,D\chi \, \exp\left(\frac{1}{2}\chi^{\mathrm{T}} M \chi + \eta^{\mathrm{T}} B\chi\right),
\end{align*}
where the second line uses Eq.~\eqref{integralantisymmetric} along with the fact that the arguments of the exponentials commute since they are both even elements. Now, if $M$ is invertible, then we can directly apply Eq.~\eqref{A.15(a)} followed by Eq.~\eqref{GrassmannGaussian} to obtain $g(B,M) = (-1)^{K/2} \Pf(M) \Pf(BM^{-1}B^{\mathrm{T}}) = \Pf(M) \Pf(-BM^{-1}B^{\mathrm{T}})$, which is what Step~3a returns. This is Theorem~A.15(d) of Ref.~\citenum{Caracciolo}, which we stated in Eq.~\eqref{A.15(d)}. 

If $M$ is not invertible, then just as in Step 3b, let $2r$ be the rank of $M$, let $R \in \mathrm{SO}(2N)$ be any special orthogonal matrix and $M'$ any invertible (antisymmetric) $2r \times 2r$ such that $M = R^{\mathrm{T}} \begin{pmatrix} M' &0 \\
0 &0 \end{pmatrix} R$, and let $BR = \begin{pmatrix} B' & B'' \end{pmatrix}$, where $B'$ has size $K \times 2r$ and $B''$ has size $K \times 2N -2r$. Also make the change of variables $\widetilde{\chi}_{\mu} = \sum_{\nu=1}^{2N} R^{\mathrm{T}}_{\mu\nu} \chi_\nu$, and write $\widetilde{\chi} = R^{\mathrm{T}}\chi = \begin{pmatrix}\widetilde{\chi}' \\ \widetilde{\chi}'' \end{pmatrix}$, where $\widetilde{\chi}'$ is the vector consisting of the first $2r$ variables $\widetilde{\chi}_1, \dots, \widetilde{\chi}_{2r}$, and $\widetilde{\chi}''$ consists of the remaining $2N - 2r$ variables $\widetilde{\chi}_{2r+1},\dots \widetilde{\chi}_{2N}$. Since $R \in \mathrm{SO}(2n)$, $\int D\chi = \int D\widetilde{\chi} \equiv \int D\widetilde{\chi}'' \, D\widetilde{\chi}'$ by Eq.~\eqref{Dchitilde}. Hence,
\begin{align*}
    g(B,M) &= (-1)^{K/2}\int D\eta\, D\widetilde{\chi}'' \, D\widetilde{\chi}' \exp\left(\frac{1}{2}\widetilde{\chi}'^{\mathrm{T}} M' \widetilde{\chi}' + \eta^{\mathrm{T}} B' \widetilde{\chi}' + \eta^{\mathrm{T}} B''\widetilde{\chi}'' \right) \\
    &= (-1)^{K/2} \int D\eta \left[\int D\widetilde{\chi}'  \exp\left(\frac{1}{2}\widetilde{\chi}'^{\mathrm{T}} M' \widetilde{\chi}' + \eta^{\mathrm{T}} B'\widetilde{\chi}'\right) \right] \left[\int D\widetilde{\chi}'' \, \exp\left(\eta^{\mathrm{T}} B''\widetilde{\chi}''\right) \right] \\
    &= (-1)^{K/2} \int D\eta \, \Pf(M') \exp\left(\frac{1}{2} \eta^{\mathrm{T}} B' M'^{-1} B'^{\mathrm{T}} \eta \right) (-1)^{(2N-2r)/2} (-B''^{\mathrm{T}} \eta)_1 \dots (-B''^{\mathrm{T}} \eta)_{2N - 2r} \\
    &= (-1)^{K/2+ N} (-1)^r \Pf(M') \int D\eta\, (B''^{\mathrm{T}} \eta)_1 \dots (B''^{\mathrm{T}}\eta)_{2N - 2r} \exp\left(\frac{1}{2}\eta^{\mathrm{T}} B' M'^{-1} B' \eta\right) \\
    &= (-1)^{K/2 + N} \Pf(-M') \, g(B''^{\mathrm{T}}, B' M'^{-1} B'^{\mathrm{T}}).
\end{align*}
Here, we use the fact that the arguments of the exponential are even together with Eq.~\eqref{Grassmannintegraldef2} to obtain the second line. We then apply Eq.~\eqref{A.15(a)} to the $\int D\widetilde{\chi}'$ integral, and Eq.~\eqref{Bchi} (with $\chi \to \eta$ and $\eta \to \widetilde{\chi}''$) to the $\int D\widetilde{\chi}''$ integral to get the third line, noting that $\eta^{\mathrm{T}} B''\widetilde{\chi}'' = - \widetilde{\chi}''^{\mathrm{T}} B''^{\mathrm{T}} \eta$. The fourth line follows from the fact that even elements commute, and the fifth uses the property $\Pf(-M') = (-1)^{r} \Pf(M')$ for any $2r \times 2r$ antisymmetric matrix $M'$, along with the definition of $g(B,M)$ (Eq.~\eqref{gBM}). Therefore, we can evaluate $g(B,M)$ by evaluating $g(B''^{\mathrm{T}}, B' M'^{-1} B'^{\mathrm{T}})$. This is a smaller problem because   $B'M'^{-1}B'^{\mathrm{T}}$ is of size $K \times K$, and we are in the case $K < 2N$. \end{proof}

\section{$\mathcal{O}(\sqrt{n} \log n)$ variance bound for expectation values of fermionic Gaussian density operators} \label{app:variancebound}

In this appendix, we place a tighter asymptotic bound on the variance for our matchgate shadows estimates of $\tr(\varrho\rho)$ (where $\varrho$ is any fermionic Gaussian state and $\rho$ is an unknown state) than the straightforward but loose $\mathcal{O}(n^3)$ bound given in subsection~\ref{sec:variance_densityoperators} of the main text. Specifically, we argue that the expression on the RHS of Eq.~\eqref{variancedensityoperators}, which we will denote by $b(n)$, is asymptotically bounded in terms of the number of fermionic modes $n$ as $\sqrt{n} \log n$:
\[b(n) \coloneqq  \frac{1}{2^{2n}} \sum_{\substack{\ell_1,\ell_2,\ell_3 \geq 0\\\ell_1 + \ell_2 + \ell_3 \leq n}} \frac{{n\choose \ell_1,\ell_2,\ell_3,n-\ell_1-\ell_2-\ell_3}^2}{{2n\choose 2\ell_1,2\ell_2,2\ell_3, 2(n-\ell_1-\ell_2-\ell_3)}}\frac{{2n\choose 2(\ell_1+\ell_3)}}{{n\choose \ell_1 + \ell_3}}\frac{{2n\choose 2(\ell_2 + \ell_3)}}{{n\choose \ell_2 + \ell_3}} = \mathcal{O}(\sqrt{n}\log n). \]
This is consistent with Figure~\ref{fig:bnzetaplot}. 

We use the bounds Eq.~\eqref{binbound} and~\eqref{multbound} for the binomial and multinomial terms in $b(n)$. For simplicity, we can subsume all the cases in Eq.~\eqref{multbound} as 
\begin{align*}
    \frac{{n\choose k_1,k_2,k_3,k_4}^2}{{2n\choose 2k_1,2k_2,2k_3,2k_4}} &\leq \sqrt{\frac{n}{\max\{k_1,1\}\max\{k_2,1\}\max\{k_3,1\}\max\{k_4,1\}}} \\
    &\leq 4\sqrt{\frac{n}{(k_1 + 1)(k_2+1)(k_3+1)(k_4+1)}},
\end{align*}
where we bound $\max\{k,1\} \geq (k+1)/2$ in the second line. Thus,  \begin{align*}
    b(n) &\leq \frac{1}{2^{2n}} \sum_{\substack{\ell_1,\ell_2,\ell_3 \\ \ell_1 + \ell_2 + \ell_3 \leq n}} 4 \sqrt{\frac{n}{(\ell_1 + 1)(\ell_2 + 1)(\ell_3 + 1)(n-\ell_1 - \ell_2 - \ell_3 +1)}} 2^{nH_{\mathrm{b}}\left((\ell_2 + \ell_3)/n\right)} 2^{nH_{\mathrm{b}}\left((\ell_3 + \ell_1)/n\right)} \\
    &\sim \frac{4}{2^{2n}} \int_{-\frac{1}{2}}^{n +\frac{1}{2}} d\ell_1\, \int_{-\frac{1}{2}}^{n-\ell_1 + \frac{1}{2}} d\ell_2\, \int_{-\frac{1}{2}}^{n-\ell_1 -\ell_2 + \frac{1}{2}}d\ell_3\, \sqrt{\frac{n}{(\ell_1 + 1)(\ell_2 + 1)(\ell_3 + 1)(n-\ell_1 - \ell_2 - \ell_3 +1)}} \\
    &\qquad \times 2^{n\left(H_{\mathrm{b}}\left((\ell_2 + \ell_3)/n\right) + (\ell_2 + \ell_3)/n\right)} \\
    &= \frac{4}{2^{2n}} n\sqrt{n} \iiint d\ell_1'\, d\ell_2'\,d\ell_3'\, \mathbf{1}_{\ell_1',\ell_2',\ell_3' \geq -1/(2n), \,\ell_1'+\ell_2'+\ell_3' \leq 1 + 1/(2n)} \\
    &\qquad \times \frac{1}{\sqrt{(\ell_1' + \frac{1}{n})(\ell_2' + \frac{1}{n})(\ell_3' + \frac{1}{n})(n-\ell_1'-\ell_2'-\ell_3' + \frac{1}{n})}} \exp\left[(\ln 2) n \left(H_{\mathrm{b}}(\ell_2' + \ell_3') + H_{\mathrm{b}}(\ell_3' + \ell_1') \right) \right],
\end{align*}
using the midpoint rule in the second line, and changing variables to $\ell_i' = \ell_i/n$ for $i \in \{1,2,3\}$ in the third.

Next, we use Laplace's method to evaluate the integral over two of the dimensions. Note that the function $(\ln 2)(H_{\mathrm{b}}(\ell_2' + \ell_3') + H_{\mathrm{b}}(\ell_3' + \ell_1'))$ in the argument of the exponential is maximised along the line $L$ given by $\ell_1' = \ell_2'$ and $\ell_3' = 1/2 - \ell_1'$. If we fix $\ell_3'$ and consider the integral over $\ell_1'$ and $\ell_2'$, this line does not intersect (the interior of) the region of integration over $\ell_1'$ and $\ell_2'$ whenever $\ell_3' \geq 1/2 + 1/(2n)$, so Laplace's method does not apply in this case. Hence, we change variables, letting
\[ x = \frac{1}{2}(\ell_1' + \ell_2'), \qquad y = \frac{1}{2}(-\ell_1' + \ell_2'), \]
so that for every relevant value of $x$, the line $L$ intersects the region of integration over $y$ and $\ell_3'$ at exactly one point in the interior of the region. The determinant of the Jacobian for this change of variables is $2$, and we have
\begin{align*}
    b(n) &\lesssim \frac{8}{2^{2n}} n\sqrt{n} \iiint dx\, dy\, d\ell_3' \, \mathbf{1}_{x-y,x+y, \ell_3' \geq -1/(2n), \, 2x + \ell_3' \leq 1+ 1/(2n)} \\
    &\qquad \times \frac{1}{\sqrt{(x-y+\frac{1}{n})(x + y + \frac{1}{n})(\ell_3' + \frac{1}{n})(n-2x+\ell_3' + \frac{1}{n})}}\exp\left[n(\ln 2) \left(H_{\mathrm{b}}(x+ y + \ell_3') + H_{\mathrm{b}}(x - y + \ell_3') \right) \right].
\end{align*}
We now apply Laplace's method to evaluate the integral over $y$ and $\ell_3'$. Denote the function in the exponential by $f_x(y,\ell_3') = (\ln 2) (H_{\mathrm{b}}(x + y + \ell_3') + H_{\mathrm{b}}(x - y + \ell_3'))$. For any fixed $x \in (-1/(2n), 1/2 + 1/(4n))$, $f_x$ is maximised at the point $y = 0$, $\ell_3' = 1/2 - x$ (with maximum value $2\ln2$), and this point is always within the region of integration. The Hessian $\mathbf{H}_{f_x}$ of $f_x$ at this point is $\text{diag}(-8,-8)$ (which is clearly negative definite), so $\sqrt{-\det(\mathbf{H}_{f_x}(0, 1/2-x))} = 8$.
Laplace's method therefore gives
\begin{align*}
    b(n) &\lesssim \frac{8}{2^{2n}} n\sqrt{n} \int_{-\frac{1}{2n}}^{\frac{1}{2} + \frac{1}{4n}} dx\, \frac{2\pi}{n} \frac{1}{8} \frac{1}{\sqrt{(x + \frac{1}{n})^2 (\frac{1}{2}-x + \frac{1}{n})^2}} e^{n (2\ln 2)} \\
    &= 2\pi \sqrt{n} \int_{-\frac{1}{2n}}^{\frac{1}{2} + \frac{1}{4n}} dx\, \frac{1}{(x + \frac{1}{n})(\frac{1}{2} + \frac{1}{n} - x)} \\
    &= 2\pi \sqrt{n} \frac{2n}{\frac{n}{2} + 4} \int_{-\frac{1}{2}}^{\frac{n}{2} + \frac{1}{4}} dx'\, \left(\frac{1}{x' + 1} + \frac{1}{\frac{n}{2}+1 - x'} \right) \\
    &= 4\pi \sqrt{n} \frac{n}{n+4} \left[\ln(x' + 1) - \ln\left(\frac{n}{2} + 1 - x'\right) \right]\Bigg|_{-1/2}^{n/2 + 1/4} \\
    &= \mathcal{O}(\sqrt{n} \ln n),
\end{align*}
where we changed variables to $x' = x/n$ in the third line.

\bibliography{ryan, refs}

\end{document}